\documentclass[12pt]{article}
\usepackage[utf8]{inputenc}
\usepackage[T1]{fontenc}
\usepackage[a4paper,margin=2.5cm]{geometry}
\usepackage{amsmath}
\usepackage{amssymb}
\usepackage{amsthm}
\usepackage{enumerate}
\usepackage{color}
\usepackage{tikz}
\usetikzlibrary{calc}
\usepackage{calc}
\usepackage{xintexpr}
\usepackage{tabu}
\usepackage{longtable}
\usepackage{longfigure}
\usepackage{graphicx}
\usepackage{authblk}
\usepackage{hyperref}

\newcommand{\N}{\mathbb{N}}
\newcommand{\Z}{\mathbb{Z}}

\newcommand{\Zn}[1]{\Z/{#1}\Z}
\newcommand{\Zd}{\Z_{2}}
\newcommand{\ST}[1]{\mathcal{ST}(#1)}

\newcommand{\RST}[1]{\mathcal{RST}(#1)}
\newcommand{\HST}[1]{\mathcal{HST}(#1)}
\newcommand{\DST}[1]{\mathcal{DST}(#1)}
\newcommand{\DSTz}[1]{\mathcal{DST}_0(#1)}
\newcommand{\PT}[1]{\mathcal{PT}(#1)}
\newcommand{\RPT}[1]{\mathcal{RPT}(#1)}
\newcommand{\HPT}[1]{\mathcal{HPT}(#1)}
\newcommand{\DPT}[1]{\mathcal{DPT}(#1)}
\newcommand{\ESG}[1]{\mathcal{ESG}(#1)}
\newcommand{\OSG}[1]{\mathcal{OSG}(#1)}
\newcommand{\PRSG}[1]{\mathcal{PRSG}(#1)}
\newcommand{\Strig}[1]{\nabla{#1}}
\newcommand{\iStrig}[1]{\nabla_{\!#1}} % Strig whose index is with \!
\newcommand{\Ptrig}[2]{\Delta\left(#1,#2\right)}
\newcommand{\Smat}[1]{\mathrm{M}\left(#1\right)}
\newcommand{\Sgraph}[1]{\mathrm{G}\left(#1\right)}

\newcommand{\SG}[1]{\mathcal{SG}(#1)}
\newcommand{\BS}[3]{\mathrm{S}^{(#1)}_{#2,#3}} %Binomial sequence of length #1 from the element of row #2 in the column #3
\newcommand{\ES}[2]{\mathrm{E}^{\left(#1\right)}_{#2}} %Canonical basis sequence of length #1 from the element #2
\newcommand{\He}{\mathrm{H}}
\newcommand{\binomd}[2]{{\binom{#1}{#2}}_{\!2}}
\newcommand{\deltamod}[3]{\delta_{#1,#2\,(\mathrm{mod}\: #3)}}

\newcommand{\tpmod}[1]{\left(\mathrm{mod}\ #1\right)}

\newcommand*\MyBinom[2]
	{\xinttheiiexpr binomial(#1,#2)/:2\relax}
	
\newcounter{c}

\newcommand{\figLine}[9]{%{seq}{xpos}{ypos}{no_of_row}{col_min}{col_max}{Colorline}{FillColor0}{FillColor1}
	\foreach[count=\ji from 0] \j in {#5,...,#6}{
		\setcounter{c}{0}
		\foreach[count=\xi] \x in #1{
			\ifnum \x>0
				\ifnum \MyBinom{#4-1}{\j-\xi}>0
					\addtocounter{c}{1}
				\fi
			\fi
		}
		\pgfmathparse{int(mod(\value{c},2))}\let\n\pgfmathresult
		\ifnum \n>0
			\draw[#7,fill=#9] (#2+2*\ji,#3) circle (1);		
		\else
			\draw[#7,fill=#8] (#2+2*\ji,#3) circle (1);
		\fi
		\draw[#7] (#2+2*\ji,#3) node {\n};
	}
}

\newcommand{\figST}[5]{%{seq}{Colorline}{FillColor0}{FillColor1}{TrigFillColor}
	\setcounter{c}{0}
	\foreach \x in #1{
		\addtocounter{c}{1}
	}
	\pgfmathparse{\value{c}}\let\size\pgfmathresult
	\pgfmathparse{sqrt(3)}\let\sq\pgfmathresult
	\pgfmathparse{(\size-1)*\sq}\let\v\pgfmathresult
	\pgfmathparse{2*(\size-1)}\let\w\pgfmathresult
	\draw[#2,fill=#5] (-\sq,1) -- (\w+\sq,1) -- (\size-1,-2-\v) -- (-\sq,1);
	\foreach \k in {1,...,\size}{
		\pgfmathparse{sqrt(3)*(\k-1)}\let\l\pgfmathresult
		\figLine{#1}{\k-1}{-\l}{\k}{\k}{\size}{#2}{#3}{#4}
	}
}

\newcommand{\figSTcolsq}[5]{%{seq}{Colorline}{FillColor0}{FillColor1}{SqFillColor}
	\setcounter{c}{0}
	\foreach \x in #1{
		\addtocounter{c}{1}
	}
	\pgfmathparse{\value{c}}\let\size\pgfmathresult
	\pgfmathparse{sqrt(3)}\let\sq\pgfmathresult
	\pgfmathparse{(\size-1)*\sq}\let\v\pgfmathresult
	\pgfmathparse{2*(\size-1)}\let\w\pgfmathresult
	\draw[white,fill=#5] (-\sq-1,2) -- (\w+\sq+1,2) -- (\w+\sq+1,-3-\v) -- (-\sq-1,-3-\v) -- (-\sq-1,2);
	\draw[#2] (-\sq,1) -- (\w+\sq,1) -- (\size-1,-2-\v) -- (-\sq,1);
	\foreach \k in {1,...,\size}{
		\pgfmathparse{sqrt(3)*(\k-1)}\let\l\pgfmathresult
		\figLine{#1}{\k-1}{-\l}{\k}{\k}{\size}{#2}{#3}{#4}
	}
}

\newcommand{\figPT}[5]{%{seq}{Colorline}{FillColor0}{FillColor1}{TrigFillColor}
	\setcounter{c}{0}
	\foreach \a in #1{
		\addtocounter{c}{1}
	}
	\pgfmathparse{int((\value{c}+1)*0.5)}\let\n\pgfmathresult
	\pgfmathparse{sqrt(3)}\let\sq\pgfmathresult
	\pgfmathparse{(\n-1)*sqrt(3)}\let\v\pgfmathresult
	\draw[#2,fill=#5] (0,2) -- (-\n+1-\sq,-\v-1) -- (\n-1+\sq,-\v-1) -- (0,2);
	\foreach \a in {1,...,\n} {
		\pgfmathparse{sqrt(3)*(\a-1)}\let\b\pgfmathresult
		\pgfmathparse{\n+\a-1}\let\m\pgfmathresult
		\figLine{#1}{-\a+1}{-\b}{\a}{\n}{\m}{#2}{#3}{#4}
	}
}

\newcommand{\figSG}[3]{%{seq}{colorline}{fill color of nodes}
	\setcounter{c}{0}
	\foreach \x in #1 {\addtocounter{c}{1}}
	\pgfmathparse{\value{c}+1}\let\n\pgfmathresult
	\pgfmathparse{\value{c}}\let\nn\pgfmathresult	
	\pgfmathparse{0.5*\n}\let\size\pgfmathresult
	\def \radius {\size cm}
	\def \margin {8} % margin in angles, depends on the radius
	\foreach \s in {1,...,\n}
	{
		\node[#2,draw, circle, fill = #3] (\s) at ({360/\n * (\s - 1)}:\radius) {$\s$};
	}
	\foreach \i in {1,...,\nn}{
		\foreach \j in {\i,...,\nn}{
			\setcounter{c}{0}
			\foreach[count=\xi] \x in #1{
				\ifnum \x>0
					\ifnum \MyBinom{\i-1}{\j-\xi}>0
						\addtocounter{c}{1}
					\fi
				\fi
			}
			\pgfmathparse{int(mod(\value{c},2))}\let\m\pgfmathresult
			\ifnum \m>0
				\pgfmathparse{int(\i)}\let\ii\pgfmathresult
				\pgfmathparse{int(\j+1)}\let\jj\pgfmathresult
				\draw[#2] (\ii) -- (\jj);
			\fi
		}
	}
}

\newcommand{\figSM}[5]{%{seq}{Colorline}{FillColor0}{FillColor1}{diag fill color}
	\tikzset{
		ncbar angle/.initial=90,
		ncbar/.style={
			to path=(\tikztostart)
			-- ($(\tikztostart)!#1!\pgfkeysvalueof{/tikz/ncbar angle}:(\tikztotarget)$)
			-- ($(\tikztotarget)!($(\tikztostart)!#1!\pgfkeysvalueof{/tikz/ncbar angle}:(\tikztotarget)$)!\pgfkeysvalueof{/tikz/ncbar angle}:(\tikztostart)$)
			-- (\tikztotarget)
		},
		ncbar/.default=0.5cm,
	}

	\tikzset{square left brace/.style={ncbar=0.5cm}}
	\tikzset{square right brace/.style={ncbar=-0.5cm}}

	\tikzset{round left paren/.style={ncbar=0.5cm,out=100,in=-100}}
	\tikzset{round right paren/.style={ncbar=0.5cm,out=80,in=-80}}

	\setcounter{c}{0}
	\foreach \x in #1 {\addtocounter{c}{1}}
	\pgfmathparse{\value{c}}\let\n\pgfmathresult
	\foreach \i in {1,...,\n} {
		\pgfmathparse{2*(\i-1)}\let\y\pgfmathresult
		\draw[#2,fill=#5] (\y,-\y) circle (1);
		\draw[#2] (\y,-\y) node {$0$};
		\foreach \j in {\i,...,\n}{
			\setcounter{c}{0}
			\foreach[count=\zi] \z in #1{
				\ifnum \z>0
					\ifnum \MyBinom{\i-1}{\j-\zi}>0
						\addtocounter{c}{1}
					\fi
				\fi
			}
			\pgfmathparse{int(mod(\value{c},2))}\let\num\pgfmathresult
			\pgfmathparse{2*\j}\let\x\pgfmathresult
			\ifnum \num>0
				\draw[#2,fill=#4] (\x,-\y) circle (1);
				\draw[#2,fill=#4] (\y,-\x) circle (1);
			\else
				\draw[#2,fill=#3] (\x,-\y) circle (1);
				\draw[#2,fill=#3] (\y,-\x) circle (1);
			\fi
			\draw[#2] (\x,-\y) node {$\num$};
			\draw[#2] (\y,-\x) node {$\num$};
		}
	}
	\pgfmathparse{2*\n}\let\y\pgfmathresult
	\draw[#2,fill=#5] (\y,-\y) circle (1);
	\draw[#2] (\y,-\y) node {$0$};
	\draw [#2, thick] (-1,-\y-1) to [round left paren ] (-1,1);
	\draw [#2, thick] (\y+1,-\y-1) to [round right paren ] (\y+1,1);
}

\newcommand{\figLineGS}[9]{%{seq}{xpos}{ypos}{no_of_row}{col_min}{col_max}{Colorline}{FillColor0}{FillColor1}
	\foreach[count=\ji from 0] \j in {#5,...,#6}{
		\setcounter{c}{0}
		\foreach[count=\xi] \x in #1{
			\ifnum \x>0
				\ifnum \MyBinom{#4-1}{\j-\xi}>0
					\addtocounter{c}{1}
				\fi
			\fi
		}
		\pgfmathparse{int(mod(\value{c},2))}\let\n\pgfmathresult
		\ifnum \n>0
			\draw[fill=#9] (#2+2*\ji,#3) circle (1);		
		\else
			\draw[fill=#8] (#2+2*\ji,#3) circle (1);
		\fi
		\draw[#7] (#2+2*\ji,#3) node {\n};
	}
}

\newcommand{\figSTGS}[5]{%{seq}{Colorline}{FillColor0}{FillColor1}{TrigFillColor}
	\setcounter{c}{0}
	\foreach \x in #1{
		\addtocounter{c}{1}
	}
	\pgfmathparse{\value{c}}\let\size\pgfmathresult
	\pgfmathparse{sqrt(3)}\let\sq\pgfmathresult
	\pgfmathparse{(\size-1)*\sq}\let\v\pgfmathresult
	\pgfmathparse{2*(\size-1)}\let\w\pgfmathresult
	\draw[fill=#5] (-\sq,1) -- (\w+\sq,1) -- (\size-1,-2-\v) -- (-\sq,1);
	\foreach \k in {1,...,\size}{
		\pgfmathparse{sqrt(3)*(\k-1)}\let\l\pgfmathresult
		\figLineGS{#1}{\k-1}{-\l}{\k}{\k}{\size}{#2}{#3}{#4}
	}
}

\newcommand{\PascalMatrix}[4]{%{size}{Colorline}{FillColor0}{FillColor1}
	\pgfmathparse{#1-1}\let\a\pgfmathresult
	\foreach \i in {-#1,...,\a}{
		\foreach \j in {-#1,...,-1}{
			\draw[#2,fill=#3] (2*\j,-2*\i) circle (1);
			\draw[#2] (2*\j,-2*\i) node {$0$};
		}
	}
	\foreach \i in {0,...,\a}{
		\foreach \j in {0,...,\a}{
			\pgfmathparse{\MyBinom{\i}{\j}}\let\n\pgfmathresult
			\ifnum \n>0
				\draw[#2,fill=#4] (2*\j,-2*\i) circle (1);
			\else
				\draw[#2,fill=#3] (2*\j,-2*\i) circle (1);			
			\fi
			\draw[#2] (2*\j,-2*\i) node {$\n$};
		}
	}
	\foreach \i in {-#1,...,-1}{
		\foreach \j in {0,...,\a}{
			\pgfmathparse{\MyBinom{\j-\i-1}{\j}}\let\n\pgfmathresult
			\ifnum \n>0
				\draw[#2,fill=#4] (2*\j,-2*\i) circle (1);
			\else
				\draw[#2,fill=#3] (2*\j,-2*\i) circle (1);			
			\fi
			\draw[#2] (2*\j,-2*\i) node {$\n$};
		}
	}
	\foreach \i in {-#1,...,\a}{
		\draw[#2] (-2*#1-2.5,-2*\i) node {$\cdots$};
		\draw[#2] (2*#1+0.5,-2*\i) node {$\cdots$};
	}
	\foreach \j in {-#1,...,\a}{
		\draw[#2] (2*\j,2*#1+3) node {$\vdots$};
		\draw[#2] (2*\j,-2*#1) node {$\vdots$};
	}
}

\newcommand{\figT}[6]{%{size}{seq of seq}{ColorLine}{FillColor0}{FillColor1}{TrigFill}
	\pgfmathparse{sqrt(3)}\let\sq\pgfmathresult
	\pgfmathparse{(#1-1)*\sq}\let\v\pgfmathresult
	\pgfmathparse{2*(#1-1)}\let\w\pgfmathresult
	\draw[#3,fill=#6] (-\sq,1) -- (\w+\sq,1) -- (#1-1,-2-\v) -- (-\sq,1);
	\foreach[count=\aj from 0] \a in #2{
		\pgfmathparse{#1-\aj+1}\let\jj\pgfmathresult
		\pgfmathparse{\aj*\sq}\let\rj\pgfmathresult
		\foreach[count=\bi from 0] \b in \a{
			\pgfmathparse{2*\bi+\aj}\let\x\pgfmathresult
			\ifnum \b>0
				\draw[#3,fill=#5] (\x,-\rj) circle (1);
			\else
				\draw[#3,fill=#4] (\x,-\rj) circle (1);
			\fi
			\draw[#3] (\x,-\rj) node {$\b$};
		}
	}
}

\newcommand{\figTC}[9]{%{size}{seq of seq}{ColorLine}{FillColor0}{FillColor1}{FillColor2}{FillColor3}{FillColor4}{FillColor5}
	\pgfmathparse{sqrt(3)}\let\sq\pgfmathresult
	\pgfmathparse{(#1-1)*\sq}\let\v\pgfmathresult
	\pgfmathparse{2*(#1-1)}\let\w\pgfmathresult
	\draw[#3] (-\sq,1) -- (\w+\sq,1) -- (#1-1,-2-\v) -- (-\sq,1);
	\foreach[count=\aj from 0] \a in #2{
		\pgfmathparse{#1-\aj+1}\let\jj\pgfmathresult
		\pgfmathparse{\aj*\sq}\let\rj\pgfmathresult
		\foreach[count=\bi from 0] \b in \a{
			\pgfmathparse{2*\bi+\aj}\let\x\pgfmathresult
			\ifnum \b=0
				\draw[#3,fill=#4] (\x,-\rj) circle (1);
			\else
				\ifnum \b=1
					\draw[#3,fill=#5] (\x,-\rj) circle (1);
				\else
					\ifnum \b=2
						\draw[#3,fill=#6] (\x,-\rj) circle (1);
					\else
						\ifnum \b=3
							\draw[#3,fill=#7] (\x,-\rj) circle (1);
						\else
							\ifnum \b=4
								\draw[#3,fill=#8] (\x,-\rj) circle (1);
							\else
								\draw[#3,fill=#9] (\x,-\rj) circle (1);
							\fi
						\fi
					\fi
				\fi
			\fi
			\draw[#3] (\x,-\rj) node {$\b$};
		}
	}
}

\def\restriction#1#2{\mathchoice
              {\setbox1\hbox{${\displaystyle #1}_{\scriptstyle #2}$}
              \restrictionaux{#1}{#2}}
              {\setbox1\hbox{${\textstyle #1}_{\scriptstyle #2}$}
              \restrictionaux{#1}{#2}}
              {\setbox1\hbox{${\scriptstyle #1}_{\scriptscriptstyle #2}$}
              \restrictionaux{#1}{#2}}
              {\setbox1\hbox{${\scriptscriptstyle #1}_{\scriptscriptstyle #2}$}
              \restrictionaux{#1}{#2}}}
\def\restrictionaux#1#2{{#1\,\smash{\vrule height .8\ht1 depth .85\dp1}}_{\,#2}} 

\definecolor{col0}{rgb}{1.0, 1.0, 1.0}
\definecolor{col1}{gray}{0.66}
\definecolor{col2}{rgb}{1.0, 0.5, 0.5}
\definecolor{col3}{rgb}{0.0, 0.72, 0.92}
\renewcommand{\ge}{\geqslant}
\renewcommand{\geq}{\geqslant}
\renewcommand{\le}{\leqslant}

\theoremstyle{plain}% default
\newtheorem{thm}{Theorem}[section]
\newtheorem{lem}[thm]{Lemma}
\newtheorem{prop}[thm]{Proposition}
\newtheorem{cor}[thm]{Corollary}

\newtheorem*{quest}{Question}
\theoremstyle{definition}

\theoremstyle{remark}
\newtheorem*{rem}{Remark}

\newtheorem{case}{Case}
\providecommand{\msc}[1]{\textbf{\textit{MSC2010: }} #1}
\providecommand{\keywords}[1]{\textbf{\textit{Keywords:}} #1}
\title{Symmetric binary Steinhaus triangles and parity-regular Steinhaus graphs}
\author{Jonathan CHAPPELON\footnote{E-mail address: \href{mailto:jonathan.chappelon@umontpellier.fr}{jonathan.chappelon@umontpellier.fr}}}
\affil{IMAG, Université de Montpellier, CNRS, Montpellier, France}
\date{October 29, 2021}
\begin{document}
%%%%%%%%%%%%%%%%%%%%%%%
%%%%%%%%%%%%%%%%%%%%%%%
\maketitle
%%%%%%%%%%%%%%%%
\begin{abstract}
A binary Steinhaus triangle is a triangle of zeroes and ones that points down and with the same local rule as the Pascal triangle modulo 2. A binary Steinhaus triangle is said to be rotationally symmetric, horizontally symmetric or dihedrally symmetric if it is invariant under the 120 degrees rotation, the horizontal reflection or both, respectively. The first part of this paper is devoted to the study of linear subspaces of rotationally symmetric, horizontally symmetric and dihedrally symmetric binary Steinhaus triangles. We obtain simple explicit bases for each of them by using elementary properties of the binomial coefficients. A Steinhaus graph is a simple graph with an adjacency matrix whose upper-triangular part is a binary Steinhaus triangle. A Steinhaus graph is said to be even or odd if all its vertex degrees are even or odd, respectively. One of the main results of this paper is the existence of an isomorphism between the linear subspace of even Steinhaus graphs and a certain linear subspace of dihedrally symmetric binary Steinhaus triangles. This permits us to give, in the second part of this paper, an explicit basis for even Steinhaus graphs and for the vector space of parity-regular Steinhaus graphs; i.e., the linear subspace of Steinhaus graphs that are even or odd. Finally, in the last part of this paper, we consider the generalized Pascal triangles, that are triangles of zeroes and ones, that point up now, and always with the same local rule as the Pascal triangle modulo 2. New simple bases for each linear subspace of symmetric generalized Pascal triangles are deduced from the results of the first part.
\end{abstract}
\vspace{1ex}
\noindent\msc{05B30, 05A15, 05A10, 11A99, 11B75, 11B50, 11B65, 11B85, 05C50, 05C30.} \\[2ex]
\noindent\keywords{Steinhaus triangles, symmetric triangles, symmetric sequences, rotational symmetry, dihedral symmetry, binomial coefficients, Steinhaus graphs, parity-regular graphs, even graphs, generalized Pascal triangles.}
%%%%%%%%%%%%%%%%%%%%%%%
\newpage
\tableofcontents
%%%%%%%%%%%%%%%%
%%%%%%%%%%%%%%%%
\section{Introduction}\label{sec:1}
%%%%%%%%%%%%%%%%
%%%%%%%%%%%%%%%%
A {\em binary Steinhaus triangle (or Steinhaus triangle for short) of size $n$} is a triangle $\left(a_{i,j}\right)_{1\le i\le j\le n}$ of $0$'s and $1$'s with the same local rule as Pascal's triangle modulo $2$; i.e.,
\begin{gather}\tag{LR}\label{eq13}
a_{i,j} \equiv a_{i-1,j-1}+a_{i-1,j} \pmod{2},
\end{gather}
for all integers $i,j$ such that $2\le i\le j\le n$. Note that $(0)$ and $(1)$ are the Steinhaus triangles of size $1$ and $\emptyset$ is the Steinhaus triangle of size $0$. An example of a Steinhaus triangle of size $7$ is depicted in Figure~\ref{fig1}.

\begin{figure}[htbp]
\centerline{
\begin{tikzpicture}[scale=0.25]
\figST{{0,0,1,0,1,0,0}}{black}{col0}{col1}{white}
\end{tikzpicture}}
\caption{A Steinhaus triangle of size $7$}\label{fig1}
\end{figure}

It is clear that a Steinhaus triangle $(a_{i,j})_{1\le i\le j\le n}$ is completely determined by its first row $(a_{1,j})_{1\le j\le n}$. Indeed, by induction on $i$ and using \eqref{eq13}, we obtain
\begin{equation}\label{eq14}
a_{i,j} \equiv \sum_{k=0}^{i-1}\binom{i-1}{k}a_{1,j-k} \pmod{2},
\end{equation}
for all integers $i,j$ such that $1\le i\le j\le n$, where the binomial coefficient $\binom{a}{b}$ is the coefficient of the monomial $X^b$ in the expansion of $(1+X)^a$, for all non-negative integers $a$ and $b$ such that $b\le a$. In the sequel, the Steinhaus triangle whose first row is the sequence $S$ is denoted by $\Strig{S}$. The Steinhaus triangle in Figure~\ref{fig1} is then $\Strig{(0010100)}$.

Since the set $\ST{n}$ of binary Steinhaus triangles of size $n$ is closed under addition modulo $2$, it follows that $\ST{n}$ is a vector space over $\Zn{2}$. Moreover, since a Steinhaus triangle is uniquely determined by its first row, the dimension of $\ST{n}$ is $n$.

This kind of binary triangles was introduced by Hugo Steinhaus in his problem book \cite{Steinhaus:1958aa,Steinhaus:1964aa}, where he posed, as an unsolved problem, the following

\begin{quest}
\hspace{-0.2pt}Does there exist, for every non-negative integer $n$ such that $n\equiv 0$ or $3\tpmod{4}$, a Steinhaus triangle of size $n$ with the exactly the same number of $0$'s and $1$'s?
\end{quest}

The Steinhaus triangle $\Strig{(0010100)}$ depicted in Figure~\ref{fig1} solves this problem for $n=7$, since it contains $14$ zeroes and $14$ ones. Note that since a triangle of size $n$ contains $\binom{n+1}{2}$ elements, the condition $n\equiv 0$ or $3\tpmod{4}$ is a necessary and sufficient condition for having a triangle of size $n$ with an even number of terms. The Steinhaus problem was solved for the first time by Heiko Harborth in 1972 \cite{Harborth:1972aa}. Since then, many solutions of this problem have appeared in the literature \cite{Eliahou:2004aa,Eliahou:2005aa,Eliahou:2007aa,Chappelon:2017aa}. Generalizations of this problem in $\Zn{m}$, for all $m\ge2$, can be found in \cite{Molluzzo:1978aa,Chappelon:2008ab,Chappelon:2011aa,Chappelon:2012aa} and in higher dimensions in \cite{Chappelon:2015aa}.

The local rule \eqref{eq13} can also be written as
$$
a_{i-1,j-1} \equiv a_{i-1,j}+a_{i,j} \pmod{2}\quad\text{or}\quad a_{i,j} \equiv a_{i-1,j}+a_{i-1,j-1} \pmod{2},
$$
for all integers $i,j$ such that $2\le i\le j\le n$. This is the reason why the $120$ degrees rotation and the horizontal reflection of a Steinhaus triangle are also Steinhaus triangles, of the same size. We denote by $r$ and $h$ the corresponding automorphisms of $\ST{n}$; namely
$$
\begin{array}{llll}
r : \ST{n} \longrightarrow \ST{n} & \text{by} & r\left((a_{i,j})_{1\le i\le j\le n}\right) = (a_{j-i+1,n-i+1})_{1\le i\le j\le n} & \text{and} \\[1.5ex]
h : \ST{n} \longrightarrow \ST{n} & \text{by} & h\left((a_{i,j})_{1\le i\le j\le n}\right) = (a_{i,n-j+i})_{1\le i\le j\le n}. \\
\end{array}
$$
These automorphisms satisfy the following identities
$$
r^3 = h^2 = \left(hr\right)^2 = id_{\ST{n}},
$$
where $id_{\ST{n}}$ is the identity map on $\ST{n}$. Therefore, the subgroup $\left\langle r , h\right\rangle$ generated by $r$ and $h$, of the automorphism group of $\ST{n}$, is isomorphic to the dihedral group $D_3$. This induces a faithful representation of $D_3$ on $\ST{n}$. In the sequel, the automorphism subgroup $\left\langle r , h\right\rangle$ is simply denoted by $D_3$. For instance, for $S=(11001)$ and for all $g\in D_3$, the Steinhaus triangles $g\left(\Strig{S}\right)$ are depicted in Figure~\ref{fig2}.

\begin{figure}[htbp]
\centerline{\includegraphics[width=\textwidth]{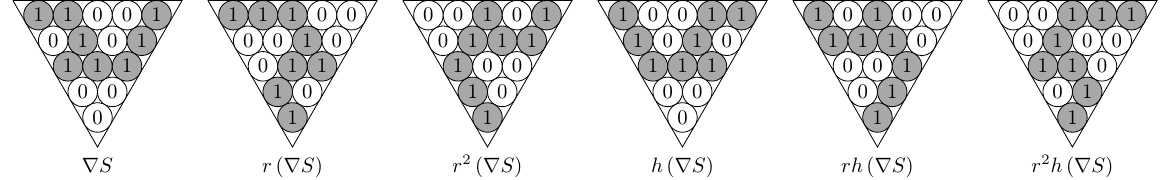}}
\caption{Action of $D_3$ on $\Strig{(11001)}$}\label{fig2}
\end{figure}

For any subgroup $G$ of $D_3$, we consider the linear subspace of invariant triangles of $\ST{n}$ under $G$, that is,
$$
{\ST{n}}^{G} = \left\{ \Strig{}\in\ST{n}\ \middle|\ \forall g\in G, g\left(\Strig{}\right)=\Strig{}\right\}.
$$
It is well known that there are exactly $6$ subgroups of $D_3$, these are $\left\{id_{\ST{n}}\right\}$, $\left\langle h\right\rangle$, $\left\langle rh\right\rangle$, $\left\langle r^2h\right\rangle$, $\left\langle r\right\rangle$ and $D_3$. Obviously, we have ${\ST{n}}^{G}=\ST{n}$ for the trivial subgroup $G=\left\{id_{\ST{n}}\right\}$. Moreover, by the linear maps
$$
\begin{array}{llll}
f : {\ST{n}}^{\left\langle h\right\rangle} \longrightarrow {\ST{n}}^{\left\langle rh\right\rangle} & \text{by} & f\left(\Strig{}\right)=r^2\left(\Strig{}\right) & \text{and} \\[1.5ex]
g : {\ST{n}}^{\left\langle h\right\rangle} \longrightarrow {\ST{n}}^{\left\langle r^2h\right\rangle} & \text{by} & g\left(\Strig{}\right)=r\left(\Strig{}\right),\\
\end{array}
$$
it is clear that the three linear subspaces ${\ST{n}}^{\left\langle h\right\rangle}$, ${\ST{n}}^{\left\langle rh\right\rangle}$ and ${\ST{n}}^{\left\langle r^2h\right\rangle}$ are isomorphic to each other. Therefore, we only consider the linear subspaces ${\ST{n}}^{\left\langle h\right\rangle}$, ${\ST{n}}^{\left\langle r\right\rangle}$ and ${\ST{n}}^{D_3}$, that will be denoted by $\HST{n}$, $\RST{n}$ and $\DST{n}$, respectively, in the sequel of this paper. Obviously, these vector spaces simply correspond to $\ker\left(h-id_{\ST{n}}\right)$, $\ker\left(r-id_{\ST{n}}\right)$ and $\ker\left(r-id_{\ST{n}}\right)\cap\ker\left(h-id_{\ST{n}}\right)$, respectively.

A Steinhaus triangle $\Strig{}$ of $\HST{n}$, $\RST{n}$ or $\DST{n}$ is said to be {\em horizontally symmetric}, {\em rotationally symmetric} or {\em dihedrally symmetric}, respectively, and satisfies $h\left(\Strig{}\right)=\Strig{}$, $r\left(\Strig{}\right)=\Strig{}$ or $r\left(\Strig{}\right)=h\left(\Strig{}\right)=\Strig{}$, respectively. Examples of such symmetric Steinhaus triangles appear in Figure~\ref{fig3}.

\begin{figure}[htbp]
\centerline{
\begin{tabular}{c@{\quad\quad\quad}c@{\quad\quad\quad}c}
\begin{tikzpicture}[scale=0.25]
\figST{{1,1,0,0,1,1}}{black}{col0}{col1}{white}
\end{tikzpicture}
&
\begin{tikzpicture}[scale=0.25]
\figST{{1,0,0,1,1,1}}{black}{col0}{col1}{white}
\end{tikzpicture}
&
\begin{tikzpicture}[scale=0.25]
\figST{{0,1,1,1,1,0}}{black}{col0}{col1}{white}
\end{tikzpicture}
\end{tabular}}
\caption{Triangles of $\HST{6}$, $\RST{6}$ and $\DST{6}$.}\label{fig3}
\end{figure}

In \cite{Barbe:2000aa}, it was proved that, for all non-negative integers $n$, we have
\begin{itemize}
\item
$\dim\HST{n} = \left\lceil\frac{n}{2}\right\rceil$,
\item
$\dim\RST{n} = \left\lfloor\frac{n}{3}\right\rfloor + \deltamod{1}{n}{3}$,
\item
$\dim\DST{n} = \left\lfloor\frac{n+3}{6}\right\rfloor + \deltamod{1}{n}{6}$,
\end{itemize}
where $\deltamod{i}{n}{j}$ is equal to $1$ if $n\equiv i\tpmod{j}$ and $0$ otherwise. Bases of $\HST{n}$, $\RST{n}$ and $\DST{n}$ were obtained in \cite{Brunat:2011aa}. In this paper, we give new bases, for each of these three linear subspaces, which are simpler than those mentioned. They are obtained by considering elementary properties of generalized binomial coefficients.

A {\em Steinhaus graph} of order $n\ge1$ is a simple graph whose adjacency matrix has an upper-triangular part which is a binary Steinhaus triangle of size $n-1$. For any sequence $S=(a_1,a_2,\ldots,a_{n-1})$ of $0$'s and $1$'s of length $n-1$, its associated Steinhaus graph $\Sgraph{S}$ is the simple graph of order $n$ whose adjacency matrix $\Smat{S}=\left(a_{i,j}\right)_{1\le i,j\le n}$ verifies
\begin{enumerate}[i)]
\item
$a_{i,j} = a_{j,i}$, for all $i,j\in\{1,\ldots,n\}$, (symmetry)
\item
$a_{i,i} = 0$, for all $i\in\{1,\ldots,n\}$, (diagonal of zeroes)
\item
$a_{1,j} = a_{j-1}$, for all $j\in\{2,\ldots,n-1\}$, (sequence $S$)
\item
$a_{i,j} \equiv a_{i-1,j-1}+a_{i-1,j}\tpmod{2}$, for all integers $i,j$ such that $2\le i<j\le n$, (local rule of $\Strig{S}$)
\end{enumerate}
where $\left\{x,\ldots,y\right\}$ denotes the set of integers $\left\{i\in\Z\ \middle|\ x\le i\le y\right\}$, for any integers $x$ and $y$. For example, for $S=(0010100)$, the Steinhaus graph $\Sgraph{S}$ and its adjacency matrix $\Smat{S}$ are depicted in Figure~\ref{fig4}.

\begin{figure}[htbp]
\centerline{
\begin{tabular}{cc}
\begin{minipage}{5cm}
\begin{tikzpicture}[scale=0.5]
\figSG{{0,0,1,0,1,0,0}}{black}{white}
\end{tikzpicture}
\end{minipage}
&
\begin{minipage}{5cm}
\begin{tikzpicture}[scale=0.25]
\figSM{{0,0,1,0,1,0,0}}{black}{col0}{col1}{col3}
\end{tikzpicture}
\end{minipage}
\\ \ \\
$\Sgraph{0010100}$ & $\Smat{0010100}$ \\
\end{tabular}}
\caption{The Steinhaus graph $\Sgraph{0010100}$ and its adjacency matrix $\Smat{0010100}$}\label{fig4}
\end{figure}

The set of Steinhaus graphs of order $n$ is denoted by $\SG{n}$.  It is clear that there is a natural correspondence between $\SG{n}$ and $\ST{n-1}$. Therefore, the set $\SG{n}$ is a vector space over $\Zn{2}$ of dimension $n-1$.

The family of Steinhaus graphs was introduced in \cite{Molluzzo:1978aa}. In \cite{Delahan:1998aa}, it was proved that any simple graph of order $n$ is isomorphic to an induced subgraph of a Steinhaus graph of order $\binom{n}{2}+1$. A general problem on Steinhaus graphs is to characterize those, or their associated binary sequences, having a given graph property such as connectedness, planarity, bipartition, regularity, etc. It is easy to see that a Steinhaus graph is either connected or totally disconnected (the edgeless graph). The bipartite Steinhaus graphs are characterized in \cite{Dymacek:1986aa,Dymacek:1995aa,Chang:1999aa} and the planar ones in \cite{Dymacek:2000aa}. In \cite{Dymacek:1979aa,Bailey:1988aa}, it was conjectured that there is only one regular Steinhaus graph of odd degree, the complete graph $K_2 = \Sgraph{1}$, and that the regular Steinhaus graphs of even degree are the edgeless graphs $\overline{K_n} = \Sgraph{00\cdots 0}$ of orders $n$, for all positive integers $n$, and the non-trivial graphs $\Sgraph{110110\cdots 110}$ of orders $n=3m+1$, for all positive integers $m$. This conjecture was verified up to $117$ vertices in \cite{Augier:2008aa} and up to $1500$ vertices in \cite{Chappelon:2009aa} for the odd case.

A Steinhaus graph is said to be {\em even} (resp. {\em odd}) if every vertex has even degree (resp. odd degree). Examples of even and odd Steinhaus graphs are given in Figure~\ref{fig5}. The sets of even Steinhaus graphs and of odd Steinhaus graphs of order $n$ are denoted by $\ESG{n}$ and $\OSG{n}$, respectively. In \cite{Dymacek:1979aa}, it was proved that $\ESG{n}$ is a linear subspace of $\SG{n}$ of dimension $\left\lfloor\frac{n-1}{3}\right\rfloor$ for all positive integers $n$ and $\OSG{n}$ is an affine subspace of direction $\ESG{n}$ for all even numbers $n$. Obviously, since the number of vertices of odd degrees is always even, $\OSG{n}=\emptyset$ when $n$ is odd. According to the terminology used in \cite{Augier:2008aa}, a {\em parity-regular} Steinhaus graph is a Steinhaus graph that is even or odd. The set of parity-regular Steinhaus graphs of order $n$ is denoted by $\PRSG{n}$, that is, $\PRSG{n}=\ESG{n}\cup\OSG{n}$. As shown in \cite{Augier:2008aa}, the set $\PRSG{n}$ is a linear subspace of $\SG{n}$ of dimension $\left\lceil\frac{n}{3}\right\rceil - \deltamod{1}{n}{2}$. Bases of $\PRSG{n}$ have been computed, for $n\le 30$, in \cite{Augier:2008aa}. In this paper, we determine bases of $\ESG{n}$ and $\PRSG{n}$ for all $n\ge1$. This is achieved by showing that the vector space $\ESG{n}$ is isomorphic to a particular linear subspace of $\DST{2n-1}$.

\begin{figure}[htbp]
\centerline{\includegraphics[width=\textwidth]{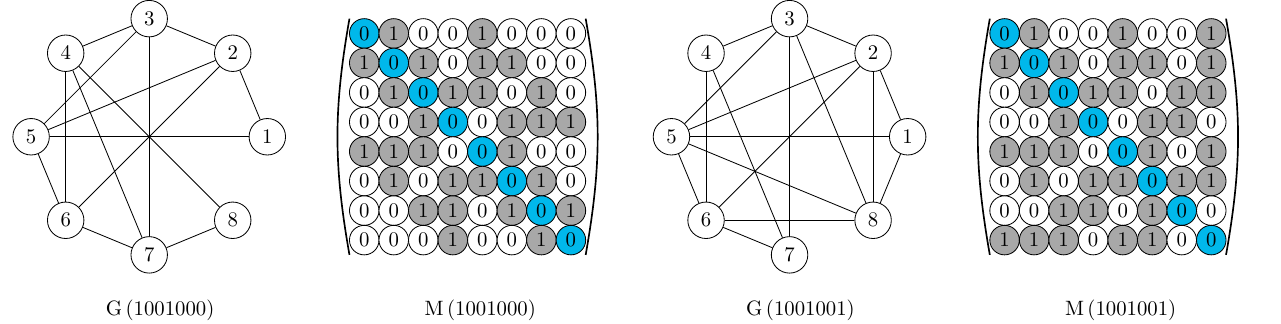}}
\caption{Even and odd Steinhaus graphs}\label{fig5}
\end{figure}

This paper is organized as follows. In Section~\ref{sec:2}, some basic properties on generating index sets of Steinhaus triangles and on derived and antiderived binary sequences are introduced. After that, for all non-negative integers $n$, the linear subspaces $\RST{n}$, $\HST{n}$ and $\DST{n}$ are studied in detail, with the determination of generating index sets and bases for each of them: for $\RST{n}$ in Section~\ref{sec:4}, for $\HST{n}$ in Section~\ref{sec:3} and for $\DST{n}$ in Section~\ref{sec:5}. For any positive integer $n$, a certain correspondence between the vector spaces $\ESG{n}$ and $\DST{2n-1}$ is established and bases of $\ESG{n}$ and $\PRSG{n}$ are given in Section~\ref{sec:6}. Finally, in Section~\ref{sec:7}, we deal with symmetric generalized Pascal triangles, that are binary triangles which point up and always with the same local rule.

%%%%%%%%%%%%%%%%
%%%%%%%%%%%%%%%%
\section{Preliminary results}\label{sec:2}
%%%%%%%%%%%%%%%%
%%%%%%%%%%%%%%%%

We introduce, in this section, the notions of generating index sets of Steinhaus triangles and of derived and antiderived binary sequences.

%%%%%%%%%%%%%%%%
\subsection{Generating index sets}
%%%%%%%%%%%%%%%%

Let $n$ be a positive integer. We denote by $\Strig{(n)}$ the index set of Steinhaus triangles of size $n$; i.e.,
$$
\Strig{(n)} = \left\{ (i,j)\in\N^2\ \middle|\ 1\le i\le j\le n\right\}.
$$
A subset $G$ of $\Strig{(n)}$ is said to be a {\em generating index set of $\ST{n}$} if the knowledge of the values $a_{i,j}$, for all $(i,j)\in G$, uniquely determines the whole Steinhaus triangle $(a_{i,j})_{1\le i\le j\le n}$ where the size of $G$ is minimal; i.e., a subset $G$ that makes the following linear map an isomorphism:
$$
\pi_G : \ST{n} \longrightarrow \Zd^{|G|}\quad\text{by}\quad \pi_G\left((a_{i,j})_{1\le i\le j\le n}\right)=(a_{i,j})_{(i,j)\in G},
$$
where $\Zd=\left\{0,1\right\}$. Since $\dim\ST{n}=n$, we deduce that the cardinality of a generating index set of $\ST{n}$ is always $n$. From \eqref{eq14}, it is clear that the set of top row indices of a Steinhaus triangle of size $n$, that is,
$$
G_{1} = \left\{ (1,1) , (1,2) , \ldots , (1,n)\right\},
$$
is a generating index set of $\ST{n}$. Note that $\pi_{G_1}^{-1}(S)=\Strig{S}$, for all $S\in\Zd^{|G_1|}$. It follows that the set $G$ is a generating index set of $\ST{n}$ if and only if the linear map $\pi_G\circ\pi_{G_1}^{-1} : \Zd^{|G_1|}\rightarrow \Zd^{|G|}$ is an isomorphism. For instance, the $16$ generating index sets of $\ST{3}$ ($4$ up to the action of the dihedral group $D_3$) are depicted in Figure~\ref{fig7}, where a disk is either black if its position is in the generating index set or white otherwise. Moreover, the four $3$-subsets of $\Strig{(3)}$ that are not generating index sets of $\ST{3}$ are given in Figure~\ref{fig7b}.

\begin{figure}[htbp]
\centerline{\includegraphics[width=\textwidth]{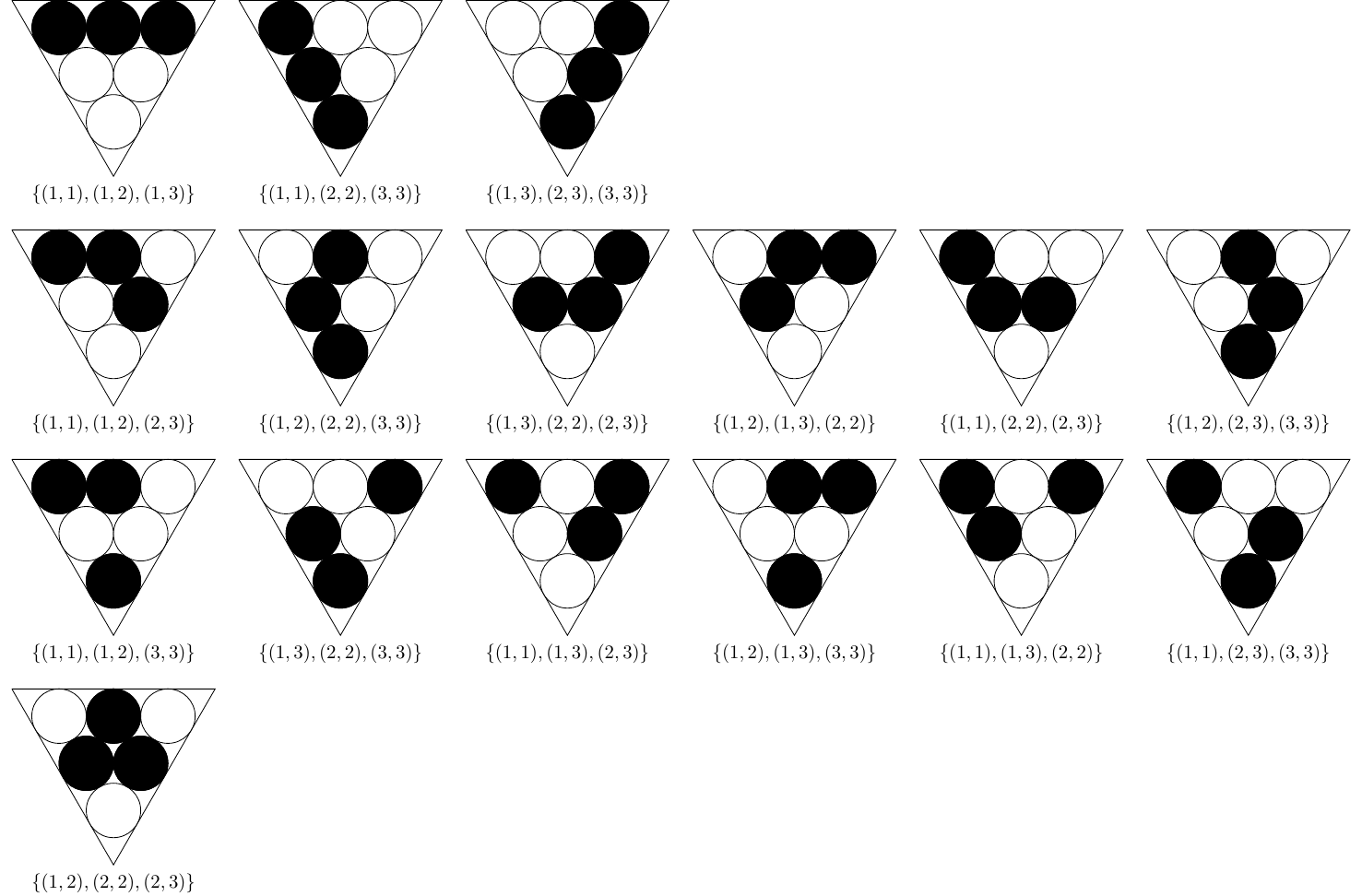}}
\caption{Generating index sets of $\ST{3}$}\label{fig7}
\end{figure}

\begin{figure}[htbp]
\centerline{\includegraphics[width=0.75\textwidth]{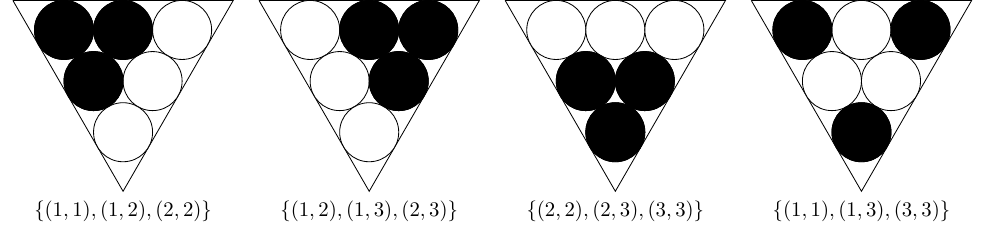}}
\caption{$3$-subsets of $\Strig{(3)}$ that are not generating index sets of $\ST{3}$}\label{fig7b}
\end{figure}

Since the sets of right side indices,
$$
G_r = \left\{ (1,n) , (2,n) , \ldots , (n,n) \right\},
$$
and left side indices,
$$
G_\ell = \left\{ (1,1) , (2,2) , \ldots , (n,n) \right\},
$$
of a Steinhaus triangle $\Strig{}$ of size $n$ can be seen as the sets of top row indices of the Steinhaus triangles $r\left(\Strig{}\right)$ and $r^2\left(\Strig{}\right)$, respectively, it follows that $G_r$ and $G_\ell$ are generating index sets of $\ST{n}$ too. Therefore, each element of a Steinhaus triangle can be expressed as a function of the terms of its first row, of its right side or of its left side.

For any non-negative integers $a$ and $b$ such that $b\le a$, the binomial coefficient $\binom{a}{b}$ is the coefficient of the monomial $X^b$ in the polynomial expansion of the binomial power $(1+X)^{a}$. It corresponds to the number of ways to choose $b$ elements in a set of $a$ elements. Here, we extend this notation by supposing that $\binom{a}{b}=0$, for all integers $b$ such that $b<0$ or $b>a$. For this generalization, the Pascal identity
$$
\binom{a}{b} = \binom{a-1}{b-1} + \binom{a-1}{b}
$$
holds, for all positive integers $a$ and all integers $b$. 

\begin{lem}\label{lem1}
Let $(a_{i,j})_{1\le i\le j\le n}$ be a binary Steinhaus triangle of size $n$. Then, we have
$$
a_{i,j} \equiv \sum_{k=1}^{n}\binom{i-1}{j-k}a_{1,k} \equiv \sum_{k=1}^{n}\binom{n-j}{k-i}a_{k,n} \equiv \sum_{k=1}^{n}\binom{j-i}{k-i}a_{k,k} \pmod{2},
$$
for all integers $i,j$ such that $1\le i\le j\le n$. 
\end{lem}

\begin{proof}
This is easily seen by induction on $i$ using \eqref{eq14} and the local rule.
\end{proof}

\begin{prop}\label{prop2}
Let $G=\{(i_1,j_1),(i_2,j_2),\ldots,(i_n,j_n)\}$ be a subset of $\Strig{(n)}$ whose cardinality is $|G|=n$. Then, the set $G$ is a generating index set of $\ST{n}$ if and only if $\det(M_G)\equiv1\tpmod{2}$, where
$$
M_G = \left(\binom{i_k-1}{j_k-\ell}\right)_{1\le k,\ell\le n}.
$$
\end{prop}

\begin{proof}
From Lemma~\ref{lem1}, we know that
$$
a_{i_k,j_k} \equiv \sum_{\ell=1}^{n}\binom{i_k-1}{j_k-\ell}a_{1,\ell} \pmod{2},
$$
for all $k\in\{1,\ldots,n\}$. It follows that
$$
\pi_G\circ\pi_{G_1}^{-1}(S) \equiv S.{M_G}^{t} \pmod{2},
$$
for all $S\in\Zd^{|G_1|}$. Finally, the linear map $\pi_G\circ\pi_{G_1}^{-1}$ is an isomorphism if and only if $\det(M_G)\equiv1\tpmod{2}$.
\end{proof}

The notion of generating index sets and the result of Proposition~\ref{prop2} appear in a more general context in \cite{Barbe:2004aa,Barbe:2009aa}, where it is also proved that the set of generating index sets of $\ST{n}$ define a matroid called the Pascal matroid modulo $2$. Note that a generating index set is simply called a generating set in \cite{Barbe:2004aa,Barbe:2009aa}.

The definition of generating index sets can be extended to any linear subspace $\mathcal{V}$ of $\ST{n}$. A subset $G$ of $\Strig{(n)}$ is said to be a {\em generating index set of $\mathcal{V}$} if the linear map
$$
\pi_G : \mathcal{V} \longrightarrow \Zd^{|G|}\quad\text{by}\quad \pi_G\left((a_{i,j})_{1\le i\le j\le n}\right)=(a_{i,j})_{(i,j)\in G}
$$
is an isomorphism. Note that $|G|=\dim\mathcal{V}$, for any generating index set $G$ of $\mathcal{V}$. In this paper, we consider generating index sets of the linear subspaces $\RST{n}$, $\HST{n}$ and $\DST{n}$.

%%%%%%%%%%%%%%%%
\subsection{Derived and antiderived sequences}
%%%%%%%%%%%%%%%%

Let $S=\left(a_j\right)_{1\le j\le n}$ be a sequence of $0$'s and $1$'s of length $n$. The {\em derived sequence} $\partial S$ of $S$ is the sequence
\begin{equation}\label{eq15}
\partial S = \left( (a_j+a_{j+1}) \tpmod{2} \right)_{1\le j\le n-1}
\end{equation}
of length $n-1$ when $n\ge2$ and the empty sequence when $n\le1$. The {\em iterated derived sequences} $\partial^{i}S$ of $S$ are recursively defined by $\partial^{i}S=\partial(\partial^{i-1}S)$, for all $i\ge1$, with $\partial^{0}S=S$. The Steinhaus triangle $\Strig{S}$ can then be seen as the collection $\left(\partial^{i}S\right)_{0\le i\le n-1}$, where, for every $i\in\{1,\ldots,n\}$, the $i$th row of $\Strig{S}$ corresponds to the derived sequence $\partial^{i-1}S$.

The set of binary sequences of length $n$ can be seen as a vector space over $\Zn{2}$ of dimension $n$. Indeed, for two binary sequences $S_1=\left(a_j\right)_{1\le j\le n}$ and $S_2=\left(b_j\right)_{1\le j\le n}$ of the same length $n\ge1$, their sum is the sequence $S_1+S_2=\left(a_j+b_j\right)_{1\le j\le n}$ of length $n$. Therefore, it is clear that the derivation map $\partial$ is linear, i.e., $\partial(S_1+S_2)=\partial S_1 + \partial S_2$ for all binary sequences $S_1$ and $S_2$ of same length.

For any $i\in\{1,\ldots,n+1\}$ and any $x\in\Zd$, the {\em antiderived sequence of $S={\left(a_{j}\right)}_{1\le j\le n}$ whose $i$th term is $x$} is the sequence $\int_{i,x}S = (b_j)_{1\le j\le n+1}$ of length $n+1$ defined by
$$
b_j = \left\{\begin{array}{ll}
x + \displaystyle\sum_{k=j}^{i-1}a_k \pmod{2}& \text{for } j\in\{1,\ldots,i-1\},\\
x & \text{for } j=i,\\
x + \displaystyle\sum_{k=i}^{j-1}a_k \pmod{2}& \text{for } j\in\{i+1,\ldots,n+1\}.
\end{array}\right.
$$
In a more concise way, we have
\begin{equation}\label{eq16}
b_j \equiv \left(x+\sum_{k=1}^{i-1}a_k+\sum_{k=1}^{j-1}a_k\right)\pmod{2},
\end{equation}
for all $j\in\{1,\ldots,n+1\}$. Further, it is straightforward to obtain a fundamental theorem of calculus.

For any non-negative integer $n$, the constant sequence of length $n$ equal to $x$ is denoted by $\left(x\right)_{n}$. For $n=1$, the sequence $\left(x\right)_{1}$ is simply denoted $\left(x\right)$.

\begin{prop}
Let $S=\left(a_j\right)_{1\le j\le n}$ be a binary sequence of length $n$. For any $i\in\{1,\ldots,n+1\}$ and any $x\in\Zd=\{0,1\}$, we have that
\begin{enumerate}[i)]
\item\label{item1}
$\partial\int_{i,x}S = S$,
\item\label{item2}
$\int_{i,x}\partial S = S + \left((a_i+x)\tpmod{2}\right)_{n}$.
\end{enumerate}
\end{prop}

\begin{proof}
For \ref{item1}), let $\int_{i,x}S=\left(b_j\right)_{1\le j\le n+1}$ and $\partial\int_{i,x}S = \left(c_j\right)_{1\le j\le n}$. From the definition and equations~\eqref{eq15} and \eqref{eq16}, we have
$$
c_j \equiv b_j+b_{j+1} \equiv \left(x+\sum_{k=1}^{i-1}a_k+\sum_{k=1}^{j-1}a_k\right) + \left(x+\sum_{k=1}^{i-1}a_k+\sum_{k=1}^{j}a_k\right) \equiv a_j\pmod{2},
$$
for all $j\in\{1,\ldots,n\}$. Therefore $\partial\int_{i,x}S = S$.

Now, for \ref{item2}), let $\partial S=\left(b_j\right)_{1\le j\le n-1}$ and $\int_{i,x}\partial S = \left(c_j\right)_{1\le j\le n}$. From the definition and equations~\eqref{eq15} and \eqref{eq16}, we have
$$
c_j
\begin{array}[t]{l}
\equiv x + \displaystyle\sum_{k=1}^{i-1}b_k + \sum_{k=1}^{j-1}b_k \equiv x + \sum_{k=1}^{i-1}(a_k+a_{k+1}) + \sum_{k=1}^{j-1}(a_k+a_{k+1}) \\ \ \\
\equiv x + (a_1+a_i) + (a_1+a_j) \equiv a_j+(a_i+x)\pmod{2},
\end{array}
$$
for all $j\in\{1,\ldots,n\}$. Therefore $\int_{i,x}\partial S = S + \left((a_i+x)\tpmod{2}\right)_{n}$.
\end{proof}

A similar result has been obtained for infinite binary sequences in \cite{Nathanson:1971aa}. It follows that every binary sequence $S$ of length $n$ admits only two different antiderived sequences which are the sequences $\int_{i,0}S$ and $\int_{i,1}S$ for some $i\in\{1,\ldots,n+1\}$. Moreover, it is easy to see that $\int_{i,0}S + \int_{i,1}S = \left(1\right)_n$. For example, the sequence $S=(0100)$ admits the two antiderived sequences $(00111)$ and $(11000)$.

%%%%%%%%%%%%%%%%
%%%%%%%%%%%%%%%%
\section{Rotationally symmetric Steinhaus triangles}\label{sec:4}
%%%%%%%%%%%%%%%%
%%%%%%%%%%%%%%%%

In this section, after characterizing rotationally symmetric Steinhaus triangles, we determine generating index sets and bases of $\RST{n}$.

%%%%%%%%%%%%%%%%
\subsection{Characterizations of \texorpdfstring{$\RST{n}$}{RST(n)}}
%%%%%%%%%%%%%%%%

First, by definition of the automorphism $r$, we have
$$
r\left(\left(a_{i,j}\right)_{1\le i\le j\le n}\right)=\left(a_{j-i+1,n-i+1}\right)_{1\le i\le j\le n}=\Strig{\left(a_{j,n}\right)_{1\le j\le n}},
$$
for any Steinhaus triangle $\left(a_{i,j}\right)_{1\le i\le j\le n}=\Strig{\left(a_{1,j}\right)_{1\le j\le n}}$. Therefore, a Steinhaus triangle $\left(a_{i,j}\right)_{1\le i\le j\le n}$ is rotationally symmetric if and only if its first row $\left(a_{1,j}\right)_{1\le j\le n}$ and its right side $\left(a_{j,n}\right)_{1\le j\le n}$ correspond.

\begin{prop}\label{prop20}
The Steinhaus triangle $\left(a_{i,j}\right)_{1\le i\le j\le n}$ is rotationally symmetric if and only if $\left(a_{1,j}\right)_{1\le j\le n}=\left(a_{j,n}\right)_{1\le j\le n}$.
\end{prop}

For two binary sequences $S_1=\left(a_1,a_2,\ldots,a_{n_1}\right)$ and $S_2=\left(b_1,b_2,\ldots,b_{n_2}\right)$ of length $n_1$ and $n_2$, respectively, we denote by $S_1\cdot S_2$ {\em the concatenated sequence} of length $n_1+n_2$ defined by $S_1\cdot S_2=\left(a_1,a_2,\ldots,a_{n_1},b_1,b_2,\ldots,b_{n_2}\right)$.

Let $\He$ be the linear map that assigns, to each Steinhaus triangle of order $n\ge 3$, its subtriangle of order $n-3$ obtained by removing its first row and its left and right sides; i.e.,
$$
\He : \ST{n} \longrightarrow \ST{n-3}\quad\text{by}\quad \He\left((a_{i,j})_{1\le i\le j\le n}\right)=(a_{1+i,2+j})_{1\le i\le j\le n-3}.
$$
Note that the linear map $\He$ is surjective. Indeed, for any $\Strig{S'}\in\ST{n-3}$, it is easy to verify that $\Strig{S'}=\He\left(\Strig{S}\right)$ if and only if $S$ is one of the eight sequences of the form $S=(x_1)\cdot\left(\int_{i,x}S'\right)\cdot (x_2)$, where $x_1,x_2\in\Zd$ and $\int_{i,x}S'$ is one of the two antiderived sequences of $S'$. Examples of a Steinhaus triangle $\Strig{S}$ and its subtriangle $\He\left(\Strig{S}\right)$ are depicted in Figure~\ref{fig22}.

\begin{figure}[htbp]
\centerline{
\begin{tikzpicture}[scale=0.25]
\figST{{1,0,1,1,1,1,0}}{black}{col0}{col1}{white}
\pgfmathparse{sqrt(3)}\let\sq\pgfmathresult
\draw[fill=col2!50!col1] (3,-\sq) circle (1);
\draw (3,-\sq) node {$1$};
\draw[fill=col2] (5,-\sq) circle (1);
\draw (5,-\sq) node {$0$};
\draw[fill=col2] (7,-\sq) circle (1);
\draw (7,-\sq) node {$0$};
\draw[fill=col2] (9,-\sq) circle (1);
\draw (9,-\sq) node {$0$};
\draw[fill=col2!50!col1] (4,-2*\sq) circle (1);
\draw (4,-2*\sq) node {$1$};
\draw[fill=col2] (6,-2*\sq) circle (1);
\draw (6,-2*\sq) node {$0$};
\draw[fill=col2] (8,-2*\sq) circle (1);
\draw (8,-2*\sq) node {$0$};
\draw[fill=col2!50!col1] (5,-3*\sq) circle (1);
\draw (5,-3*\sq) node {$1$};
\draw[fill=col2] (7,-3*\sq) circle (1);
\draw (7,-3*\sq) node {$0$};
\draw[fill=col2!50!col1] (6,-4*\sq) circle (1);
\draw (6,-4*\sq) node {$1$};
\end{tikzpicture}}
\caption{$\He\left(\Strig{(1011110)}\right)=\Strig{(1000)}$}\label{fig22}
\end{figure}

For any binary sequence $S=\left(a_j\right)_{1\le j\le n}$, we denote by $\sigma(S)$ its sum $\sigma(S) = \sum_{j=1}^{n}a_j$ and and by $\sigma_2(S)$ this sum modulo $2$. Hence $\sigma(S)$ is the number of $1$'s in the sequence $S$ and $\sigma_2(S)$ is this number modulo $2$.

For any positive integer $n\ge3$, by definition of $\RST{n}$ and $\He$, it is clear that $\He\left(\RST{n}\right)\subset\RST{n-3}$. The precise relationship between a rotationally symmetric Steinhaus triangle $\Strig{S}$ and its subtriangle $\He\left(\Strig{S}\right)$ is given in the following

\begin{prop}\label{prop5}
Let $S$ be a finite binary sequence of length $n\ge 3$. The Steinhaus triangle $\Strig{S}$ is rotationally symmetric if and only if $\He\left(\Strig{S}\right)=\Strig{S'}$ is rotationally symmetric and $S=\left(\sigma_2\left(S'\right)\right)\cdot\left(\int_{i,x}S'\right)\cdot\left(\sigma_2\left(S'\right)\right)$, for some $i\in\{1,\ldots,n-2\}$ and some $x\in\Zd$.
\end{prop}

Proposition~\ref{prop5} appears in \cite{Barbe:2001aa} in a more general context. For the convenience of the reader, a proof is given here.

\begin{proof}
Let $\Strig{S}=(a_{i,j})_{1\le i\le j\le n}\in\ST{n}$ be such that
$$
\Strig{S'}=\He(\Strig{S})=(a_{1+i,2+j})_{1\le i\le j\le n-3}\in\RST{n-3}.
$$
Since $\Strig{S'}$ is rotationally symmetric, its top row $(a_{2,j})_{3\le j\le n-1}$, its right side $(a_{i,n-1})_{2\le i\le n-2}$ and the reverse of its left side $(a_{n-i,n-i+1})_{2\le i\le n-2}$ correspond by Proposition~\ref{prop20}. Moreover, since $(a_{1,j})_{2\le j\le n-1}$, $(a_{i,n})_{2\le i\le n-1}$ and $(a_{n-i+1,n-i+1})_{2\le i\le n-1}$ are antiderived sequences of the sequences $(a_{2,j})_{3\le j\le n-1}$, $(a_{i,n-1})_{2\le i\le n-2}$ and $(a_{n-i,n-i+1})_{2\le i\le n-2}$, respectively, we deduce that they correspond if and only if there exist $i_1,i_2\in\{2,\ldots,n-1\}$ such that $a_{1,i_1}=a_{i_1,n}$ and $a_{1,i_2}=a_{n-i_2+1,n-i_2+1}$. Since $a_{1,1}\equiv \left(a_{1,2}+a_{2,2}\right)\tpmod{2}$ and $a_{1,n}\equiv \left(a_{1,n-1}+a_{2,n}\right)\tpmod{2}$, it follows from Proposition~\ref{prop20} that the Steinhaus triangle $\Strig{S}$ is rotationally symmetric if and only if $a_{1,1}=a_{1,n}\equiv \left(a_{1,2}+a_{1,n-1}\right)\tpmod{2}$. Finally, using \eqref{eq13}, we have that
$$
\sigma\left(S'\right) = \sum_{j=3}^{n-1}a_{2,j}\equiv \sum_{j=3}^{n-1}a_{1,j-1}+a_{1,j}\equiv a_{1,2}+a_{1,n-1} \pmod{2}.
$$
This completes the proof.
\end{proof}

%%%%%%%%%%%%%%%%
\subsection{Generating index sets of \texorpdfstring{$\RST{n}$}{RST(n)}}
%%%%%%%%%%%%%%%%

We are now ready to determine generating index sets of the linear subspace of rotationally symmetric Steinhaus triangles.

\begin{thm}\label{thm6}
Let $n$ be a non-negative integer. The set
$$
G_R = \left\{ (i,j_i)\ \middle|\  i\in\left\{1,\ldots,\left\lfloor\frac{n}{3}\right\rfloor + \deltamod{1}{n}{3} \right\} \right\},
$$
where $j_i\in\left\{2i,\ldots,n-i\right\}$ for all $i\in\left\{1,\ldots,\left\lfloor\frac{n}{3}\right\rfloor\right\}$ and $j_{\frac{n+2}{3}}=\frac{2n+1}{3}$ when $n\equiv1\tpmod{3}$, is a generating index set of $\RST{n}$.
\end{thm}

\begin{proof}
We proceed with induction on $n$. For $n=0$ and $n=2$, it is clear that $\emptyset$ and $\Strig{(00)}$ are the only rotationally symmetric Steinhaus triangles of sizes $0$ and $2$. Therefore, the empty set $\emptyset$ is a generating index set of $\RST{0}$ and $\RST{2}$. For $n=1$, the Steinhaus triangles $\Strig{(0)}$ and $\Strig{(1)}$ are both rotationally symmetric and thus the set $\left\{(1,1)\right\}$ is a generating index set of $\RST{1}$.

For $n\ge3$, suppose that the result is true for the sets of rotationally symmetric Steinhaus triangles of size strictly less than $n$. Let $m=\left\lfloor\frac{n}{3}\right\rfloor + \deltamod{1}{n}{3}$. We consider the subset $\He\left(G_R\right)\subset\Strig{(n-3)}$ defined by
$$
\He\left(G_R\right) = \left\{ (i-1,j_i-2) \ \middle|\ i\in\left\{2,\ldots, m\right\} \right\}
$$
and the linear maps $f_1$ and $f_2$, where $\Strig{S}=(a_{i,j})_{1\le i\le j\le n}$, defined by
$$
\begin{array}{llll}
f_1 : \RST{n} \longrightarrow \Zd\times\RST{n-3} & \text{by} & f_1\left(\Strig{S}\right)=\left(a_{1,j_1} , \He\left(\Strig{S}\right) \right) & \text{and} \\[1.5ex]
f_2 : \Zd\times\RST{n-3} \longrightarrow \Zd^m & \text{by} & f_2\left(x,\Strig{S'}\right)=(x)\cdot\left(\pi_{\He\left(G_R\right)}\left(\Strig{S'}\right)\right). \\
\end{array}
$$
Then, for any $(a_{i,j})_{1\le i\le j\le n}\in\RST{n}$, we have
$$
f_2f_1\left((a_{i,j})_{1\le i\le j\le n}\right) \begin{array}[t]{l}
 = f_2\left(a_{1,j_1},(a_{1+i,2+j})_{1\le i\le j\le n-3}\right) \\
 = (a_{1,j_1})\cdot \pi_{\He\left(G_R\right)}\left((a_{1+i,2+j})_{1\le i\le j\le n-3}\right) \\
 = (a_{1,j_1})\cdot \left(a_{i,j_i}\right)_{2\le i\le m} \\
 = \left(a_{i,j_i}\right)_{1\le i\le m} = \pi_{G_R}\left((a_{i,j})_{1\le i\le j\le n}\right). \\
\end{array}
$$
Therefore $f_2f_1=\pi_{G_R}$. From Proposition~\ref{prop5}, we know that $f_1$ is an isomorphism whose inverse is defined by $f_1^{-1}\left(x,\Strig{S'}\right)=\Strig{\left((\sigma_2(S'))\cdot\left(\int_{j_1-1,x}S'\right)\cdot(\sigma_2(S'))\right)}$, for all $\left(x,\Strig{S'}\right)\in\Zd\times\RST{n-3}$. Moreover, since we have
$$
1\le i-1\le \frac{n-3}{3}\quad\text{and}\quad2(i-1)\le j_i-2\le (n-3)-(i-1),
$$
for all $i\in\left\{2,\ldots,\left\lfloor\frac{n}{3}\right\rfloor\right\}$, and
$$
m-1 = \frac{(n-3)+2}{3}\quad\text{and}\quad j_{m}-2 = \frac{2(n-3)+1}{3},
$$
when $n\equiv1\tpmod{3}$, by the induction hypothesis the set $\He\left(G_R\right)$ is a generating index set of $\RST{n-3}$. Therefore $\pi_{\He\left(G_R\right)}$ and thus $f_2$ are isomorphisms. Finally, since the linear map $\pi_{G_R}=f_2f_1$ is an isomorphism, the set $G_R$ is a generating index set of $\RST{n}$.
\end{proof}

\begin{cor}\label{cor2}
Let $n$ be a non-negative integer. The set
$$
G_R := \left\{ \left(i,n-\left\lfloor\frac{n}{3}\right\rfloor\right)\ \middle|\  i\in\left\{1,\ldots,\left\lfloor\frac{n}{3}\right\rfloor + \deltamod{1}{n}{3}\right\} \right\}
$$
is a generating index set of $\RST{n}$.
\end{cor}

\begin{proof}
This follows from Theorem~\ref{thm6} since $2i\le n-\left\lfloor\frac{n}{3}\right\rfloor\le n-i$ for all $i\in\left\{1,\ldots,\left\lfloor\frac{n}{3}\right\rfloor\right\}$ and $n-\left\lfloor\frac{n}{3}\right\rfloor = \frac{2n+1}{3}$ when $n\equiv1\tpmod{3}$.
\end{proof}

Since the dimension of $\RST{n}$ corresponds to the cardinality of a generating index set $G_R$, it is straightforward to obtain the following

\begin{cor}\label{cor14}
For all non-negative integers $n$, we have $\dim\RST{n} = \left\lfloor\frac{n}{3}\right\rfloor + \deltamod{1}{n}{3}$.
\end{cor}

%%%%%%%%%%%%%%%%
\subsection{Bases of \texorpdfstring{$\RST{n}$}{RST(n)}}
%%%%%%%%%%%%%%%%

At the end of this section, using the generating index sets $G_R$ introduced above, we determine bases of the linear subspace of rotationally symmetric Steinhaus triangles.

First, we consider the linear map $\rho : \ST{n}\longrightarrow \RST{n}$ defined by $\rho=r^2+r+id_{\ST{n}}$. Obviously, this map is surjective since $\rho(\Strig{})=\Strig{}$ for all $\Strig{}\in\RST{n}$. Moreover, as detailed below, all the terms of $\rho\left(\Strig{}\right)$ can be expressed in terms of those of $\Strig{}$.

\begin{prop}\label{prop16}
For all $\left(a_{i,j}\right)_{1\le i\le j\le n}\in\ST{n}$, we have
$$
\rho\left(\left(a_{i,j}\right)_{1\le i\le j\le n}\right) = \left(\left(a_{i,j}+a_{j-i+1,n-i+1}+a_{n-j+1,n+i-j}\right)\tpmod{2}\right)_{1\le i\le j\le n}.
$$
\end{prop}

\begin{proof}
First, by definition of $r$, we know for all $\left(a_{i,j}\right)_{1\le i\le j\le n}\in\ST{n}$ that
$$
r\left(\left(a_{i,j}\right)_{1\le i\le j\le n}\right) = \left(a_{j-i+1,n-i+1}\right)_{1\le i\le j\le n}.
$$
Thus, we have
$$
r^2\left(\left(a_{i,j}\right)_{1\le i\le j\le n}\right) = r\left(\left(a_{j-i+1,n-i+1}\right)_{1\le i\le j\le n}\right) = \left(a_{n-j+1,n+i-j}\right)_{1\le i\le j\le n}.
$$
The result follows.
\end{proof}

For any non-negative integer $n$, let $U_n$ be the Steinhaus triangle of size $n$ defined by
$$
U_{n}=\rho\left(\Strig{(1)_n}\right).
$$
It is clear that $U_0=\emptyset$, $U_1=\Strig{(1)}$, $U_2=\Strig{(00)}$ and $U_{n}=\Strig{(011\cdots 110)}$ for $n\geq 3$, since
{\small
$$
U_n = \rho\left(\Strig{(1)_n}\right) = \Strig{(1\cdots1)}+\Strig{(10\cdots0)}+\Strig{(0\cdots01)} = \Strig{\left((1\cdots1)+(10\cdots0)+(0\cdots01)\right)}.
$$}
The Steinhaus triangles $U_n$ are depicted in Figure~\ref{fig9}, for the first few values of $n$. Moreover, an explicit formula for the terms of $U_n$ is given in the following

\begin{figure}[htbp]
\centerline{\includegraphics[width=\textwidth]{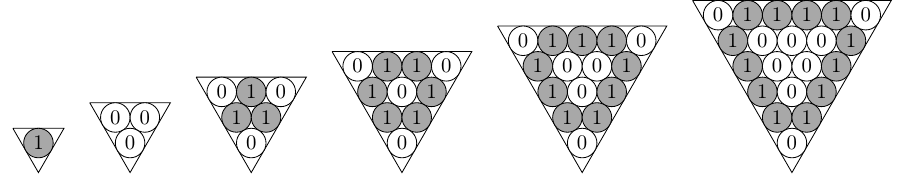}}
\caption{$U_n$ for $n\in\{1,\ldots,6\}$}\label{fig9}
\end{figure}

\begin{prop}\label{prop17}
For any non-negative integer $n$, we have that
$$
U_n = \left(\left(\delta_{i,1}+\delta_{i,j}+\delta_{j,n}\right)\tpmod{2}\right)_{1\le i\le j\le n}.
$$
\end{prop}

\begin{proof}
First, if we denote $\Strig{(1)_n}=\left(a_{i,j}\right)_{1\le i\le j\le n}$ and $U_n=\rho\left(\Strig{(1)_n}\right)=\left(b_{i,j}\right)_{1\le i\le j\le n}$, we know from Proposition~\ref{prop16} that
$$
b_{i,j} \equiv \left(a_{i,j}+a_{j-i+1,n-i+1}+a_{n-j+1,n+i-j}\right) \pmod{2}.
$$
Moreover, it is clear that $\partial^{i}(1)_n=(0)_{n-i}$, for all $i\in\{1,\ldots,n-1\}$. Therefore, $\Strig{(1)_n}=\left(\delta_{i,1}\right)_{1\le i\le j\le n}$. This leads to
$$
b_{i,j} \equiv \left(\delta_{i,1} + \delta_{j-i+1,1} + \delta_{n-j+1,1}\right) \pmod{2}.
$$
Finally, since $\delta_{j-i+1,1}=\delta_{i,j}$ and $\delta_{n-j+1,1}=\delta_{j,n}$, the result follows.
\end{proof}

\begin{cor}
For all positive integers $n\ge3$, we have $\He\left(U_n\right)=\Strig{(0)_{n-3}}$.
\end{cor}

\begin{proof}
Let $U_n=\left(a_{i,j}\right)_{1\le i\le j\le n}$. By definition of $\He$ and Proposition~\ref{prop17}, we have that
$$
\He\left(U_n\right) = \left(a_{i+1,j+2}\right)_{1\le i\le j\le n-3} = \left(\left(\delta_{i+1,1}+\delta_{i+1,j+2}+\delta_{j+2,n}\right)\tpmod{2}\right)_{1\le i\le j\le n-3} =\Strig{(0)_{n-3}}.
$$
\end{proof}

For any non-negative integer $k$ such that $3k\le n$, we consider the iterated operator
$$
\He^k=\underbrace{\He\He\cdots\He}_{k\text{ times}} : \ST{n} \longrightarrow \ST{n-3k} \quad\text{by}\quad \He^k\left((a_{i,j})_{1\le i\le j\le n}\right) = (a_{k+i,2k+j})_{1\le i\le j\le n-3k}.
$$
Using the operators $\He^k$ and the generating index set $G_R$, we obtain a family of bases of $\RST{n}$.

\begin{thm}\label{thm1}
Let $n$ and $m$ be non-negative integers such that $m=\left\lfloor\frac{n}{3}\right\rfloor + \deltamod{1}{n}{3}$. For every $k\in\left\{0,\ldots,m-1\right\}$, let $\iStrig{k}\in\RST{n}$ be such that $\He^{k}\left(\iStrig{k}\right)=U_{n-3k}$. Then, the set $\left\{\iStrig{0},\ldots,\iStrig{m-1}\right\}$ is a basis of $\RST{n}$.
\end{thm}

The proof is based on the following

\begin{lem}\label{lem2}
Let $\Strig{}=(a_{i,j})_{1\le i\le j\le n}\in\ST{n}$ be such that $\He^{k}\left(\Strig{}\right)=U_{n-3k}$, for some $k\in\left\{0,\ldots,\left\lfloor\frac{n}{3}\right\rfloor-1\right\}$. Then,
$$
a_{i,n-\left\lfloor\frac{n}{3}\right\rfloor} = \left\{\begin{array}{ll}
 1 & \text{for } i=k+1,\\
 0 & \text{for } i\in\left\{k+2,\ldots,\left\lfloor\frac{n}{3}\right\rfloor\right\}.\\ 
\end{array}\right.
$$
Moreover, when $n\equiv1\tpmod{3}$, if $\He^{\left\lfloor\frac{n}{3}\right\rfloor}\left(\Strig{}\right)=U_1=(1)$, then $a_{\left\lfloor\frac{n}{3}\right\rfloor+1,n-\left\lfloor\frac{n}{3}\right\rfloor}=1$.
\end{lem}

\begin{proof}[Proof of Lemma~\ref{lem2}]
Let $k\in\left\{0,\ldots,\left\lfloor\frac{n}{3}\right\rfloor-1\right\}$. By definition of $\He^k$ and Proposition~\ref{prop17}, we deduce that
\begin{equation}\label{eq20}
\left(a_{k+i,2k+j}\right)_{1\le i\le j\le n-3k} = \He^k\left(\Strig{}\right) = U_{n-3k} = \left(\left(\delta_{i,1}+\delta_{i,j}+\delta_{j,n-3k}\right)\tpmod{2}\right)_{1\le i\le j\le n-3k}.
\end{equation}
Since $k\le\left\lfloor\frac{n}{3}\right\rfloor-1$, we have
$$
2k+1\le 2\left\lfloor\frac{n}{3}\right\rfloor-1<n-\left\lfloor\frac{n}{3}\right\rfloor<n-k.
$$
It follows from \eqref{eq20} that for all integers $i$ such that $1\le i-k\le n-\left\lfloor\frac{n}{3}\right\rfloor-2k$ or what is the same thing, $k+1\le i\le n-\left\lfloor\frac{n}{3}\right\rfloor-k$, that
\begin{equation}\label{eq21}
a_{i,n-\left\lfloor\frac{n}{3}\right\rfloor} = a_{k+(i-k),2k+(n-\left\lfloor\frac{n}{3}\right\rfloor-2k)} \equiv \delta_{i-k,1} + \delta_{i-k,n-\left\lfloor\frac{n}{3}\right\rfloor-2k} + \delta_{n-\left\lfloor\frac{n}{3}\right\rfloor-2k,n-3k} \pmod{2}.
\end{equation}
Since $k\le\left\lfloor\frac{n}{3}\right\rfloor-1$, we obtain
\begin{equation}\label{eq22}
\delta_{n-\left\lfloor\frac{n}{3}\right\rfloor-2k,n-3k} = \delta_{k,\left\lfloor\frac{n}{3}\right\rfloor} = 0.
\end{equation}
Moreover, this leads to
$$
n-\left\lfloor\frac{n}{3}\right\rfloor-k \ge n-2\left\lfloor\frac{n}{3}\right\rfloor+1\ge \left\lfloor\frac{n}{3}\right\rfloor+1
$$
and for all integers $i\le\left\lfloor\frac{n}{3}\right\rfloor$,
\begin{equation}\label{eq23}
\delta_{i-k,n-\left\lfloor\frac{n}{3}\right\rfloor-2k} = \delta_{i,n-\left\lfloor\frac{n}{3}\right\rfloor-k} = 0.
\end{equation}
Therefore, for all $i\in\left\{k+1,\ldots,\left\lfloor\frac{n}{3}\right\rfloor\right\}$, we obtain from \eqref{eq21}, \eqref{eq22} and \eqref{eq23} that
$$
a_{i,n-\left\lfloor\frac{n}{3}\right\rfloor} = \delta_{i-k,1} = \delta_{i,k+1}.
$$
This completes the proof when $k\in\left\{0,\ldots,\left\lfloor\frac{n}{3}\right\rfloor-1\right\}$.

Finally, when $n\equiv1\tpmod{3}$, it is clear that
$$
(1) = U_1 = \He^{\frac{n-1}{3}}\left(\Strig{}\right) = \left(a_{\frac{n+2}{3},\frac{2n+1}{3}}\right) = \left(a_{\left\lfloor\frac{n}{3}\right\rfloor+1,n-\left\lfloor\frac{n}{3}\right\rfloor}\right).
$$
This completes the proof.
\end{proof}

\begin{proof}[Proof of Theorem~\ref{thm1}]
We consider the set
$$
G_R:=\left\{ \left(i,n-\left\lfloor\frac{n}{3}\right\rfloor\right)\ \middle|\ i\in\{1,2,\ldots,m\}\right\}.
$$
For any $k\in\{0,\ldots,m-1\}$, since $\He^{k}\left(\iStrig{k}\right)=U_{n-3k}$, it follows from Lemma~\ref{lem2} that
$$
\pi_{G_R}(\iStrig{k}) = ( \underbrace{\ast , \ldots, \ast}_{k} , 1 , \underbrace{0 , \ldots , 0}_{m-k-1}  ),
$$
where $\ast$ stands for any element of $\Zd$. Therefore, the set $\left\{ \pi_{G_R}(\iStrig{k})\ \middle|\ k\in\{0,\ldots,m-1\}\right\}$ is a basis of $\Zd^{|G_R|}$. Finally, by Corollary~\ref{cor2}, since $G_R$ is a generating index set of $\RST{n}$, we conclude that the set $\left\{ \iStrig{k}\ \middle|\ k\in\{0,\ldots,m-1\}\right\}$ is a basis of $\RST{n}$.
\end{proof}

Since $U_{n-3k}=\rho\left((1)_{n-3k}\right)$ by definition, this leads to the following

\begin{cor}\label{cor1}
Let $n$ and $m$ be non-negative integers such that $m=\left\lfloor\frac{n}{3}\right\rfloor + \deltamod{1}{n}{3}$. For every $k\in\{0,\ldots,m-1\}$, let $S_k$ be a binary sequence of length $n$ such that $\partial^{k}S_k = (1)_{n-k}$. Then, the set $\left\{\rho\left(\Strig{S_0}\right),\ldots,\rho\left(\Strig{S_{m-1}}\right)\right\}$ is a basis of $\RST{n}$.
\end{cor}

\begin{proof}
Let $k\in\{0,\ldots,m-1\}$. First, by definition of the linear map $\rho$, we know that $\rho\left(\Strig{S_k}\right)\in\RST{n}$. Moreover, since $\partial^{k}S_k=(1)_{n-k}$, it follows that
$$
\He^{k}\left(\rho\left(\Strig{S_k}\right)\right)=\rho\left(\He^{k}\left(\Strig{S_k}\right)\right)=\rho\left(\Strig{(1)_{n-3k}}\right)=U_{n-3k}.
$$
Therefore, from Theorem~\ref{thm1}, the set $\left\{\rho\left(\Strig{S_k}\right)\ \middle|\ k\in\{0,\ldots,m-1\}\right\}$ is a basis of $\RST{n}$.
\end{proof}

Binomial coefficients $\binom{a}{b}$ have been defined before for any non-negative integer $a$ and any integer $b$. Now, we extend this definition for negative integers $a$. For any integers $a$ and $b$, let $\binom{a}{b}$ denote the integers recursively defined by
\begin{itemize}
\item
$\binom{a}{0}=1$, for all $a\in\Z$,
\item
$\binom{0}{b}=0$, for all $b\in\N\setminus\{0\}$,
\item
$\binom{a}{b} = \binom{a-1}{b-1} + \binom{a-1}{b}$, for all $a,b\in\Z$.
\end{itemize}
When $a$ is non-negative, this corresponds with the previous definition. Moreover, for any negative integer $a$, the following equality holds
$$
\binom{a}{b} = \left\{\begin{array}{ll}
0 & \text{for } b<0,\\
(-1)^b\binom{b-a-1}{b} & \text{for } b\ge0.
\end{array}\right.
$$

\begin{rem}
There exist different extensions of the binomial coefficients $\binom{a}{b}$. For instance, for real numbers $a$ and non-negative integers $b$, one can consider the numbers
$$
\frac{1}{b!}\prod_{i=0}^{b-1}(a-i) = \frac{a(a-1)\cdots(a-b+1)}{b(b-1)\cdots 1}.
$$
Another extension, for complex numbers $a$ and $b$, is
$$
\lim_{x\to a}\lim_{y\to b}\frac{\Gamma(x+1)}{\Gamma(y+1)\Gamma(x-y+1)}
$$
where $\Gamma$ is the Gamma function. Note that all these extensions are compatible on the intersection of their domains.
\end{rem}

In this paper, we mainly consider the {\em infinite Pascal matrix modulo $2$}; i.e., the doubly indexed sequence $\left(\binomd{a}{b}\right)_{(a,b)\in\Z^2}$, where $\binomd{a}{b}$ is the value of $\binom{a}{b}\tpmod{2}$. The first few values of this doubly infinite sequence are shown in Figure~\ref{fig6}, where the terms $\binomd{a}{0}$ are in blue for all integers $a$ and the terms $\binomd{0}{b}$ are in red for all positive integers $b$.

\begin{figure}[htbp]
\centerline{\includegraphics[width=1.2\textwidth]{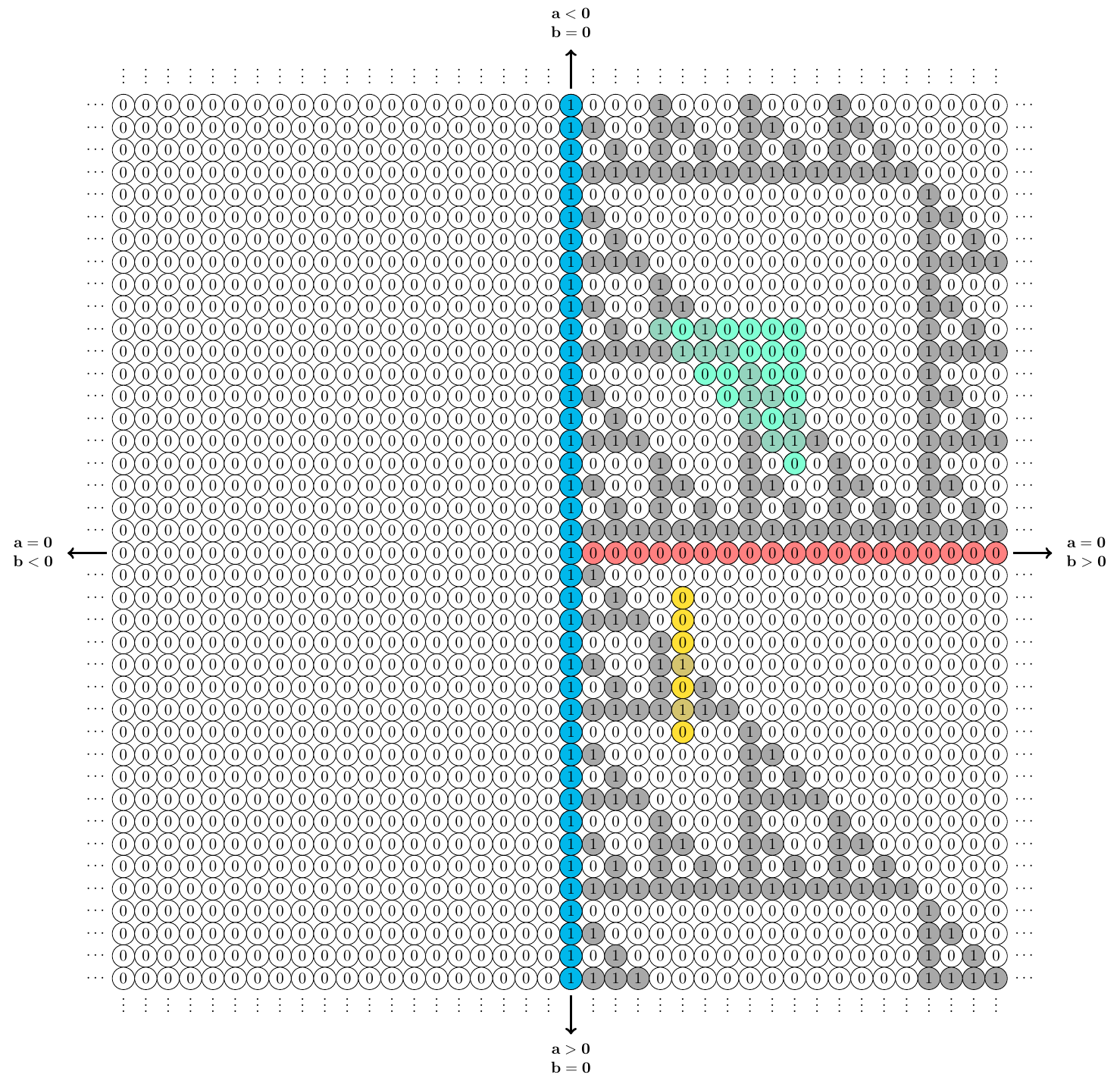}}
\caption{\normalsize The infinite Pascal matrix modulo $2$ with $\left(\binomd{a}{0}\right)_{a\in\Z}$ in blue, $\left(\binomd{0}{b}\right)_{b>0}$ in red, $\BS{7}{5}{2}$ in yellow and $\Strig{\BS{7}{10}{-10}}$ in green}\label{fig6}
\end{figure}

For any integers $k$ and $\ell$, let $\BS{n}{k}{\ell}$ be the subsequence of length $n$ of the $k$th column of the infinite Pascal matrix modulo $2$ defined by
$$
\BS{n}{k}{\ell} = \left(\binomd{\ell+j-1}{k}\right)_{1\le j\le n} = \left(\binomd{\ell}{k},\binomd{\ell+1}{k},\ldots\ldots,\binomd{\ell+n-1}{k}\right).
$$
For instance, the sequence $\BS{7}{5}{2}=\left(\binomd{j+1}{5}\right)_{1\le j\le 7}=(0001010)$ appears in yellow in Figure~\ref{fig6}. Since we have the local rule \eqref{eq13} in the infinite Pascal matrix modulo $2$, it is straightforward to obtain the following

\begin{prop}\label{prop18}
Let $k$ and $\ell$ be two integers and let $n$ be a positive integer. Then,
$$
\partial^{i}\BS{n}{k}{\ell} = \BS{n-i}{k-i}{\ell} = \left(\binomd{\ell+j-1}{k-i}\right)_{1\le j\le n-i},
$$
for all $i\in\{0,\ldots,n-1\}$ and
$$
\Strig{\BS{n}{k}{\ell}} = \left(\binomd{\ell+j-i}{k+1-i}\right)_{1\le i\le j\le n}.
$$
\end{prop}

\noindent For instance, the Steinhaus triangle $\Strig{\BS{7}{10}{-10}}=\Strig{(0000110)}$ appears in green in Figure~\ref{fig6}.

We are now ready to give explicit bases of $\RST{n}$ using Corollary~\ref{cor1} with the binary sequences $\BS{n}{k}{\ell}$.

\begin{thm}\label{thm4}
Let $n$ and $m$ be non-negative integers such that $m=\left\lfloor\frac{n}{3}\right\rfloor + \delta_{1,(n\bmod{3})}$. For any integers $\ell_0,\ldots,\ell_{m-1}$, the set $\left\{ \rho\left(\Strig{\BS{n}{0}{\ell_0}}\right), \ldots , \rho\left(\Strig{\BS{n}{m-1}{\ell_{m-1}}}\right) \right\}$ is a basis of $\RST{n}$. Moreover, we have
\begin{equation}\label{eq17}
\rho\left(\Strig{\BS{n}{k}{\ell_k}}\right) = \left( \left(\binom{\ell_k+j-i}{k+1-i} + \binom{\ell_k+n-j}{k+i-j} + \binom{\ell_k+i-1}{k+j-n}\right) \tpmod{2} \right)_{1\le i\le j\le n},
\end{equation}
for all $k\in\{0,\ldots,m-1\}$.
\end{thm}

\begin{proof}
By Proposition~\ref{prop18}, we know that $\partial^{k}\BS{n}{k}{\ell_k} = \BS{n-k}{0}{\ell_k} = \left(\binomd{\ell_k+j-1}{0}\right)_{1\le j\le n-k}=(1)_{n-k}$, for all $k\in\{0,\ldots,m-1\}$. It follows from Corollary~\ref{cor1} that the set
$$
\left\{\rho\left(\Strig{\BS{n}{k}{\ell_k}}\right)\ \middle|\ k\in\{0,\ldots,m-1\}\right\}
$$
is a basis of $\RST{n}$. Moreover, the formula for $\rho\left(\Strig{\BS{n}{k}{\ell_k}}\right)$ from \eqref{eq17} directly follows from Proposition~\ref{prop18} and Proposition~\ref{prop16}.
\end{proof}

\begin{rem}
For any integer $l_0$, we have $\rho\left(\Strig{\BS{n}{0}{\ell_0}}\right)=\rho\left(\Strig{(1)_n}\right)=U_n$.
\end{rem}

For instance, for $n=10$ and $\ell_0=\ell_1=\ell_2=\ell_3=0$, we obtain the basis given in Table~\ref{tab1}.
\begin{table}[htbp]
$$
\begin{array}{|c|c|c|}
\hline\vspace{-2ex} & & \\
k & \BS{10}{k}{0} & \rho\left(\Strig{\BS{10}{k}{0}}\right) \\[1.5ex]
\hline
0 & (1111111111) & \iStrig{0}=\Strig{(0111111110)} \\
\hline
1 & (0101010101) & \iStrig{1}=\Strig{(1001010111)} \\
\hline
2 & (0011001100) & \iStrig{2}=\Strig{(0001001000)} \\
\hline
3 & (0001000100) & \iStrig{3}=\Strig{(0010001100)} \\
\hline
\end{array}
$$
\caption{A basis of $\RST{10}$}\label{tab1}
\end{table}
All the rotationally symmetric Steinhaus triangles of size $10$ are depicted in Figure~\ref{fig14}, where the elements of the basis $\left\{\iStrig{0},\iStrig{1},\iStrig{2},\iStrig{3}\right\}$ are in red and for every $\Strig{}\in\RST{10}$, the coordinate vector $(x_0,x_1,x_2,x_3)$ of $\Strig{}=x_0\iStrig{0}+x_1\iStrig{1}+x_2\iStrig{2}+x_3\iStrig{3}$ is given.

\begin{figure}[htbp]
\centerline{
\includegraphics[width=\textwidth]{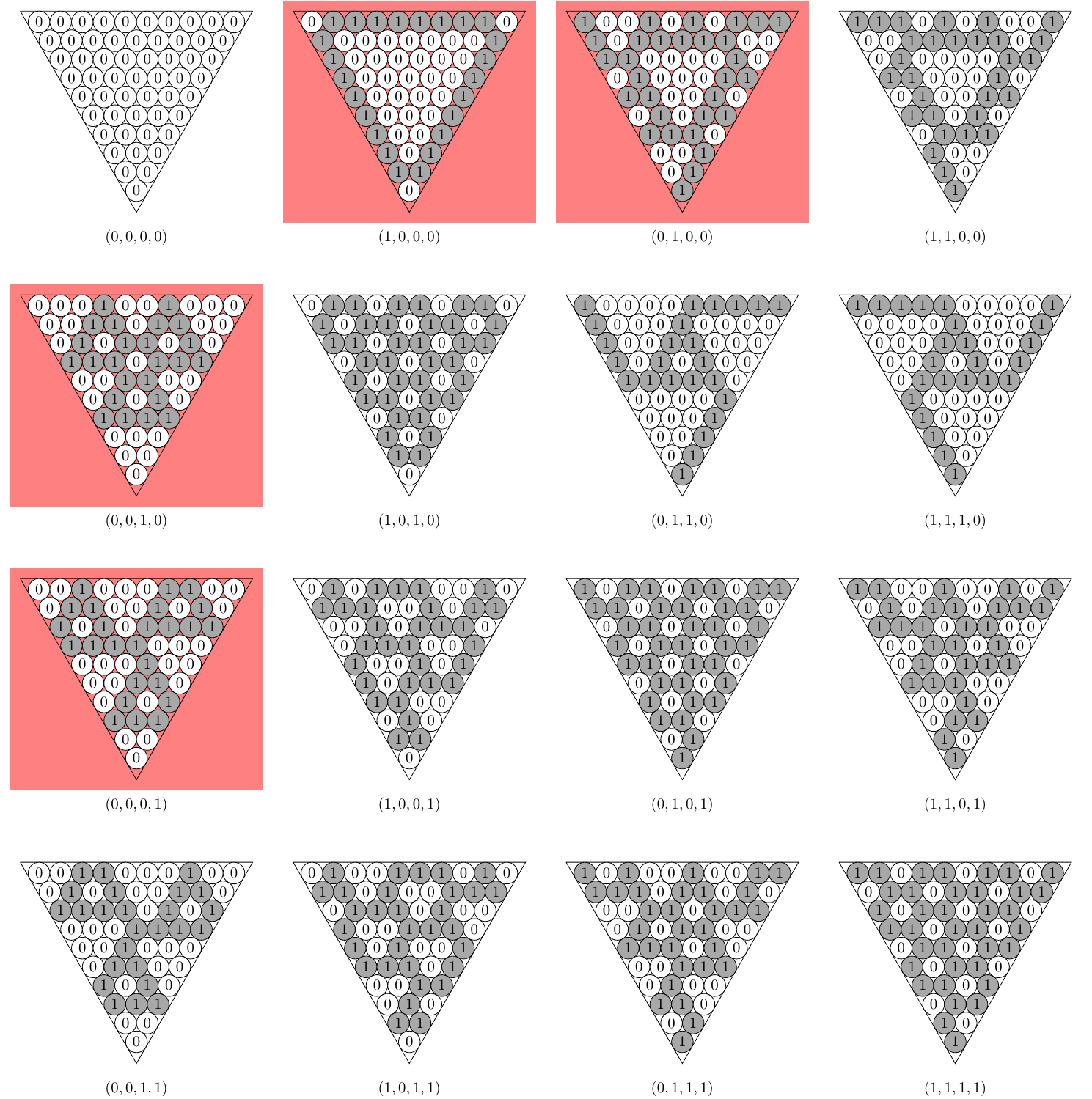}}
\caption{The $16$ triangles of $\RST{10}$ where the $4$ red triangles form a basis}\label{fig14}
\end{figure}

%%%%%%%%%%%%%%%%
%%%%%%%%%%%%%%%%
\section{Horizontally symmetric Steinhaus triangles}\label{sec:3}
%%%%%%%%%%%%%%%%
%%%%%%%%%%%%%%%%

In this section, we characterize the horizontally symmetric Steinhaus triangles and we give a generating index set of $\HST{n}$. This allows us to obtain bases of $\HST{n}$.

%%%%%%%%%%%%%%%%
\subsection{Characterizations of \texorpdfstring{$\HST{n}$}{HST(n)}}
%%%%%%%%%%%%%%%%

A binary sequence $S=\left(a_j\right)_{1\le j\le n}$ is said to be {\em symmetric} if $a_{n-j+1}=a_j$ for all $j\in\{1,\ldots,n\}$. For instance, the sequence $(010010010)$ is symmetric. As shown in the following result, the symmetry is preserved under the derivation process.

\begin{prop}\label{prop1}
The binary sequence $S$ is symmetric if and only if $\partial S$ is symmetric and $\sigma_2(\partial S)=0$.
\end{prop}

\begin{proof}
Let $S=\left(a_j\right)_{1\le j\le n}$ and $\partial S=\left(b_j\right)_{1\le j\le n-1}$. By the definition of $\partial S$, we have
\begin{equation}\label{eq1}
b_{i} + b_{(n-1)-i+1} \begin{array}[t]{l}
 = b_{i} + b_{n-i} \\
 \equiv (a_{i}+a_{i+1})+(a_{n-i}+a_{n-i+1}) \pmod{2} \\
 = (a_{i}+a_{n-i+1})+(a_{i+1}+a_{n-(i+1)+1}), \\
\end{array}
\end{equation}
for all $i\in\{1,\ldots,n-1\}$. Moreover, the sum $\sigma(\partial S)$ satisfies
\begin{equation}\label{eq2}
\sigma(\partial S) = \sum_{i=1}^{n-1}b_i \equiv \sum_{i=1}^{n-1}(a_{i}+a_{i+1}) \equiv a_1+a_n \pmod{2}.
\end{equation}
First, if $S$ is symmetric, we deduce from \eqref{eq1} that $\partial S$ is also symmetric and from \eqref{eq2} that $\sigma(\partial S)$ is even. Conversely, if we suppose that $\partial S$ is symmetric of even sum $\sigma\left(\partial S\right)$, we know from \eqref{eq2} that $a_1=a_n$. Using this equality and \eqref{eq1}, we can prove, by induction on $i$, that $a_i=a_{n-i+1}$, for all $i\in\{1,\ldots,n\}$. This completes the proof.
\end{proof}

%\begin{rem}
%It is natural to ask if there exists a similar result for the case where $\partial S$ is symmetric and $\sigma(\partial S)$ is odd. It is known that the binary sequence $S=(a_j)_{1\le j\le n}$ is antisymmetric, i.e., $a_{n-j+1}\equiv a_j+1\bmod{2}$ for all $i\in\{1,\ldots,n\}$, if and only if $\partial S$ is symmetric and $\sigma_2\left(\partial S\right)=1$. Note that antisymmetric binary sequences only exist for odd lengths.
%\end{rem}

It follows that the horizontal symmetry of a Steinhaus triangle is only related to the symmetry of its first row.

\begin{prop}\label{prop3}
The Steinhaus triangle $\Strig{S}$ is horizontally symmetric if and only if the sequence $S$ is symmetric.
\end{prop}

\begin{proof}
First, it is clear that if $\Strig{S}$ is horizontally symmetric, then $S$ is a symmetric sequence. Conversely, if we suppose that $S$ is symmetric, we know from Proposition~\ref{prop1} that the iterated derived sequences $\partial^{i}S$ are symmetric for all $i\in\{0,\ldots,n-1\}$. Therefore, the Steinhaus triangle $\Strig{S}$ is horizontally symmetric.
\end{proof}

We now show that the horizontal symmetry of a Steinhaus triangle only depends on the values of the middle terms of its rows of odd lengths.

\begin{prop}\label{prop11}
The Steinhaus triangle $\left(a_{i,j}\right)_{1\le i\le j\le n}$ of size $n$ is horizontally symmetric if and only if $a_{n-2i,n-i}=0$ for all $i\in\left\{0,\ldots,\left\lfloor\frac{n}{2}\right\rfloor-1\right\}$.
\end{prop}

The proof is based on the following lemma which is straightforward from the definition of a symmetric sequence.

\begin{lem}\label{lem6}
Let $S=(a_1,a_2,\ldots,a_n)$ be a symmetric binary sequence of length $n$. Then,
$$
\sigma_2(S) = \left\{\begin{array}{ll}
0 & \text{if } n \text{ is even},\\
a_{\frac{n+1}{2}} & \text{if } n \text{ is odd}.
\end{array}\right.
$$
\end{lem}

\begin{proof}[Proof of Proposition~\ref{prop11}]
First, suppose that the Steinhaus triangle $\Strig{S}=\left(a_{i,j}\right)_{1\le i\le j\le n}$ of size $n$ is horizontally symmetric. Then, since the iterated derived sequences $\partial^{i}S=\left(a_{i+1,j}\right)_{i+1\le j\le n}$ are symmetric for all $i\in\{0,\ldots,n-1\}$, we know from Proposition~\ref{prop1} that $\sigma_2\left(\partial^{i}S\right)=0$ for all $i\in\{1,\ldots,n-1\}$. Moreover, for any $i\in\{1,2,\ldots,n\}$, the sequence $\partial^{n-i}S$ is of length $i$. It follows from Lemma~\ref{lem6} that $a_{n-2i,n-i}=\sigma_2\left(\partial^{n-(2i+1)}S\right)=0$, for all $i\in\left\{0,\ldots,\left\lfloor\frac{n}{2}\right\rfloor-1\right\}$.

Conversely, suppose that $a_{n-2i,n-i}=0$, for all $i\in\left\{0,\ldots,\left\lfloor\frac{n}{2}\right\rfloor-1\right\}$. We proceed by induction on $n$. For $n=1$, the result is clear since any Steinhaus triangle of size $n=1$ is horizontally symmetric. Suppose that the result is true for any Steinhaus triangle of size strictly less than $n$. We consider the subtriangle $\Strig{\partial S}=\left(a_{i+1,j+1}\right)_{1\le i\le j\le n-1}$ of size $n-1$. Since the identities $a_{1+(n-1)-2i,1+(n-1)-i}=0$ hold for all $i\in\left\{0,\ldots,\left\lfloor\frac{n-1}{2}\right\rfloor-1\right\}$, it follows that $\Strig{\partial S}$ is horizontally symmetric by the induction hypothesis. Therefore, the sequence $\partial^{i}S$ is symmetric for all $i\in\{1,\ldots,n-1\}$. If $n$ is odd, then $\partial S$ is symmetric and $\sigma_2\left(\partial S\right)=0$ by Lemma~\ref{lem6}. Otherwise, if $n$ is even, then again by Lemma~\ref{lem6} we have that $\partial S$ is symmetric and $\sigma_2\left(\partial S\right)=a_{2,\frac{n}{2}+1}=0$. It follows from Proposition~\ref{prop1} that $S$ is symmetric. Therefore, in any case, the Steinhaus triangle $\Strig{S}$ is horizontally symmetric by Proposition~\ref{prop3}.
\end{proof}

%%%%%%%%%%%%%%%%
\subsection{Generating index set of \texorpdfstring{$\HST{n}$}{HST(n)}}
%%%%%%%%%%%%%%%%

\begin{prop}\label{prop4}
Let $n$ be a non-negative integer. The set
$$
G_H := \left\{(1,j)\ \middle|\ j\in\left\{1,\ldots,\left\lceil\frac{n}{2}\right\rceil\right\}\right\}
$$
is a generating index set of $\HST{n}$.
\end{prop}

\begin{proof}
From Proposition~\ref{prop3}, we deduce that $\HST{n}$ is isomorphic to the vector space of symmetric binary sequences of length $n$. Obviously, a symmetric sequence of length $n$ is entirely determined by its $\left\lceil\frac{n}{2}\right\rceil$ first terms.
\end{proof}

Since the dimension of $\HST{n}$ corresponds to the cardinality of the generating index set $G_H$, it is straightforward to obtain the following

\begin{cor}\label{cor13}
For all non-negative integers $n$, we have $\dim\HST{n} = \left\lceil\frac{n}{2}\right\rceil$.
\end{cor}

%%%%%%%%%%%%%%%%
\subsection{Bases of \texorpdfstring{$\HST{n}$}{HST(n)}}
%%%%%%%%%%%%%%%%

Let $n$ be a positive integer. For any positive integer $k\in\{1,\ldots,n\}$, we denote by $\ES{n}{k}$ the binary sequence of length $n$ consisting only of zeroes, except at position $k$. In other words, we have
$$
\ES{n}{k} = \left(\binomd{0}{j-k}\right)_{1\le j\le n},
$$
for all $k\in\{1,\ldots,n\}$. Since we have the local rule of the infinite Pascal matrix modulo $2$, we know that
$$
\Strig{\ES{n}{k}} = \left(\binomd{i-1}{j-k}\right)_{1\le i\le j\le n},
$$
for all integers $k\in\{1,\ldots,n\}$.

\begin{prop}
Let $n$ be a positive integer. Then, the set $\left\{\iStrig{1},\ldots,\iStrig{\left\lceil\frac{n}{2}\right\rceil}\right\}$ is a basis of $\HST{n}$, where
$$
\iStrig{k} = \Strig{\left(\ES{n}{k} + \ES{n}{n-k+1}\right)} = \left(\left(\binom{i-1}{j-k}+\binom{i-1}{j-n+k-1}\right)\tpmod{2}\right)_{1\le i\le j\le n},
$$
for all $k\in\left\{1,\ldots,\left\lfloor\frac{n}{2}\right\rfloor\right\}$, and when $n$ is odd
$$
\iStrig{\frac{n+1}{2}}=\Strig{\ES{n}{\frac{n+1}{2}}} = \left(\binomd{i-1}{j-\frac{n+1}{2}}\right)_{1\le i\le j\le n}.
$$
\end{prop}

\begin{proof}
First, by definition, it is clear that $\iStrig{k}\in\HST{n}$, for all $k\in\left\{1,\ldots,\left\lceil\frac{n}{2}\right\rceil\right\}$. Now we consider the set $G_H := \left\{(1,j)\ \middle|\ j\in\left\{1,\ldots,\left\lceil\frac{n}{2}\right\rceil\right\}\right\}$. Since for all $k\in\left\{1,\ldots,\left\lceil\frac{n}{2}\right\rceil\right\}$,
$$
\pi_{G_H} \left(\iStrig{k}\right) = \ES{\left\lceil\frac{n}{2}\right\rceil}{k},
$$
it follows that the set $\left\{ \pi_{G_H}\left(\iStrig{1}\right) , \ldots , \pi_{G_H}\left(\iStrig{\left\lceil\frac{n}{2}\right\rceil}\right) \right\}$ is a basis of $\Zd^{|G_H|}$. Moreover, since $G_H$ is a generating index set of $\HST{n}$ by Proposition~\ref{prop4}, we conclude that the set $\left\{\iStrig{1},\ldots,\iStrig{\left\lceil\frac{n}{2}\right\rceil}\right\}$ is a basis of $\HST{n}$.
\end{proof}

At the end of this section, we show that the generating index set $G_H$ also permits us to obtain a basis from the sequences $\BS{n}{k}{\ell}=\left(\binomd{\ell+j-1}{k}\right)_{1\le j\le n}$ introduced in Section~\ref{sec:4}.

\begin{lem}\label{lem3}
Let $k$ and $n$ be two positive integers of the same parity. Then, the $n$-length sequence $\BS{n}{k}{\ell}$ is symmetric for $\ell=\frac{k-n}{2}$. Moreover, the $k+2$ middle terms of $\BS{n}{k}{\frac{k-n}{2}}$ are of the form $(\ast\cdots\ast1\underbrace{00\cdots00}_{k}1\ast\cdots\ast)$, where $\ast$ stands for any element of $\Zd$.
\end{lem}

\begin{proof}
Suppose that $k$ and $n$ are of the same parity; i.e., $\ell=\frac{k-n}{2}$ is an integer. Let $\BS{n}{k}{\ell}=\left(a_{j}\right)_{1\le j\le n}$ with $a_j=\binomd{\ell+j-1}{k}$, for all $j\in\{1,\ldots,n\}$. Since $\binomd{a}{b}=\binomd{b-a-1}{b}$, for any integers $a$ and $b$ such that $b\ge0$, it follows that
{\small
$$
a_{n-j+1}  = \binomd{\frac{k-n}{2}+(n-j+1)-1}{k} = \binomd{k-\left(\frac{k-n}{2}+n-j\right)-1}{k} = \binomd{\frac{k-n}{2}+j-1}{k} = a_j,
$$}
for all $j\in\{1,\ldots,n\}$. Therefore, the sequence $\BS{n}{k}{\ell}$ is symmetric for $\ell=\frac{k-n}{2}$. Moreover, since $k\ge1$, the $k+2$ middle terms of $\BS{n}{k}{\ell}$ are $\binomd{j}{k}$, for $j\in\{-1,\ldots,k\}$, where $\binomd{-1}{k}=\binomd{k}{k}=1$ and $\binomd{j}{k}=0$, for all $j\in\{0,\ldots,k-1\}$. This completes the proof.
\end{proof}

\begin{thm}\label{thm5}
Let $n$ be a non-negative integer. The set $\left\{\iStrig{1},\ldots,\iStrig{\left\lceil\frac{n}{2}\right\rceil}\right\}$ is a basis of $\HST{n}$, where
$$
\iStrig{k} = \Strig{\BS{n}{n-2k}{-k}} = \left(\binomd{-k+j-i}{n-2k+1-i}\right)_{1\le i\le j\le n},
$$
for all $k\in\left\{1,\ldots,\left\lfloor\frac{n}{2}\right\rfloor\right\}$ and $\iStrig{\frac{n+1}{2}}=\Strig{(1)_n}$ when $n$ is odd. 
\end{thm}

\begin{proof}
First, we know from Lemma~\ref{lem3} that the sequence $\BS{n}{n-2k}{-k}$ is symmetric for all $k\in\left\{1,\ldots,\left\lfloor\frac{n}{2}\right\rfloor\right\}$. Obviously, when $n$ is odd, the constant sequence $(1)_n$ is also symmetric. Therefore by Proposition~\ref{prop3} and for all $k\in\left\{1,\ldots,\left\lceil\frac{n}{2}\right\rceil\right\}$, we have $\iStrig{k}\in\HST{n}$. Now, we consider the set $G_H := \left\{(1,j)\ \middle|\ j\in\left\{1,\ldots,\left\lceil\frac{n}{2}\right\rceil\right\}\right\}$. By Lemma~\ref{lem3} again, we know that the $n-2k+2$ middle terms of $\BS{n}{n-2k}{-k}$ are of the form $(\underbrace{\ast\cdots\ast}_{k-1}1\underbrace{00\cdots00}_{n-2k}1\underbrace{\ast\cdots\ast}_{k-1})$ and thus for all $k\in\left\{1,\ldots,\left\lfloor\frac{n}{2}\right\rfloor\right\}$, we have
$$
\pi_{G_H}\left(\iStrig{k}\right) = (\underbrace{\ast,\cdots,\ast}_{k-1},1,\underbrace{0,\cdots,0}_{\left\lceil\frac{n}{2}\right\rceil-k}).
$$
Moreover, $\pi_{G_H}\left(\iStrig{\frac{n+1}{2}}\right)=\pi_{G_H}\left(\Strig{(1)_n}\right)=(1,1,\ldots,1)$ when $n$ is odd. Therefore, the set $\left\{ \pi_{G_H}\left(\iStrig{k}\right) \ \middle|\ k\in\left\{1,\ldots,\left\lceil\frac{n}{2}\right\rceil\right\} \right\}$ is a basis of $\Zd^{|G_H|}$. Finally, since $G_H$ is a generating index set of $\HST{n}$ by Proposition~\ref{prop4}, we conclude that the set $\left\{\iStrig{1},\ldots,\iStrig{\left\lceil\frac{n}{2}\right\rceil}\right\}$ is a basis of $\HST{n}$.
\end{proof}

\begin{rem}
When $n$ is even, we have $\iStrig{\frac{n}{2}}=\Strig{\BS{n}{0}{-\frac{n}{2}}}=\Strig{(1)_n}$. Therefore, $\iStrig{\left\lceil\frac{n}{2}\right\rceil}=\Strig{(1)_n}$ for all integers $n$.
\end{rem}

For instance, for $n=7$, we obtain the basis
$$
\begin{array}{ll}
\iStrig{1} = \Strig{\BS{7}{5}{-1}} = \Strig{(1000001)},
&
\iStrig{2} = \Strig{\BS{7}{3}{-2}} = \Strig{(0100010)},
\\
\iStrig{3} = \Strig{\BS{7}{1}{-3}} = \Strig{(1010101)},
&
\iStrig{4} = \Strig{(1111111)}. \\
\end{array}
$$
All the horizontally symmetric Steinhaus triangles of size $7$ are depicted in Figure~\ref{fig8}, where the elements of the basis $\left\{\iStrig{1},\iStrig{2},\iStrig{3},\iStrig{4}\right\}$ are in red and, for every $\Strig{}\in\HST{7}$, the coordinate vector $(x_1,x_2,x_3,x_4)$ of $\Strig{}=x_1\iStrig{1}+x_2\iStrig{2}+x_3\iStrig{3}+x_4\iStrig{4}$ is given.

\begin{figure}[htbp]
\centerline{\includegraphics[width=\textwidth]{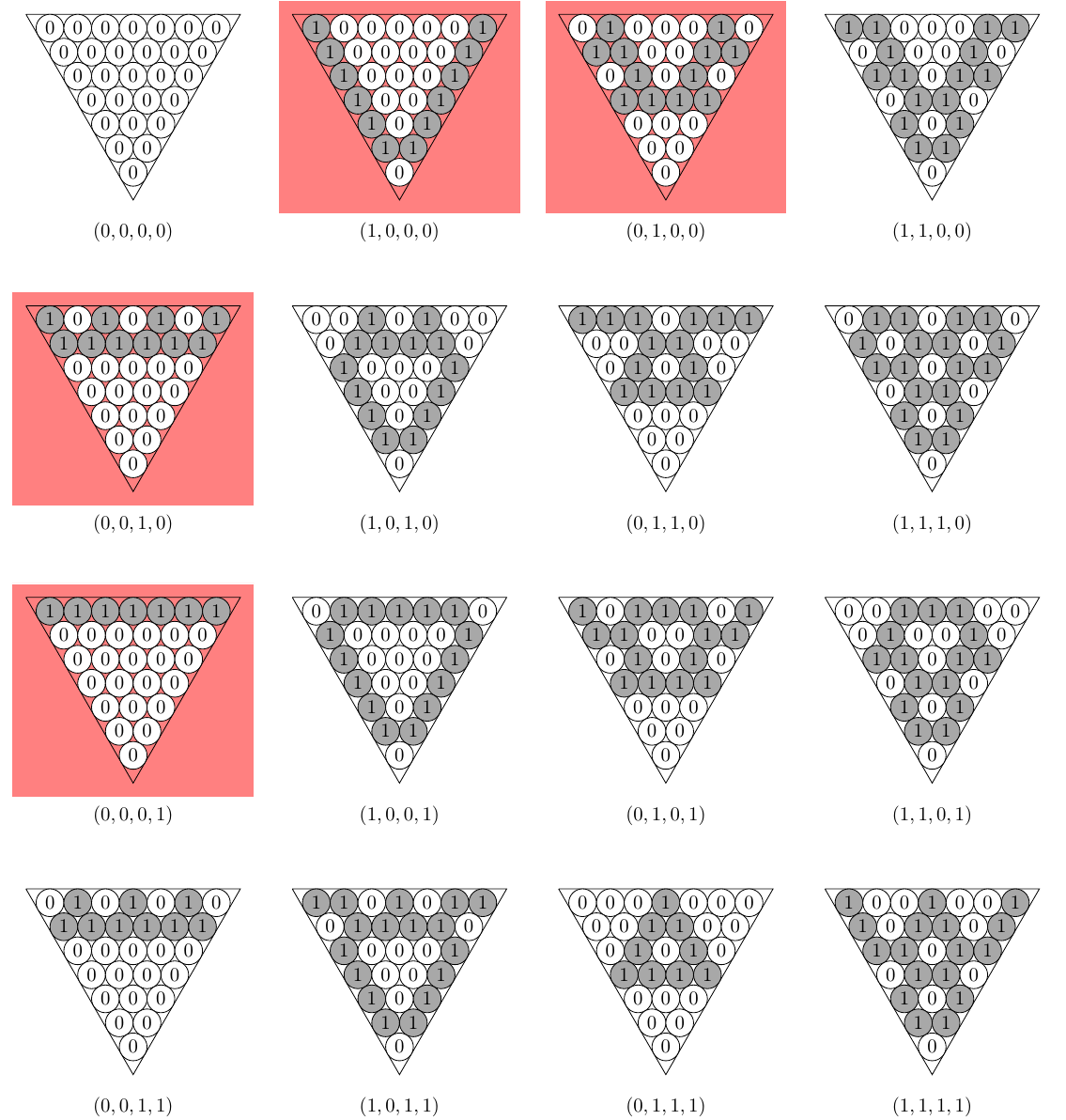}}
\caption{The $16$ triangles of $\HST{7}$ where the $4$ red triangles form a basis}\label{fig8}
\end{figure}

%%%%%%%%%%%%%%%%
%%%%%%%%%%%%%%%%
\section{Dihedrally symmetric Steinhaus triangles}\label{sec:5}
%%%%%%%%%%%%%%%%
%%%%%%%%%%%%%%%%

In this section, after characterizing dihedrally symmetric Steinhaus triangles, we determine generating index sets and a basis of $\DST{n}$.

%%%%%%%%%%%%%%%%
\subsection{Characterizations of \texorpdfstring{$\DST{n}$}{DST(n)}}
%%%%%%%%%%%%%%%%

We begin by showing that the dihedral symmetry of a Steinhaus triangle is only related to the symmetry of its first row and of its right and left sides.

\begin{prop}\label{prop12}
The Steinhaus triangle $\Strig{}$ is dihedrally symmetric if and only if two of the three Steinhaus triangles $\Strig{}$, $r\left(\Strig{}\right)$ and $r^2\left(\Strig{}\right)$ are horizontally symmetric.
\end{prop}

\begin{proof}
This follows directly from the fact that the automorphisms $h$ and $hr$ are generators of $D_3$.
First, if $\Strig{}$ is dihedrally symmetric, then $\Strig{}=r\left(\Strig{}\right)=r^2\left(\Strig{}\right)$. Since a dihedrally symmetric triangle is also horizontally symmetric, the result follows.

Suppose now that the Steinhaus triangles $\Strig{}$ and $r\left(\Strig{}\right)$ are horizontally symmetric. It follows that $h\left(\Strig{}\right)=\Strig{}$ and $hr\left(\Strig{}\right)=r\left(\Strig{}\right)$. Since $r^3=h^2=(hr)^2=id_{\ST{n}}$, we have $rhr=h$. Since both $\Strig{}$ and $r\left(\Strig{}\right)$ are horizontally symmetric, we obtain
$$
\Strig{} = r^3\left(\Strig{}\right) = r^2r\left(\Strig{}\right) = r^2hr\left(\Strig{}\right) = r\left(rhr\left(\Strig{}\right)\right) = rh\left(\Strig{}\right) = r\left(\Strig{}\right).
$$
We conclude that the Steinhaus triangle $\Strig{}$ is dihedrally symmetric since it is horizontally symmetric and rotationally symmetric.
If the two horizontally symmetric Steinhaus triangles are $r^2\left(\Strig{}\right)$ and $\Strig{}$, or $r\left(\Strig{}\right)$ and $r^2\left(\Strig{}\right)$, we can show in the same way that the triangles $r^2\left(\Strig{}\right)$ or $r\left(\Strig{}\right)$ are dihedrally symmetric, since $r\left(r^2\left(\Strig{}\right)\right)=\Strig{}$ and $r\left(r\left(\Strig{}\right)\right)=r^2\left(\Strig{}\right)$, respectively. In any case, we observe that the Steinhaus triangle $\Strig{}=r\left(\Strig{}\right)=r^2\left(\Strig{}\right)$ is dihedrally symmetric.
\end{proof}

\begin{cor}
The Steinhaus triangle $\Strig{S}=\left(a_{i,j}\right)_{1\le i\le j\le n}$ is dihedrally symmetric if and only if two of the three sequences, its first row $\left(a_{1,j}\right)_{1\le j\le n}$, its right side $\left(a_{j,n}\right)_{1\le j\le n}$ or its left side $\left(a_{j,j}\right)_{1\le j\le n}$, are symmetric.
\end{cor}

\begin{proof}
This follows directly from Proposition~\ref{prop12}, since, by Proposition~\ref{prop3}, we know that the Steinhaus triangles $\Strig{S}$, $r\left(\Strig{S}\right)$ and $r^2\left(\Strig{S}\right)$ are horizontally symmetric if and only if the sequences $\left(a_{1,j}\right)_{1\le j\le n}$, $\left(a_{j,n}\right)_{1\le j\le n}$ and $\left(a_{j,j}\right)_{1\le j\le n}$ are symmetric, respectively.
\end{proof}

Proposition~\ref{prop12} also enables us to show that the dihedral symmetry of a Steinhaus triangle only depends on the values of the middle terms of its rows, its columns or its diagonals of odd lengths.

\begin{cor}\label{cor8}
The Steinhaus triangle $\Strig{S}=\left(a_{i,j}\right)_{1\le i\le j\le n}$ is dihedrally symmetric if and only if two of the three sets $\left\{a_{n-2i,n-i}\middle| i\in\left\{0,\ldots,\left\lfloor\frac{n}{2}\right\rfloor-1\right\}\right\}$, $\left\{a_{i,2i-1}\middle| i\in\left\{1,\ldots,\left\lfloor\frac{n}{2}\right\rfloor\right\}\right\}$  and $\left\{a_{i,n-i+1}\middle| i\in\left\{1,\ldots,\left\lfloor\frac{n}{2}\right\rfloor\right\}\right\}$ are sets of zeroes.
\end{cor}

\begin{proof}
From Proposition~\ref{prop11}, we know that the Steinhaus triangles $\Strig{S}$, $r\left(\Strig{S}\right)$ and $r^2\left(\Strig{S}\right)$ are horizontally symmetric if and only if $a_{n-2i,n-i}=0$ for all $i\in\left\{0,\ldots,\left\lfloor\frac{n}{2}\right\rfloor-1\right\}$, $a_{i,2i-1}=0$ for all $i\in\left\{1,\ldots,\left\lfloor\frac{n}{2}\right\rfloor\right\}$ and $a_{i,n-i+1}=0$ for all $i\in\left\{1,\ldots,\left\lfloor\frac{n}{2}\right\rfloor\right\}$, respectively. Finally, by Proposition~\ref{prop12}, the Steinhaus triangle $\Strig{S}$ is dihedrally symmetric if and only if two of the three Steinhaus triangles $\Strig{S}$, $r\left(\Strig{S}\right)$ and $r^2\left(\Strig{S}\right)$ are horizontally symmetric. This completes the proof.
\end{proof}

For any positive integer $n\ge3$, by definition of $\DST{n}$ and $\He$, it is clear that $\He\left(\DST{n}\right)\subset\DST{n-3}$. The precise relationship between a dihedrally symmetric Steinhaus triangle $\Strig{S}$ and its subtriangle $\He\left(\Strig{S}\right)$ is given in the following

\begin{prop}\label{prop7}
Let $S$ be a binary sequence of length $n\ge 3$. The Steinhaus triangle $\Strig{S}$ is dihedrally symmetric if and only if $\He\left(\Strig{S}\right)=\Strig{S'}$ is dihedrally symmetric, $\sigma_2\left(S'\right)=0$ and $S=(0)\cdot\left(\int_{i,x}S'\right)\cdot(0)$, for some $i\in\{1,\ldots,n-2\}$ and some $x\in\Zd=\{0,1\}$.
\end{prop}

Proposition~\ref{prop7} appears in \cite{Barbe:2001aa} in a more general context. For the convenience of the reader, a proof is given here.

\begin{proof}
Let $\Strig{S}=(a_{i,j})_{1\le i\le j\le n}\in\ST{n}$ such that $\Strig{S'}=\He(\Strig{S})=(a_{1+i,2+j})_{1\le i\le j\le n-3}\in\DST{n-3}$. From Proposition~\ref{prop5}, Proposition~\ref{prop3} and Proposition~\ref{prop1}, we deduce that the Steinhaus triangle $\Strig{S}$ is dihedrally symmetric if and only if $a_{1,1}=a_{1,n}=\sigma_2(S')$, $\partial S$ is symmetric and $\sigma_2\left(\partial S\right)=0$. Since $\Strig{S'}\in\DST{n-3}$, the sequence $S'$ is symmetric. Therefore the sequence $\partial S$ is symmetric if and only if $a_{2,1}=a_{2,n}$. Moreover, it is clear that $\sigma\left(\partial S\right) \equiv \left(\sigma\left(S'\right)+a_{2,1}+a_{2,n}\right)\tpmod{2}$. We claim that $a_{1,1}=a_{1,n}=\sigma_2(S')$, $\partial S$ is symmetric and $\sigma_2\left(\partial S\right)=0$ if and only if $a_{1,1}=a_{1,n}=\sigma_2(S')=0$. First, suppose that $a_{1,1}=a_{1,n}=\sigma_2(S')$, $\partial S$ is symmetric and $\sigma_2\left(\partial S\right)=0$. Since $\partial S$ is symmetric, it follows that $a_{2,1}=a_{2,n}$ and thus $\sigma_2\left(S'\right)=\sigma_2\left(\partial S\right)=0$. Conversely, suppose that $a_{1,1}=a_{1,n}=\sigma_2(S')=0$. Since $S'$ is symmetric and $\sigma_2\left(S'\right)=0$, we know from Proposition~\ref{prop1} that its antiderived sequence $\int_{i,x}S'$ is symmetric too. Moreover, we have $a_{1,1}=a_{1,n}$. It follows that the sequence $S$ is symmetric and, by Proposition~\ref{prop1} again, we obtain that the sequence $\partial S$ is symmetric and $\sigma_2\left(\partial S\right)=0$. This completes the proof.
\end{proof}

The set of dihedrally symmetric Steinhaus triangles $\Strig{S}$ of size $n$ with $\sigma(S)$ even is denoted by $\DSTz{n}$. It is clear that $\DSTz{n}$ is a linear subspace of $\DST{n}$. Moreover, the vector space $\DST{n}$ can be expressed in terms of its linear subspace $\DSTz{n}$.

\begin{prop}\label{prop9}
Let $n$ be a non-negative integer. Then, we have
$$
\DST{n} = \left\{\begin{array}{ll}
\DSTz{n} & \text{for }n\text{ even},\\
\DSTz{n} \sqcup \left(\DSTz{n} + U_n\right) & \text{for }n\text{ odd},\\
\end{array}\right.
$$
where $\sqcup$ is the disjoint union of two sets.
\end{prop}

\begin{proof}
For any symmetric sequence $S=\left(a_j\right)_{1\le j\le n}$ of length $n$, we know from Lemma~\ref{lem6} that $\sigma_2(S)=0$ when $n$ is even and $\sigma_2(S)=a_{\frac{n+1}{2}}$ when $n$ is odd. For $n$ even, it follows that $\DST{n}=\DSTz{n}$. If $n$ is odd, then we consider $U_n=\rho\left((1)_n\right)$. It is clear that $U_n\in\DST{n}$ and is generated from a sequence of odd sum, for all odd numbers $n$. Since $\Strig{S}\in\DST{n}$ with $\sigma_2(S)=1$ if and only if $\Strig{S}+U_n\in\DSTz{n}$, it follows when $n$ is odd that $\DST{n} = \DSTz{n} \sqcup \left(\DSTz{n} + U_n\right)$.
\end{proof}

Using Lemma~\ref{lem6}, it is straightforward to obtain from Propositon~\ref{prop7} the following

\begin{cor}\label{cor4}
Let $S$ be a binary sequence of length $n\ge 3$. For $n$ even, the Steinhaus triangle $\Strig{S}$ is in $\DSTz{n}$ if and only if $\He(\Strig{S})=\Strig{S'}$ is in $\DSTz{n-3}$ and $S=(0)\cdot\left(\int_{i,x}S'\right)\cdot(0)$, for some $i\in\{1,\ldots,n-2\}$ and some $x\in\Zd$. For $n$ odd, the Steinhaus triangle $\Strig{S}$ is in $\DSTz{n}$ if and only if $\He(\Strig{S})=\Strig{S'}$ is in $\DSTz{n-3}$ and $S=(0)\cdot\left(\int_{\frac{n-1}{2},0}S'\right)\cdot(0)$.
\end{cor}

%%%%%%%%%%%%%%%%
\subsection{Generating index sets of \texorpdfstring{$\DST{n}$}{DST(n)}}
%%%%%%%%%%%%%%%%

We begin this subsection by giving, from Corollary~\ref{cor4}, generating index sets of $\DSTz{n}$.

\begin{thm}\label{thm7}
Let $n$ and $m$ be non-negative integers such that $m=\left\lfloor\frac{n}{6}\right\rfloor+\deltamod{4}{n}{6}$. For every integer $i\in\left\{1,\ldots,\left\lfloor\frac{n}{3}\right\rfloor\right\}$, let $j_i\in\left\{2i,\ldots,n-i\right\}$. Then, the set
$$
G_{D_0} = \left\{ (2i+1,j_{2i+1}) \ \middle|\ i\in\{0,\ldots,m-1\} \right\},
$$
when $n$ is even, or
$$
G_{D_0} = \left\{ (2i,j_{2i})\ \middle|\ i\in\{1,\ldots,m\} \right\},
$$
when $n$ is odd, is a generating index set of $\DSTz{n}$.
\end{thm}

\begin{proof}
We proceed by induction on $n$. For $n\in\{0,1,2\}$, it is clear that $\Strig{(0)_n}$ is the only element of $\DSTz{n}$. Therefore for these values of $n$, the empty set $\emptyset$ is a generating set of $\DSTz{n}$.

For $n\ge3$, suppose that the result is true for any size strictly less than $n$. Let $m=\left\lfloor\frac{n}{6}\right\rfloor + \deltamod{4}{n}{6}$.  We distinguish two cases following the parity of $n$.
\begin{case}
Suppose first that $n$ is odd.
We consider the subset $\He\left(G_{D_0}\right)\subset\Strig{(n-3)}$ defined by
$$
\He\left(G_{D_0}\right) = \left\{ (2i-1,j_{2i}-2) \ \middle|\ i\in\left\{1,\ldots, m\right\} \right\}
$$
and the linear maps $f_1$ and $f_2$ defined by
$$
\begin{array}{lllll}
f_1 : \DSTz{n} \longrightarrow \DSTz{n-3} & \text{by} & f_1\left(\Strig{S}\right) = \He\left(\Strig{S}\right) & \text{and} \\[1.5ex]
f_2 : \DSTz{n-3} \longrightarrow \Zd^m & \text{by} & f_2\left(\Strig{S'}\right) = \pi_{\He\left(G_{D_0}\right)}\left(\Strig{S'}\right). \\
\end{array}
$$
Then, for any $(a_{i,j})_{1\le i\le j\le n}\in\DSTz{n}$, we have
$$
f_2f_1\left((a_{i,j})_{1\le i\le j\le n}\right) \begin{array}[t]{l}
 = f_2\left((a_{1+i,2+j})_{1\le i\le j\le n-3}\right) \\
 = \pi_{\He\left(G_{D_0}\right)}\left((a_{1+i,2+j})_{1\le i\le j\le n-3}\right) \\
 = \left(a_{2i,j_{2i}}\right)_{1\le i\le m} = \pi_{G_{D_0}}\left((a_{i,j})_{1\le i\le j\le n}\right). \\
\end{array}
$$
Therefore $f_2f_1=\pi_{G_{D_0}}$. From Corollary~\ref{cor4}, we know that $f_1$ is an isomorphism whose inverse is defined by $f_1^{-1}\left(\Strig{S'}\right)=\Strig{\left((0)\cdot\left(\int_{\frac{n-1}{2},0}S'\right)\cdot(0)\right)}$ for all $\Strig{S'}\in\DSTz{n-3}$. Moreover, since we have
$$
2(2i-1)\le j_{2i}-2\le (n-3)-(2i-1),
$$
for all $i\in\left\{1,\ldots, m\right\}$ and $m=\left\lfloor\frac{n}{6}\right\rfloor = \left\lfloor\frac{n-3}{6}\right\rfloor + \deltamod{4}{n-3}{6}$ when $n$ is odd, the set $\He\left(G_{D_0}\right)$ is a generating index set of $\DSTz{n-3}$ by the induction hypothesis. Therefore, $f_2$ is an isomorphism. Finally, since the linear map $\pi_{G_{D_0}}=f_2f_1$ is an isomorphism, in this case the set $G_{D_0}$ is a generating index set of $\DSTz{n}$.
\end{case}
\begin{case}
Suppose now that $n$ is even.
We consider the subset $\He\left(G_{D_0}\right)\subset\Strig{(n-3)}$ defined by
$$
\He\left(G_{D_0}\right) := \left\{ (2i,j_{2i+1}-2) \ \middle|\ i\in\left\{1,\ldots, m-1\right\} \right\}
$$
and the linear maps $f_1$ and $f_2$ defined by
\begin{equation*}
\resizebox{\textwidth}{!}{$
\begin{array}{lll}
f_1 : \displaystyle\DSTz{n} \longrightarrow \Zd\times\DSTz{n-3} & \text{by} & f_1\left(\Strig{S}\right) = f_1\left((a_{i,j})_{1\le i\le j\le n}\right) = \left(a_{1,j_1} , \He\left(\Strig{S}\right) \right), \\[1.5ex]
f_2 : \displaystyle\Zd\times\DSTz{n-3} \longrightarrow \Zd^m & \text{by} & f_2\left(\left(x,\Strig{S'}\right)\right) = (x)\cdot\pi_{\He\left(G_{D_0}\right)}\left(\Strig{S'}\right). \\
\end{array}
$}
\end{equation*}
Then, for any $(a_{i,j})_{1\le i\le j\le n}\in\DSTz{n}$, we have
$$
f_2f_1\left((a_{i,j})_{1\le i\le j\le n}\right) \begin{array}[t]{l}
 = f_2\left(a_{1,j_1},(a_{1+i,2+j})_{1\le i\le j\le n-3}\right) \\
 = (a_{1,j_1})\cdot \pi_{\He\left(G_{D_0}\right)}\left((a_{1+i,2+j})_{1\le i\le j\le n-3}\right) \\
 = (a_{1,j_1})\cdot \left(a_{2i+1,j_{2i+1}}\right)_{1\le i\le m-1} \\
 = \left(a_{2i+1,j_{2i+1}}\right)_{0\le i\le m-1} = \pi_{G_{D_0}}\left((a_{i,j})_{1\le i\le j\le n}\right). \\
\end{array}
$$
Therefore $f_2f_1=\pi_{G_{D_0}}$. We know from Corollary~\ref{cor4} that $f_1$ is an isomorphism whose inverse is defined by $f_1^{-1}\left(x,\Strig{S'}\right)=\Strig{\left((0)\cdot\left(\int_{j_1-1,x}S'\right)\cdot(0)\right)}$ for all $\left(x,\Strig{S'}\right)\in\Zd\times\DSTz{n-3}$. Moreover, since we have
$$
4i\le j_{2i+1}-2\le (n-3)-2i,
$$
for all $i\in\left\{1,\ldots,m-1\right\}$ and $m-1=\left\lfloor\frac{n}{6}\right\rfloor + \deltamod{4}{n}{6}-1=\left\lfloor\frac{n-3}{6}\right\rfloor$ when $n$ is even, the set $\He\left(G_{D_0}\right)$ is a generating index set of $\DSTz{n-3}$ by the induction hypothesis. Therefore $\pi_{\He\left(G_{D_0}\right)}$ and thus $f_2$ are isomorphisms. Finally, since the linear map $\pi_{G_{D_0}}=f_2f_1$ is an isomorphism, the set $G_{D_0}$ is a generating index set of $\DSTz{n}$.
\end{case}
This completes the proof.
\end{proof}

Since the dimension of $\DSTz{n}$ corresponds to the cardinality of the generating index set $G_{D_0}$, it is easy to obtain the following

\begin{cor}\label{cor10}
For all non-negative integers $n$, we have $\dim\DSTz{n}=\left\lfloor\frac{n}{6}\right\rfloor+\deltamod{4}{n}{6}$.
\end{cor}

Using Proposition~\ref{prop9} and Theorem~\ref{thm7}, we are now ready to give a generating index set of $\DST{n}$.

\begin{thm}\label{thm13}
Let $n$ and $m$ be non-negative integers such that $m=\left\lfloor\frac{n+3}{6}\right\rfloor+\deltamod{1}{n}{6}$. For every integer $i\in\left\{1,\ldots,\left\lfloor\frac{n}{3}\right\rfloor\right\}$, let $j_i\in\left\{2i,\ldots,n-i\right\}$. Then, the set
$$
G_D = \left\{ (2i+1,j_{2i+1}) \ \middle|\ i\in\{0,\ldots,m-1\} \right\},
$$
when $n$ is even, or
$$
G_D = \left\{(1,j_1)\right\}\cup\left\{ (2i,j_{2i})\ \middle|\ i\in\{1,\ldots,m-1\} \right\},
$$
when $n$ is odd, is a generating index set of $\DST{n}$.
\end{thm}

\begin{proof}
First, suppose that $n$ is even. We know, from Proposition~\ref{prop9}, that $\DST{n}=\DSTz{n}$. Moreover, since $m=\left\lfloor\frac{n+3}{6}\right\rfloor=\left\lfloor\frac{n}{6}\right\rfloor+\deltamod{4}{n}{6}$ when $n$ is even, it follows from Theorem~\ref{thm7} that
$$
G_D = G_{D_0} = \left\{ (2i+1,j_{2i+1}) \ \middle|\ i\in\{0,\ldots,m-1\} \right\}
$$
is a generating index set of $\DST{n}$.

Suppose now that $n$ is odd. From Proposition~\ref{prop9}, we know that $\DST{n}=\DSTz{n}\sqcup \left(\DSTz{n}+U_n\right)$. Therefore $\left\{(1,j_1)\right\}\cup G_{D_0}$ is a generating set of $\DST{n}$, where $G_{D_0}$ is a generating index set of $\DSTz{n}$. Moreover, since $m-1=\left\lfloor\frac{n+3}{6}\right\rfloor+\deltamod{1}{n}{6}-1=\left\lfloor\frac{n}{6}\right\rfloor$ when $n$ is odd, it follows from Theorem~\ref{thm7} that
$$
G_D = \left\{(1,j_1)\right\}\cup G_{D_0} = \left\{(1,j_1)\right\}\cup\left\{ (2i,j_{2i})\ \middle|\ i\in\{1,\ldots,m-1\} \right\}
$$
is a generating index set of $\DST{n}$. This completes the proof.
\end{proof}

\begin{cor}\label{cor5}
Let $n$ and $m$ be non-negative integers such that $m=\left\lfloor\frac{n+3}{6}\right\rfloor+\deltamod{1}{n}{6}$. When $n$ is even, the set
$$
G_D = \left\{ \left(2i+1,n-\left\lfloor\frac{n}{3}\right\rfloor\right) \ \middle|\ i\in\{0,\ldots,m-1\} \right\}
$$
or when $n$ is odd, the set
$$
G_D = \left\{\left(1,n-\left\lfloor\frac{n}{3}\right\rfloor\right)\right\}\cup\left\{ \left(2i,n-\left\lfloor\frac{n}{3}\right\rfloor\right)\ \middle|\ i\in\{1,\ldots,m-1\} \right\}
$$
is a generating index set of $\DST{n}$.
\end{cor}

\begin{proof}
This follows from Theorem~\ref{thm13}, since $2i\le n-\left\lfloor\frac{n}{3}\right\rfloor\le n-i$, for all $i\in\left\{1,\ldots,\left\lfloor\frac{n}{3}\right\rfloor\right\}$.
\end{proof}

Since the dimension of $\DST{n}$ corresponds to the cardinality of the generating index set $G_{D}$, it is straightforward to obtain the following

\begin{cor}\label{cor15}
For all non-negative integers $n$, we have $\dim\DST{n}=\left\lfloor\frac{n+3}{6}\right\rfloor+\deltamod{1}{n}{6}$.
\end{cor}

%%%%%%%%%%%%%%%%
\subsection{Basis of \texorpdfstring{$\DST{n}$}{DST(n)}}
%%%%%%%%%%%%%%%%

First, using the operators $\He^k$ and the generating index sets $G_D$ introduced above, we obtain a family of bases of $\DST{n}$.

\begin{thm}\label{thm2}
Let $n$ and $m$ be non-negative integers such that $m=\left\lfloor\frac{n+3}{6}\right\rfloor + \deltamod{1}{n}{6}$. For every $k\in\left\{0,\ldots,\left\lfloor\frac{n}{3}\right\rfloor-1\right\}$, let $\iStrig{k}\in\DST{n}$ be such that $\He^{k}\left(\iStrig{k}\right)=U_{n-3k}$. Then when $n$ is even, the set
$$
\left\{ \iStrig{2k}\ \middle|\ k\in\left\{0,\ldots,m-1\right\}\right\}
$$
or when $n$ is odd, the set
$$
\left\{\iStrig{0}\right\}\cup\left\{ \iStrig{2k+1}\ \middle|\ k\in\left\{0,\ldots,m-2\right\}\right\}
$$
is a basis of $\DST{n}$.
\end{thm}

%\begin{rem}
%$\iStrig{0}=U_n$ in the previous result.
%\end{rem}

\begin{proof}
Suppose first that $n$ is even. We consider the set
$$
G_D = \left\{ \left(2i+1,n-\left\lfloor\frac{n}{3}\right\rfloor\right) \ \middle|\ i\in\{0,\ldots,m-1\} \right\}.
$$
Let $k\in\left\{0,\ldots,m-1\right\}$. For $\iStrig{2k}=\left(a_{i,j}\right)_{1\le i\le j\le n}$, since $\He^{2k}\left(\iStrig{2k}\right)=U_{n-6k}$, it follows from Lemma~\ref{lem2} that
$$
a_{i,n-\left\lfloor\frac{n}{3}\right\rfloor} = \left\{\begin{array}{ll}
1 & \text{for } i=2k+1, \\
0 & \text{for } i\in\left\{2k+2,\ldots,\left\lfloor\frac{n}{3}\right\rfloor\right\}.
\end{array}\right.
$$
Moreover, it is clear that $2m-1\le \left\lfloor\frac{n}{3}\right\rfloor$ and thus
$$
\pi_{G_D}\left(\iStrig{2k}\right) = ( \underbrace{\ast , \ldots, \ast}_{k} , 1 , \underbrace{0 , \ldots , 0}_{m-k-1}  ),
$$
where $\ast$ stands for any element of $\Zd$. Therefore the set $\left\{ \pi_{G_D}\left(\iStrig{2k}\right)\ \middle|\ k\in\{0,\ldots,m-1\}\right\}$ is a basis of $\Zd^{|G_D|}$. Finally, since $G_D$ is a generating index set of $\DST{n}$ by Corollary~\ref{cor5}, we conclude that
$$
\left\{ \iStrig{2k}\ \middle|\ k\in\left\{0,\ldots,m-1\right\}\right\}
$$
is a basis of $\DST{n}$.

Suppose now that $n$ is odd. The proof is similar to the even case by considering the generating index set $G_D$ from Corollary~\ref{cor5}
$$
G_D = \left\{\left(1,n-\left\lfloor\frac{n}{3}\right\rfloor\right)\right\}\cup\left\{ \left(2i,n-\left\lfloor\frac{n}{3}\right\rfloor\right)\ \middle|\ i\in\{1,\ldots,m-1\} \right\}.
$$
Using Lemma~\ref{lem2}
$$
\left\{\pi_{G_D}\left(\iStrig{0}\right)\right\}\cup\left\{ \pi_{G_D}\left(\iStrig{2k+1}\right)\ \middle|\ k\in\{0,\ldots,m-2\}\right\}
$$
is a basis of $\Zd^{|G_D|}$ . This implies that the set
$$
\left\{\iStrig{0}\right\}\cup\left\{ \iStrig{2k+1}\ \middle|\ k\in\left\{0,\ldots,m-2\right\}\right\}
$$
is a basis of $\DST{n}$. This completes the proof.
\end{proof}

We now consider the restriction of the linear map $\rho$ on the linear subspace $\HST{n}$; i.e., the linear map $\restriction{\rho}{\HST{n}} : \HST{n}\longrightarrow \DST{n}$ defined by $\rho=r^2+r+id_{n}$. Obviously this map is surjective since $\rho(\Strig{})=\Strig{}$ for all $\Strig{}\in\DST{n}$. Note that, by definition, we have $U_{n-3k}=\rho\left((1)_{n-3k}\right)$. From this, we have

\begin{cor}\label{cor6}
Let $n$ and $m$ be non-negative integers such that $m=\left\lfloor\frac{n+3}{6}\right\rfloor + \deltamod{1}{n}{6}$. For every $k\in\left\{0,\ldots,\left\lfloor\frac{n}{3}\right\rfloor-1\right\}$, let $S_k$ be a symmetric binary sequence of length $n$ such that $\partial^kS_k=(1)_{n-k}$. Then when $n$ is even, the set
$$
\left\{ \rho\left(\Strig{S_{2k}}\right)\ \middle|\ k\in\left\{0,\ldots,m-1\right\}\right\}
$$
or when $n$ is odd, the set
$$
\left\{\rho\left(\Strig{S_0}\right)\right\}\cup\left\{ \rho\left(\Strig{S_{2k+1}}\right)\ \middle|\ k\in\left\{0,\ldots,m-2\right\}\right\}
$$
is a basis of $\DST{n}$.
\end{cor}

\begin{rem}
$\rho\left(\Strig{S_0}\right) = \rho\left(\Strig{(1)_n}\right) = U_n$ in the previous result.
\end{rem}

\begin{proof}
Let $k\in\left\{0,\ldots,\left\lfloor\frac{n}{3}\right\rfloor-1\right\}$. First, since $S_k$ is symmetric, we know from Proposition~\ref{prop3} that $\Strig{S_k}\in\HST{n}$ and thus by the definition of $\rho$, we have $\rho\left(\Strig{S_k}\right)\in\DST{n}$. Moreover, since $\partial^{k}S_k=(1)_{n-k}$, it follows that
$$
\He^{k}\left(\rho\left(\Strig{S_k}\right)\right) = \rho\left(\He^{k}\left(\Strig{S_k}\right)\right) = \rho\left(\Strig{(1)_{n-3k}}\right) = U_{n-3k}.
$$
Therefore, the result follows directly from Theorem~\ref{thm2} by considering the dihedrally symmetric Steinhaus triangles $\iStrig{k} = \rho\left(\Strig{S_{k}}\right)$.
\end{proof}

We end this section by giving an explicit basis of $\DST{n}$ in terms of the $n$-length binary sequences
$$
\BS{n}{k}{\ell} = \left(\binomd{\ell+j-1}{k}\right)_{1\le j\le n},
$$
for all integers $k$ and $\ell$.

\begin{thm}\label{thm8}
Let $n$ and $m$ be non-negative integers such that $m=\left\lfloor\frac{n+3}{6}\right\rfloor + \deltamod{1}{n}{6}$. For every $k\in\left\{0,\ldots,\left\lfloor\frac{n}{3}\right\rfloor-1\right\}$ of the same parity as $n$, let 
\begin{equation*}
\resizebox{\textwidth}{!}{$
\displaystyle\iStrig{k} = \rho\left(\Strig{\BS{n}{k}{\frac{k-n}{2}}}\right) = \left( \left(\binom{\frac{k-n}{2}+j-i}{k+1-i} + \binom{\frac{k+n}{2}-j}{k+i-j} + \binom{\frac{k-n}{2}+i-1}{k+j-n}\right) \tpmod{2}\right)_{1\le i\le j\le n}.
$}
\end{equation*}
Then, when $n$ is even, the set
$$
\left\{ \iStrig{2k}\ \middle|\ k\in\left\{0,\ldots,m-1\right\}\right\}
$$
or when $n$ is odd, the set
$$
\left\{U_n\right\}\cup\left\{ \iStrig{2k+1}\ \middle|\ k\in\left\{0,\ldots,m-2\right\}\right\}
$$
is a basis of $\DST{n}$.
\end{thm}

\begin{proof}
Let $k\in\left\{0,\ldots,\left\lfloor\frac{n}{3}\right\rfloor-1\right\}$ be of the same parity as $n$. First, we know from Lemma~\ref{lem3} that the sequence $\BS{n}{k}{\frac{k-n}{2}}$ is symmetric. Moreover, by Proposition~\ref{prop18}, we have
$$
\partial^k\left(\BS{n}{k}{\frac{k-n}{2}}\right) = \BS{n-k}{0}{\frac{k-n}{2}} = (1)_{n-k}.
$$
We conclude the proof by using Corollary~\ref{cor6} with the sequences $S_k=\BS{n}{k}{\frac{k-n}{2}}$, for all $k\in\left\{0,\ldots,\left\lfloor\frac{n}{3}\right\rfloor-1\right\}$ of the same parity as $n$, and $S_0=(1)_n$, when $n$ is odd.
\end{proof}

\begin{cor}\label{cor11}
Let $n$ and $m$ be non-negative integers such that $m=\left\lfloor\frac{n}{6}\right\rfloor + \deltamod{4}{n}{6}$. For every $k\in\left\{0,\ldots,\left\lfloor\frac{n}{3}\right\rfloor-1\right\}$ of the same parity as $n$, let 
\begin{equation*}
\resizebox{\textwidth}{!}{$
\displaystyle\iStrig{k} = \rho\left(\Strig{\BS{n}{k}{\frac{k-n}{2}}}\right) = \left(\left( \binom{\frac{k-n}{2}+j-i}{k+1-i} + \binom{\frac{k+n}{2}-j}{k+i-j} + \binom{\frac{k-n}{2}+i-1}{k+j-n}\right) \tpmod{2}\right)_{1\le i\le j\le n}.
$}
\end{equation*}
Then, when $n$ is even, the set
$$
\left\{ \iStrig{2k}\ \middle|\ k\in\left\{0,\ldots,m-1\right\}\right\}
$$
or when $n$ is odd, the set
$$
\left\{ \iStrig{2k+1}\ \middle|\ k\in\left\{0,\ldots,m-1\right\}\right\}
$$
is a basis of $\DSTz{n}$.
\end{cor}

\begin{proof}
Using Theorem~\ref{thm8} and Proposition~\ref{prop9}, the result follows since $\DST{n}=\DSTz{n}$ when $n$ is even and $\DST{n} = \DSTz{n} \sqcup \left(\DSTz{n} + U_n\right)$ when $n$ is odd.
\end{proof}

For instance, for $n=22$, we obtain the basis given in Table~\ref{tab2}.
\begin{table}[htbp]
$$
\begin{array}{|c|c|c|}
\hline
\vspace{-2ex} & & \\
k & \BS{22}{2k}{k-11} & \rho\left(\Strig{\BS{22}{2k}{k-11}}\right) \\[1.5ex]
\hline
0 & (1111111111111111111111) & \iStrig{0}=\Strig{(0111111111111111111110)} \\
\hline
1 & (1100110011001100110011) & \iStrig{2}=\Strig{(0110110011001100110110)} \\
\hline
2 & (1000011110000111100001) & \iStrig{4}=\Strig{(0111111110000111111110)} \\
\hline
3 & (0000001100000011000000) & \iStrig{6}=\Strig{(0000000100000010000000)} \\
\hline
\end{array}
$$
\caption{A basis of $\DST{22}$}\label{tab2}
\end{table}
All the dihedrally symmetric Steinhaus triangles of size $22$ are depicted in Figure~\ref{fig10}, where the elements of the basis $\left\{\iStrig{0},\iStrig{2},\iStrig{4},\iStrig{6}\right\}$ are in red and, for every $\Strig{}\in\DST{22}$, the coordinate vector $(x_0,x_2,x_4,x_6)$ of $\Strig{}=x_0\iStrig{0}+x_2\iStrig{2}+x_4\iStrig{4}+x_6\iStrig{6}$ is given.

\begin{figure}[htbp]
\centerline{\includegraphics[width=1.1\textwidth]{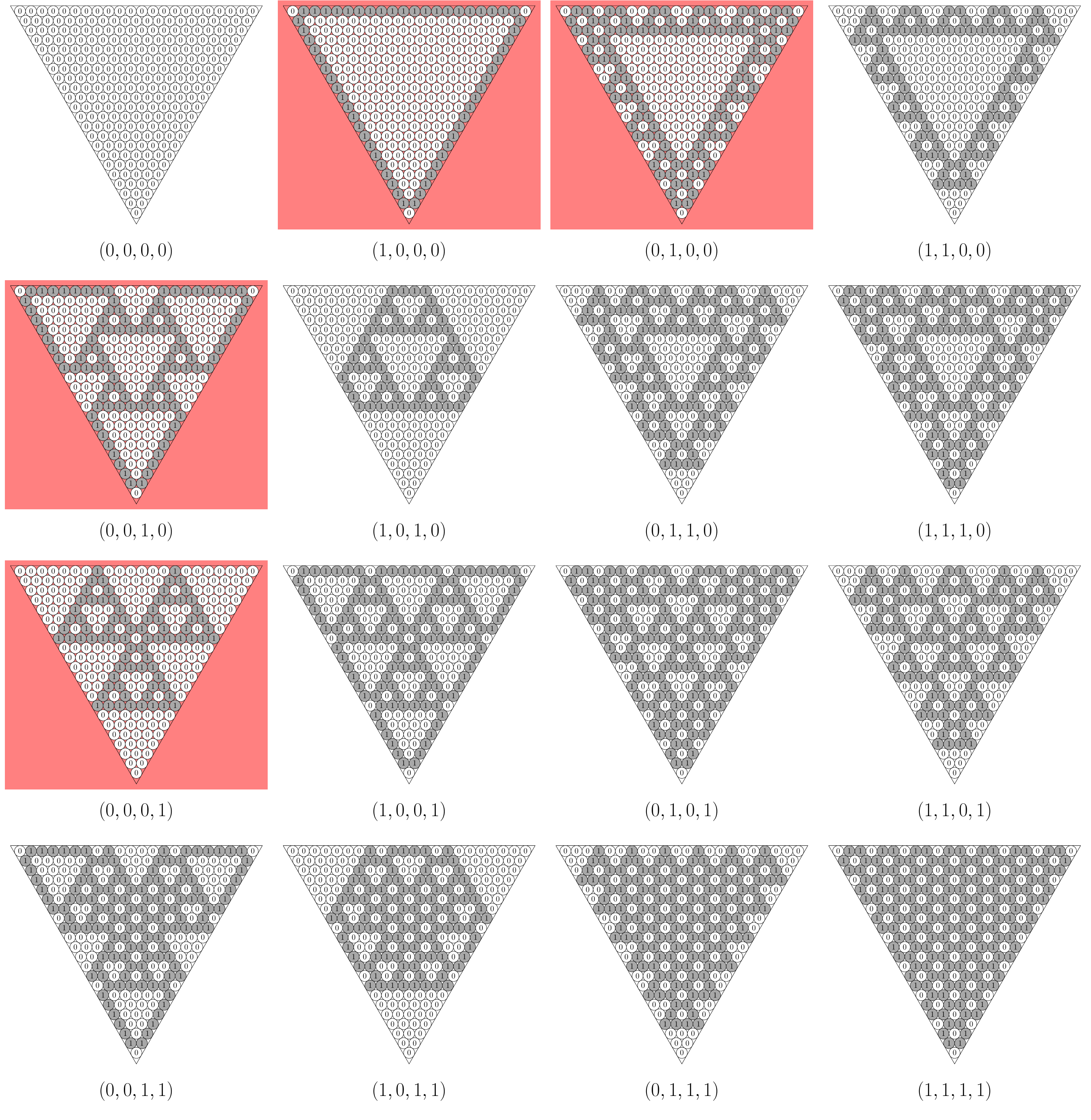}}
\caption{The $16$ triangles of $\DST{22}$ where the $4$ red triangles form a basis}\label{fig10}
\end{figure}

%%%%%%%%%%%%%%%%
%%%%%%%%%%%%%%%%
\section{Parity-regular Steinhaus graphs}\label{sec:6}
%%%%%%%%%%%%%%%%
%%%%%%%%%%%%%%%%

For any binary sequence $S=\left(a_j\right)_{1\le j\le n}$ of length $n$, we denote by $ir(S)$ the {\em interlacing of the sequence $S$ and its reversed sequence}; i.e., the sequence $ir(S)=\left(b_j\right)_{1\le j\le 2n}$ of length $2n$ defined by
$$
b_{2j-1} = a_j\quad\text{and}\quad b_{2j}=a_{n-j+1},
$$
for all $j\in\{1,\ldots,n\}$. For instance, for $S=(101000)$, we have $ir(S)=({\color{red}1}{\color{blue}0}{\color{red}0}{\color{blue}0}{\color{red}1}{\color{blue}0}{\color{red}0}{\color{blue}1}{\color{red}0}{\color{blue}0}{\color{red}0}{\color{blue}1})$.

For any positive integer $n$, we consider the linear map
$$
\theta : \SG{n} \longrightarrow \ST{2n-1}\quad\text{by}\quad\theta\left(\Sgraph{S}\right) = \Strig{\textstyle\int_{n,0}ir(S)}.
$$ 
Note that the Steinhaus triangle $\Strig{S}\in\ST{n-1}$ is then a subtriangle of $\theta(\Sgraph{S})\in\ST{2n-1}$. Indeed, for the sequence $S=\left(a_j\right)_{1\le j\le n-1}$ and the Steinhaus triangle $\theta(\Sgraph{S})=\Strig{\int_{n,0}ir(S)}=\left(a_{i,j}\right)_{1\le i\le j\le 2n-1}$, the Steinhaus triangle $\Strig{S}$ is simply the subtriangle $\left(a_{2i,2j}\right)_{1\le i\le j\le n-1}$. Since $a_{2,2j}=a_j$ for all $j\in\{1,\ldots,n-1\}$, by the definition of $\theta$ and using the local rule, we have
\begin{equation*}
\resizebox{\textwidth}{!}{$
\displaystyle a_{2i,2j} \equiv a_{2i-1,2j-1} + a_{2i-1,2j} \equiv a_{2i-2,2j-2} + 2a_{2i-2,2j-1} + a_{2i-2,2j} \equiv a_{2i-2,2j-2} + a_{2i-2,2j} \pmod{2},
$}
\end{equation*}
for all integers $i$ and $j$ such that $2\le i\le j\le n-1$. For instance, for the sequence $S=(101000)$, the Steinhaus triangle $\theta(\Sgraph{S})$ is depicted in Figure~\ref{fig11}, where the subtriangle $\Strig{S}$ appears in red.

\begin{figure}[htbp]
\centerline{
\begin{tikzpicture}[scale=0.25]
\pgfmathparse{sqrt(3)}\let\sq\pgfmathresult
\figST{{0,1,1,1,1,0,0,0,1,1,1,1,0}}{black}{col0}{col1}{white}
\draw[fill=col2!50!col1] (1,-\sq) circle (1);
\draw (1,-\sq) node {$1$};
\draw[fill=col2] (5,-\sq) circle (1);
\draw (5,-\sq) node {$0$};
\draw[fill=col2!50!col1] (9,-\sq) circle (1);
\draw (9,-\sq) node {$1$};
\draw[fill=col2] (13,-\sq) circle (1);
\draw (13,-\sq) node {$0$};
\draw[fill=col2] (17,-\sq) circle (1);
\draw (17,-\sq) node {$0$};
\draw[fill=col2] (21,-\sq) circle (1);
\draw (21,-\sq) node {$0$};
\draw[fill=col2!50!col1] (3,-3*\sq) circle (1);
\draw (3,-3*\sq) node {$1$};
\draw[fill=col2!50!col1] (7,-3*\sq) circle (1);
\draw (7,-3*\sq) node {$1$};
\draw[fill=col2!50!col1] (11,-3*\sq) circle (1);
\draw (11,-3*\sq) node {$1$};
\draw[fill=col2] (15,-3*\sq) circle (1);
\draw (15,-3*\sq) node {$0$};
\draw[fill=col2] (19,-3*\sq) circle (1);
\draw (19,-3*\sq) node {$0$};
\draw[fill=col2] (5,-5*\sq) circle (1);
\draw (5,-5*\sq) node {$0$};
\draw[fill=col2] (9,-5*\sq) circle (1);
\draw (9,-5*\sq) node {$0$};
\draw[fill=col2!50!col1] (13,-5*\sq) circle (1);
\draw (13,-5*\sq) node {$1$};
\draw[fill=col2] (17,-5*\sq) circle (1);
\draw (17,-5*\sq) node {$0$};
\draw[fill=col2] (7,-7*\sq) circle (1);
\draw (7,-7*\sq) node {$0$};
\draw[fill=col2!50!col1] (11,-7*\sq) circle (1);
\draw (11,-7*\sq) node {$1$};
\draw[fill=col2!50!col1] (15,-7*\sq) circle (1);
\draw (15,-7*\sq) node {$1$};
\draw[fill=col2!50!col1] (9,-9*\sq) circle (1);
\draw (9,-9*\sq) node {$1$};
\draw[fill=col2] (13,-9*\sq) circle (1);
\draw (13,-9*\sq) node {$0$};
\draw[fill=col2!50!col1] (11,-11*\sq) circle (1);
\draw (11,-11*\sq) node {$1$};
\end{tikzpicture}}
\caption{The Steinhaus triangle $\theta(\Sgraph{101000})$ where $\Strig{(101000)}$ appears in red}\label{fig11}
\end{figure}

By definition of the linear map $ir$, we know that the sequence $ir(S)$ is symmetric and $\sigma_2\left(ir(S)\right)=0$. It follows from Proposition~\ref{prop1} that the sequence $\int_{n,0}ir(S)$ is symmetric too. Therefore, using Proposition~\ref{prop3}, we have
$$
\theta(\Sgraph{S})=\Strig{{\textstyle\int_{n,0}ir(S)}}\in\HST{2n-1},
$$
for all $\Sgraph{S}\in\SG{n}$. Moreover, using Lemma~\ref{lem6} and since by definition, the middle term of the sequence $\int_{n,0}ir(S)$ is $0$, we have that
$$
\sigma_2\left({\textstyle\int_{n,0}ir(S)}\right)=0,
$$
for any sequence $S$ of length $n-1$.

The main result of this section is to show that the restriction of $\theta$ to the linear subspace of even Steinhaus graphs $\ESG{n}$ induces an isomorphism between $\ESG{n}$ and $\DSTz{2n-1}$.

\begin{thm}\label{thm3}
Let $S$ be a binary sequence of length $n-1\ge 0$. Then the Steinhaus graph $\Sgraph{S}$ is even if and only if the Steinhaus triangle $\theta\left(\Sgraph{S}\right)$ is dihedrally symmetric.
\end{thm}

The proof is based on the following

\begin{lem}\label{lem4}
Let $S$ be a binary sequence of length $n-1\ge 0$. For all $i\in\{1,\ldots,n\}$, let $v_i$ denote the $i$th vertex of the Steinhaus graph $\Sgraph{S}$ and let $\theta\left(\Sgraph{S}\right)=\left(b_{i,j}\right)_{1\le i\le j\le 2n-1}\in\HST{2n-1}$. Then,
$$
\deg(v_1) \equiv b_{1,1} \pmod{2}
$$
and, for any $i\in\{2,\ldots,n\}$, if the $(2i-2)$-th column $C_{2i-2}=\left(b_{j,2i-2}\right)_{1\le j\le 2i-2}$ of $\theta\left(\Sgraph{S}\right)$ is symmetric, then
$$
\deg(v_i) \equiv b_{i,2i-1} \pmod{2}.
$$
\end{lem}

\begin{proof}[Proof of Lemma~\ref{lem4}]
We consider the adjacency matrix $\Smat{S}=\left(a_{i,j}\right)_{1\le i,j\le n}$ of the Steinhaus graph $\Sgraph{S}$. We already know that the upper-triangular part of $\Smat{S}$; i.e., the Steinhaus triangle $\Strig{S}=\left(a_{i,j}\right)_{1\le i<j\le n}$, corresponds to the subtriangle $\left(b_{2i,2j}\right)_{1\le i\le j\le n-1}$ of $\theta\left(\Sgraph{S}\right)$. In other words, we have $a_{i,j}=b_{2i,2j-2}$, for all integers $i$ and $j$ such that $1\le i<j\le n$. Then,
\begin{equation}\label{eq4}
\deg(v_1) = \sum_{j=2}^{n}a_{1,j} = \sum_{j=2}^{n}b_{2,2j-2} = \sum_{j=1}^{n-1}b_{2,2j},
\end{equation}
\begin{equation}\label{eq5}
\deg(v_i) = \sum_{j=1}^{i-1}a_{j,i} + \sum_{j=i+1}^{n}a_{i,j} = \sum_{j=1}^{i-1}b_{2j,2i-2} + \sum_{j=i+1}^{n}b_{2i,2j-2} = \sum_{j=1}^{i-1}b_{2j,2i-2} + \sum_{j=i}^{n-1}b_{2i,2j},
\end{equation}
for all $i\in\{2,\ldots,n-1\}$, and
\begin{equation}\label{eq6}
\deg(v_n) = \sum_{j=1}^{n-1}a_{j,n} = \sum_{j=1}^{n-1}b_{2j,2n-2}.
\end{equation}

We claim that
\begin{equation}\label{eq7}
\sum_{j=i}^{n-1}b_{2i,2j} = \sum_{j=2i}^{n+i-1}b_{2i,j},
\end{equation}
for all $i\in\{1,\ldots,n-1\}$. Since $\theta\left(\Sgraph{S}\right)$ is horizontally symmetric, we know that its $i$th row $R_i=\left(b_{i,j}\right)_{i\le j\le 2n-1}$ is symmetric for all $i\in\{1,\ldots,2n-1\}$. Let $i\in\{1,\ldots,n-1\}$. Since the sequence $R_{2i}$ is symmetric of even length $2(n-i)$ with $b_{2i,j}=b_{2i,2n-1+2i-j}$ for all $j\in\{2i,\ldots,2n-1\}$, we obtain the following identities by dividing in half $\sigma(R_{2i})$ in two different ways:
$$
\sigma(R_{2i}) = \sum_{j=2i}^{2n-1}b_{2i,j} = \sum_{j=i}^{n-1}b_{2i,2j} + \sum_{j=i}^{n-1}b_{2i,2j+1} = \sum_{j=i}^{n-1}b_{2i,2j} + \sum_{j=i}^{n-1}b_{2i,2n-1+2i-(2j+1)} = 2\sum_{j=i}^{n-1}b_{2i,2j}
$$
and
$$
\sigma(R_{2i}) = \sum_{j=2i}^{2n-1}b_{2i,j} = \sum_{j=2i}^{n+i-1}b_{2i,j} + \sum_{j=n+i}^{2n-1}b_{2i,j} = \sum_{j=2i}^{n+i-1}b_{2i,j} + \sum_{j=n+i}^{2n-1}b_{2i,2n-1+2i-j} = 2\sum_{j=2i}^{n+i-1}b_{2i,j}.
$$
Combining these two identities, the claim \eqref{eq7} is proved.

Using \eqref{eq7} and the local rule, we deduce from \eqref{eq4} that
\begin{equation}\label{eq9}
\deg(v_1) = \sum_{j=1}^{n-1}b_{2,2j} = \sum_{j=2}^{n}b_{2,j} \equiv \sum_{j=2}^{n}(b_{1,j-1}+b_{1,j}) = \sum_{j=1}^{n-1}b_{1,j}+\sum_{j=2}^{n}b_{1,j} \equiv b_{1,1}+b_{1,n}\pmod{2}.
\end{equation}
Now, let $i\in\{2,\ldots,n\}$ and suppose that the column $C_{2i-2}$ of even size $2i-2$ is symmetric. Then, as with \eqref{eq7}, using a double counting of $\sigma(C_{2i-2})$, we obtain the following identity
\begin{equation}\label{eq8}
\sum_{j=1}^{i-1}b_{2j,2i-2} = \sum_{j=i}^{2i-2}b_{j,2i-2}.
\end{equation}
Using \eqref{eq7}, \eqref{eq8} and the local rule, we deduce from \eqref{eq5} that
\begin{equation}\label{eq10}
\deg(v_i) \begin{array}[t]{l}
= \displaystyle\sum_{j=1}^{i-1}b_{2j,2i-2} + \sum_{j=i}^{n-1}b_{2i,2j} = \sum_{j=i}^{2i-2}b_{j,2i-2} + \sum_{j=2i}^{n+i-1}b_{2i,j} \\ \ \\
\equiv \displaystyle\sum_{j=i}^{2i-2}(b_{j,2i-1}+b_{j+1,2i-1}) + \sum_{j=2i}^{n+i-1}(b_{2i-1,j-1}+b_{2i-1,j}) \\ \ \\
= \displaystyle\sum_{j=i}^{2i-2}b_{j,2i-1} + \sum_{j=i+1}^{2i-1}b_{j,2i-1} + \sum_{j=2i-1}^{n+i-2}b_{2i-1,j} + \sum_{j=2i}^{n+i-1}b_{2i-1,j} \\ \ \\
\equiv b_{i,2i-1} + b_{2i-1,2i-1} + b_{2i-1,2i-1} + b_{2i-1,n+i-1} \equiv b_{i,2i-1} + b_{2i-1,n+i-1} \pmod{2},
\end{array}
\end{equation}
for all $i\in\{2,\ldots,n-1\}$. From \eqref{eq6}, we have that
\begin{equation}\label{eq11}
\deg(v_n) \begin{array}[t]{l}
= \displaystyle\sum_{j=1}^{n-1}b_{2j,2n-2} = \sum_{j=n}^{2n-2}b_{j,2n-2} \equiv \sum_{j=n}^{2n-2}(b_{j,2n-1}+b_{j+1,2n-1}) \\ \ \\
= \displaystyle\sum_{j=n}^{2n-2}b_{j,2n-1} + \sum_{j=n+1}^{2n-1}b_{j,2n-1} \equiv b_{n,2n-1} + b_{2n-1,2n-1} \pmod{2}.
\end{array}
\end{equation}

Since $\theta\left(\Sgraph{S}\right)$ is horizontally symmetric, we know by Proposition~\ref{prop11} that
$$
b_{2n-1-2i,2n-1-i}=0\ \text{for all}\ i\in\{0,\ldots,n-2\}.
$$
Moreover, by definition of $\theta(\Sgraph{S})=\Strig{{\textstyle\int_{n,0}ir(S)}}$, we have that $b_{1,n}=0$. Therefore, we have
\begin{equation}\label{eq3}
b_{2i+1,n+i} = 0\ \text{for all}\ i\in\{0,\ldots,n-1\}.
\end{equation}
Finally, by combining \eqref{eq9}, \eqref{eq10} and \eqref{eq11} with \eqref{eq3}, Lemma~\ref{lem4} is proved.
\end{proof}

\begin{proof}[Proof of Theorem~\ref{thm3}]
First, suppose that $\theta\left(\Sgraph{S}\right)$ is dihedrally symmetric. Then, the Steinhaus triangle $r\left(\theta\left(\Sgraph{S}\right)\right)=\theta\left(\Sgraph{S}\right)$ is horizontally symmetric and the column $C_i$ of $\theta\left(\Sgraph{S}\right)$ is symmetric for all $i\in\left\{1,\ldots,2n-1\right\}$. It follows from Lemma~\ref{lem4} that
$$
\deg(v_i) \equiv b_{i,2i-1} \pmod{2}
$$
for all $i\in\{1,\ldots,n\}$. Since $\theta\left(\Sgraph{S}\right)$ is dihedrally symmetric, we know from Corollary~\ref{cor8} that $b_{i,2i-1}=0$ for all $i\in\{1,\ldots,n-1\}$. Moreover, since $\theta\left(\Sgraph{S}\right)\in\DSTz{2n-1}$, we have $b_{n,2n-1}=b_{1,n}=0$. Therefore, for every $i\in\{1,\ldots,n\}$, the vertex $v_i$ is of even degree and the Steinhaus graph $\Sgraph{S}$ is even.

Conversely, suppose that the Steinhaus graph $\Sgraph{S}$ is even. We prove, by induction on $i$, that all the columns $C_i$ of $\theta\left(\Sgraph{S}\right)$ are symmetric. From Lemma~\ref{lem4}, we know that $b_{1,1}\equiv\deg(v_1)\tpmod{2}$. Since $\deg(v_1)$ is even, it follows that $b_{1,1}=0$. Moreover, since $b_{2,2} \equiv b_{1,1}+b_{1,2} \equiv b_{1,2} \tpmod{2}$, we have that $b_{1,2}=b_{2,2}$. Therefore, the columns $C_1$ and $C_2$ are symmetric. Suppose now that the columns $C_1, C_2, \ldots,C_{2i}$ are symmetric for some $i\in\{1,\ldots,n-1\}$. First, since $C_{2i}$ is symmetric of even length, we know from Lemma~\ref{lem6} that $\sigma_2\left(C_{2i}\right)=0$. Therefore, since $\partial C_{2i+1}=C_{2i}$, it follows from Proposition~\ref{prop1} that $C_{2i+1}$ is symmetric. Moreover, since $C_{2i}$ is symmetric and the vertex $v_{i+1}$ is of even degree, we obtain by Lemma~\ref{lem4} that
$$
b_{i+1,2i+1} \equiv \deg(v_{i+1}) \equiv 0 \pmod{2}.
$$
Since $b_{i+1,2i+1}=0$, from Lemma~\ref{lem6} we have that $\sigma_2\left(C_{2i+1}\right)=0$. Therefore, when $i<n-1$, since $\partial C_{2i+2}=C_{2i+1}$, using Proposition~\ref{prop1} again, we have that $C_{2i+2}$ is symmetric. This concludes the proof that all the columns $C_i$ of $\theta\left(\Sgraph{S}\right)$ are symmetric. Obviously, it follows that the Steinhaus triangle $r\left(\theta\left(\Sgraph{S}\right)\right)$ is horizontally symmetric. Finally, since the triangles $\theta\left(\Sgraph{S}\right)$ and $r\left(\theta\left(\Sgraph{S}\right)\right)$ are horizontally symmetric, we know from Proposition~\ref{prop12} that $\theta\left(\Sgraph{S}\right)$ is dihedrally symmetric.
\end{proof}

\begin{cor}\label{cor9}
For any positive integer $n$, the restriction
$$
\theta\vert_{\ESG{n}} : \ESG{n} \longrightarrow \DSTz{2n-1}
$$
is an isomorphism.
\end{cor}

\begin{proof}
Let $n$ be a positive integer. We consider the linear map
$$
\psi : \DSTz{2n-1} \longrightarrow \ESG{n}\quad\text{by}\quad\psi\left(\left(a_{i,j}\right)_{1\le i\le j\le 2n-1}\right) = \Sgraph{\left(a_{2,2j}\right)_{1\le j\le n-1}}.
$$
We know from Theorem~\ref{thm3} that the linear maps $\theta\vert_{\ESG{n}}$ and $\psi$ are well defined. Moreover, it is easy to verify that $\psi\circ\theta\vert_{\ESG{n}}=id_{\ESG{n}}$ and $\theta\vert_{\ESG{n}}\circ\psi=id_{\DSTz{2n-1}}$. This completes the proof.
\end{proof}

This new result permits us to obtain the following two corollaries that were first proved in \cite{Dymacek:1979aa}.

\begin{cor}
For all positive integers $n$, we have $\dim\ESG{n}=\left\lfloor\frac{n-1}{3}\right\rfloor$.
\end{cor}

\begin{proof}
By Corollary~\ref{cor9}, the vector space $\ESG{n}$ is isomorphic to $\DSTz{2n-1}$. From Corollary~\ref{cor10}, we deduce that
$$
\dim\ESG{n}=\dim\DSTz{2n-1}=\left\lfloor\frac{2n-1}{6}\right\rfloor+\deltamod{4}{2n-1}{6}=\left\lfloor\frac{2n-1}{6}\right\rfloor=\left\lfloor\frac{n-1}{3}\right\rfloor,
$$
for all positive integers $n$.
\end{proof}

\begin{cor}
The Steinhaus matrix $\Smat{S}$ associated to an even Steinhaus graph $\Sgraph{S}$ is doubly symmetric, i.e., all the diagonals of $\Smat{S}$ are symmetric.
\end{cor}

\begin{rem}
This was a key result in the first proof of the formula $\dim\ESG{n}=\left\lfloor\frac{n-1}{3}\right\rfloor$ in \cite{Dymacek:1979aa}. Another simple proof of this result can also be found in \cite{Chappelon:2009aa}.
\end{rem}

\begin{proof}
Let $S$ be a binary sequence of length $n-1$ whose associated Steinhaus graph $\Sgraph{S}$ is even. In other words, we want to prove that the Steinhaus triangle $r^2\left(\Strig{S}\right)$ is horizontally symmetric. Let $\theta\left(\Sgraph{S}\right)=\left(b_{i,j}\right)_{1\le i\le j\le 2n-1}$. Since $\theta\left(\Sgraph{S}\right)\in\DSTz{2n-1}$ by Theorem~\ref{thm3}, it follows that $r^2\left(\theta\left(\Sgraph{S}\right)\right)$ is horizontally symmetric. Therefore, the diagonal $D_i=\left(b_{j,i+j}\right)_{1\le j\le 2n-1-i}$ is symmetric for all $i\in\{0,\ldots,2n-2\}$. It follows that
\begin{equation}\label{eq12}
b_{j,i+j} = b_{2n-i-j,2n-j},
\end{equation}
for all $j\in\{1,\ldots,2n-1-i\}$ and for all $i\in\{0,\ldots,2n-2\}$. Let $\Strig{S}=\left(a_{i,j}\right)_{1\le i\le j\le n-1}$. As already seen, $\Strig{S}$ corresponds to the subtriangle $\Strig{S}=\left(b_{2i,2j}\right)_{1\le i\le j\le n-1}$ of $\theta\left(\Sgraph{S}\right)$. Therefore, we have $a_{i,j}=b_{2i,2j}$ for all integers $i$ and $j$ such that $1\le i\le j\le n-1$. Let $i\in\{0,\ldots,n-2\}$. From \eqref{eq12}, we know that
$$
a_{j,i+j} = b_{2j,2i+2j} = b_{2n-2i-2j,2n-2j} = b_{2(n-i-j),2(n-j)} = a_{n-i-j,n-j},
$$
for all $j\in\{1,\ldots,n-1-i\}$. We conclude that the diagonal $\left(a_{j,i+j}\right)_{1\le j\le n-1-i}$ is symmetric for all $i\in\{0,\ldots,n-2\}$ and the Steinhaus triangle $r^2\left(\Strig{S}\right)$ is horizontally symmetric.
\end{proof}

Using Theorem~\ref{thm3} and the results of Section~\ref{sec:5}, we are now ready for giving a basis of $\ESG{n}$.

\begin{thm}\label{thm9}
Let $n$ be a positive integer. The set
$$
\left\{ \psi\rho\left(\Strig{\BS{2n-1}{2k+1}{k-n+1}}\right) \ \middle|\ k\in\left\{0,\ldots,\left\lfloor\frac{n-1}{3}\right\rfloor-1\right\} \right\}
$$
is a basis of $\ESG{n}$, where $\psi\rho\left(\Strig{\BS{2n-1}{2k+1}{k-n+1}}\right) = \Sgraph{S_k}$ with
$$
S_k = \left(\left(\binom{k-n+2j-1}{2k}+\binom{k+n-2j}{2k-2j+3}+\binom{k-n+2}{2k-2n+2j+2}\right)\tpmod{2}\right)_{1\le j\le n-1},
$$
for all $k\in\left\{0,\ldots,\left\lfloor\frac{n-1}{3}\right\rfloor-1\right\}$.
\end{thm}

\begin{proof}
Let $n$ be a positive integer. From Corollary~\ref{cor11}, we know that
$$
\left\{ \rho\left(\Strig{\BS{2n-1}{2k+1}{k-n+1}}\right) \ \middle|\ k\in\left\{0,\ldots,\left\lfloor\frac{n-1}{3}\right\rfloor-1\right\} \right\}
$$
is a basis of $\DSTz{2n-1}$, where $\rho\left(\Strig{\BS{2n-1}{2k+1}{k-n+1}}\right)$ is the triangle
\begin{equation*}
\resizebox{\textwidth}{!}{$
\displaystyle\left(\left(\binom{k-n+1+j-i}{2k+2-i}+\binom{k+n-j}{2k+1+i-j}+\binom{k-n+i}{2k+2+j-2n}\right)\tpmod{2}\right)_{1\le i\le j\le 2n-1},
$}
\end{equation*}
for all $k\in\left\{0,\ldots,\left\lfloor\frac{n-1}{3}\right\rfloor-1\right\}$. It follows, by Corollary~\ref{cor9}, that
$$
\left\{ \psi\rho\left(\Strig{\BS{2n-1}{2k}{k-n+1}}\right) \ \middle|\ k\in\left\{0,\ldots,\left\lfloor\frac{n-1}{3}\right\rfloor-1\right\} \right\}
$$
is a basis of $\ESG{n}$. Since $\psi\left(\left(a_{i,j}\right)\right)=\Sgraph{\left(a_{2,2j}\right)_{1\le j\le n-1}}$, we conclude that
$$
\psi\rho\left(\Strig{\BS{2n-1}{2k+1}{k-n+1}}\right) = \Sgraph{S_k},
$$
with
$$
S_k = \left(\left(\binom{k-n+2j-1}{2k}+\binom{k+n-2j}{2k-2j+3}+\binom{k-n+2}{2k-2n+2j+2}\right)\tpmod{2}\right)_{1\le j\le n-1},
$$
for all $k\in\left\{0,1,\ldots,\left\lfloor\frac{n-1}{3}\right\rfloor-1\right\}$.
\end{proof}

For instance, for $n=12$, we obtain the basis given in Table~\ref{tab3}.

\begin{table}[htbp]
\resizebox{\textwidth}{!}{
$
\begin{array}{|c|c|c|c|}
\hline
\vspace{-2ex} & & & \\
k & \BS{23}{2k+1}{k-11} & \rho\left(\Strig{\BS{23}{2k+1}{k-11}}\right) & \psi\left(\rho\left(\Strig{\BS{23}{2k+1}{k-11}}\right)\right) \\[1.5ex]
\hline
0 & (10101010101010101010101) & \Strig{(01101010101010101010110)} & G_1=\Sgraph{11111111110} \\
\hline
1 & (01000100010001000100010) & \Strig{(00010100010001000101000)} & G_2=\Sgraph{01101010110} \\
\hline
2 & (10000010100000101000001) & \Strig{(01111110100000101111110)} & G_3=\Sgraph{10011001000} \\
\hline
\end{array}
$}
\caption{A basis of $\ESG{12}$}\label{tab3}
\end{table}

All the even Steinhaus graphs of order $12$ are depicted in Figure~\ref{fig15}, where the elements of the basis $\left\{G_1,G_2,G_3\right\}$ are in red and, for every $G\in\ESG{12}$, the coordinate vector $(x_1,x_2,x_3)$ of $G=x_1G_1+x_2G_2+x_3G_3$ is given.

\begin{figure}[htbp]
\centerline{\includegraphics[width=\textwidth]{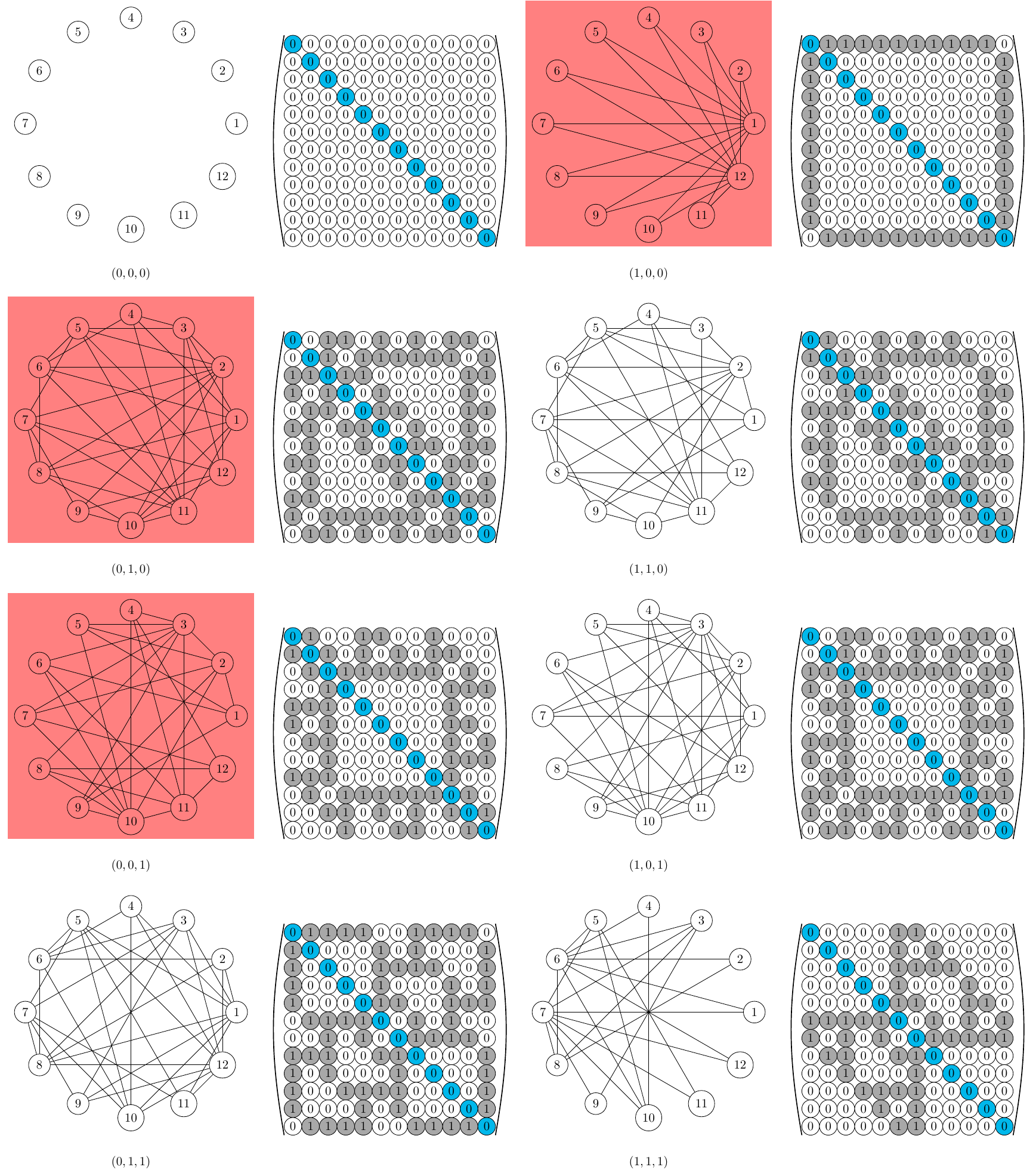}}
\caption{The $8$ graphs of $\ESG{12}$ where the $3$ red graphs form a basis}\label{fig15}
\end{figure}

We end this section by giving a basis of the linear subspace $\PRSG{n}$ of parity-regular Steinhaus graphs. By definition, we know that $\PRSG{n}=\ESG{n}\sqcup\OSG{n}$, where $\OSG{n}$ is the set of odd Steinhaus graphs of order $n$. As already remarked, it is clear that $\OSG{n}=\emptyset$, when $n$ is odd. For $n$ even, we obtain the following

\begin{prop}\label{prop19}
For any positive integer $n$,
$$
\OSG{2n} = \ESG{2n} + \Sgraph{(0)_{2n-2}\cdot(1)}.
$$
\end{prop}

The proof is based on the following lemma, where the linear map $\eta$ is defined by
$$
\eta : \SG{2n} \longrightarrow \SG{2n}\quad\text{by}\quad \eta\left(\Sgraph{S}\right) = \Sgraph{S+(0)_{2n-2}\cdot(1)}.
$$

\begin{lem}\label{lem5}
For any positive integer $n$, the Steinhaus graph $\Sgraph{S}$ of order $2n$ is even if and only if $\eta\left(\Sgraph{S}\right)$ is odd.
\end{lem}

\begin{proof}[Proof of Lemma~\ref{lem5}]
We consider the respective adjacency matrices $\Smat{S}=\left(a_{i,j}\right)_{1\le i,j\le 2n}$ of $\Sgraph{S}$ and $\Smat{S+(0)_{2n-2}\cdot(1)}=\left(b_{i,j}\right)_{1\le i,j\le 2n}$ of $\eta\left(\Sgraph{S}\right)$. It is easy to see that
$$
\left\{\begin{array}{ll}
b_{i,j} = a_{i,j} & \text{for all } 1\le i<j\le 2n-1,\\
b_{i,2n} \equiv a_{i,2n}+1\pmod{2} & \text{for all } 1\le i\le 2n-1,
\end{array}\right.
$$
since the sequence $S+(0)_{2n-2}\cdot(1)$ only differs by the last term from the sequence $S$. It follows that
$$
\deg_{\Sgraph{S}}\left(v_1\right) = \sum_{j=2}^{2n}a_{1,j} \equiv \sum_{j=2}^{2n}b_{1,n}+1 = \deg_{\eta\left(\Sgraph{S}\right)}\left(v_1\right) +1\pmod{2},
$$
$$
\deg_{\Sgraph{S}}\left(v_i\right) = \sum_{j=1}^{i-1}a_{j,i} + \sum_{j=i+1}^{2n}a_{i,j} \equiv \sum_{j=1}^{i-1}b_{j,i} + \sum_{j=i+1}^{2n}b_{i,j}+1 = \deg_{\eta\left(\Sgraph{S}\right)}\left(v_i\right) +1\pmod{2},
$$
for all $i\in\left\{2,\ldots,2n-1\right\}$, and
$$
\deg_{\Sgraph{S}}\left(v_{2n}\right) = \sum_{i=1}^{2n-1}a_{i,2n} \equiv \sum_{i=1}^{2n-1}b_{i,2n}+2n-1 \equiv \deg_{\eta\left(\Sgraph{S}\right)}\left(v_{2n}\right) +1\pmod{2},
$$
where $\deg_{\Sgraph{S}}\left(v_i\right)$ and $\deg_{\eta\left(\Sgraph{S}\right)}\left(v_i\right)$ are the degrees of the $i$th vertex of the Steinhaus graphs $\Sgraph{S}$ and $\eta\left(\Sgraph{S}\right)$, respectively, for all $i\in\left\{1,\ldots,2n\right\}$. Since for all $i\in\left\{1,\ldots,2n\right\}$, we have $\deg_{\Sgraph{S}}\left(v_i\right)\equiv\left(\deg_{\eta\left(\Sgraph{S}\right)}\left(v_i\right)+1\right)\tpmod{2}$, the result follows.
\end{proof}

\begin{proof}[Proof of Proposition~\ref{prop19}]
From Lemma~\ref{lem5}, it is clear that $\eta$ induces an involution on $\PRSG{2n}$ with $\OSG{2n} = \eta\left(\ESG{2n}\right) =  \ESG{2n} + \Sgraph{(0)_{2n-2}\cdot(1)}$.
\end{proof}

It immediately follows that, for any positive integer $n$, we have
\begin{equation}\label{eq24}
\PRSG{2n-1} = \ESG{2n-1}
\end{equation}
and
\begin{equation}\label{eq25}
\PRSG{2n} = \ESG{2n}\sqcup\left(\ESG{2n}+\Sgraph{(0)_{2n-2}\cdot(1)}\right).
\end{equation}
This proves

\begin{prop}
For all positive integers $n$, we have $\dim\PRSG{n}=\left\lfloor\frac{n-1}{3}\right\rfloor+\deltamod{0}{n}{2}$.
\end{prop}

Combining the identities \eqref{eq24} and \eqref{eq25} with Theorem~\ref{thm9}, we obtain the following

\begin{thm}
Let $n$ be a positive integer. The set
$$
\left\{ \psi\rho\left(\Strig{\BS{2n-1}{2k+1}{k-n+1}}\right) \ \middle|\ k\in\left\{0,\ldots,\left\lfloor\frac{n-1}{3}\right\rfloor-1\right\} \right\},
$$
when $n$ is odd, or the set
$$
\left\{\Sgraph{(0)_{n-2}\cdot(1)}\right\} \cup \left\{ \psi\rho\left(\Strig{\BS{2n-1}{2k+1}{k-n+1}}\right) \ \middle|\ k\in\left\{0,\ldots,\left\lfloor\frac{n-1}{3}\right\rfloor-1\right\} \right\},
$$
when $n$ is even, is a basis of $\PRSG{n}$, where $\psi\rho\left(\Strig{\BS{2n-1}{2k+1}{k-n+1}}\right) = \Sgraph{S_k}$ with
$$
S_k = \left(\left(\binom{k-n+2j-1}{2k}+\binom{k+n-2j}{2k-2j+3}+\binom{k-n+2}{2k-2n+2j+2}\right)\tpmod{2}\right)_{1\le j\le n-1},
$$
for all $k\in\left\{0,\ldots,\left\lfloor\frac{n-1}{3}\right\rfloor-1\right\}$.
\end{thm}

For instance, for $n=12$, we obtain
{\footnotesize
$$
\begin{array}{cccc}
G_0 = \Sgraph{00000000001}, & G_1=\Sgraph{11111111110}, & G_2=\Sgraph{01101010110}, & G_3=\Sgraph{10011001000}. \\
\end{array}
$$}
All the parity-regular Steinhaus graphs of order $12$ are depicted in Figure~\ref{fig15} for the even graphs and in Figure~\ref{fig24} for the odd ones, where the elements of the basis $\left\{G_0,G_1,G_2,G_3\right\}$ are in red and for every $G\in\PRSG{12}$, the coordinate vector $(x_0,x_1,x_2,x_3)$ of $G=x_0G_0+x_1G_1+x_2G_2+x_3G_3$ is given ($(x_1,x_2,x_3)$ when $x_0=0$ in Figure~\ref{fig15}).

\begin{figure}[htbp]
\centerline{\includegraphics[width=\textwidth]{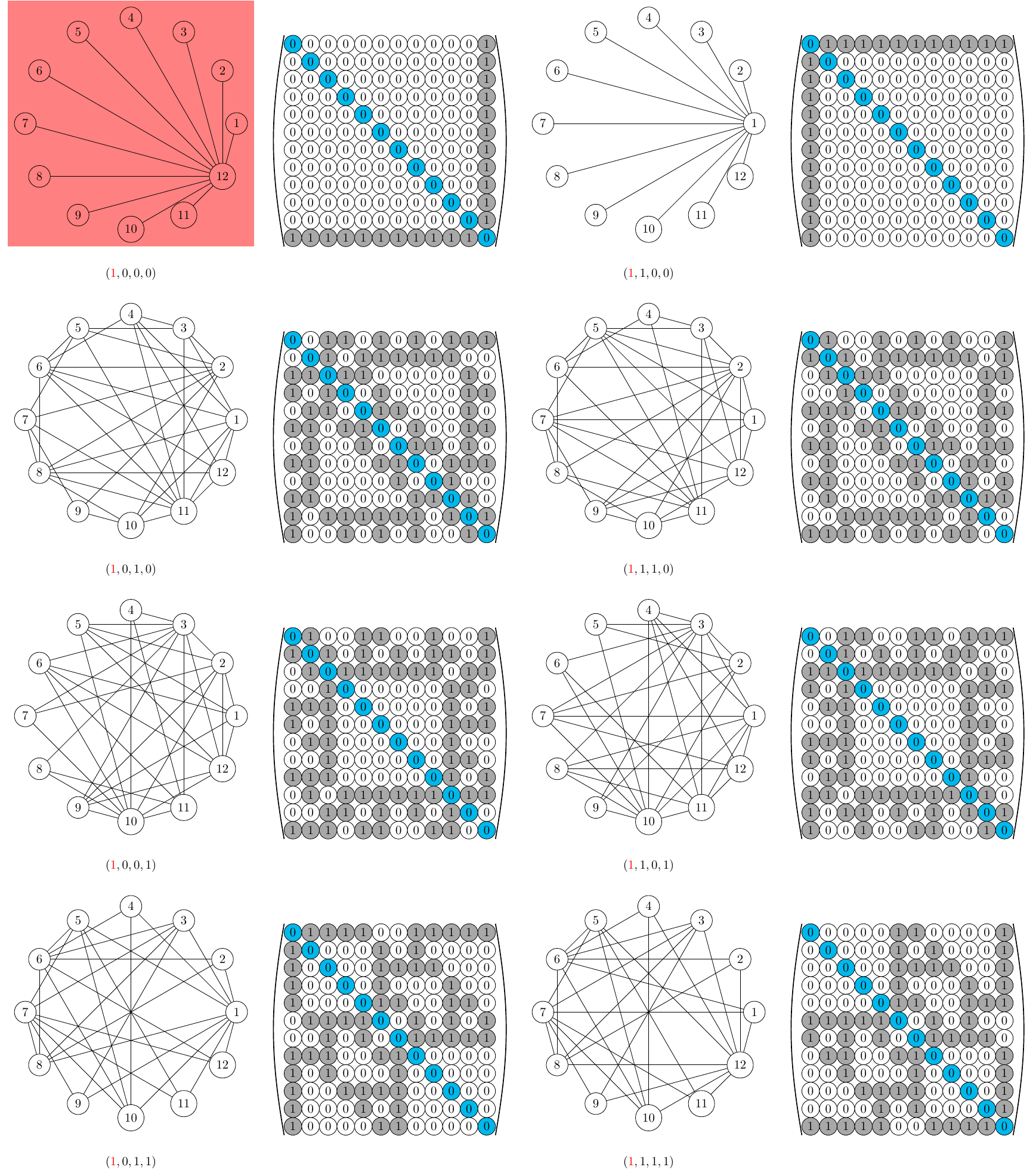}}
\caption{The $8$ graphs of $\OSG{12}$}\label{fig24}
\end{figure}

%%%%%%%%%%%%%%%%
%%%%%%%%%%%%%%%%
\section{Symmetric generalized Pascal triangles}\label{sec:7}
%%%%%%%%%%%%%%%%
%%%%%%%%%%%%%%%%

Other kinds of binary triangles with a similar definition as Steinhaus triangles can be considered. A {\em generalized Pascal triangle of size $n$} is a triangle $\Delta=(a_{i,j})_{1\le j\le i\le n}$ of $0$'s and $1$'s satisfying the local rule; i.e., $a_{i,j}\equiv\left(a_{i-1,j-1}+a_{i-1,j}\right)\tpmod{2}$ for all integers $i,j$ such that $2\le j< i\le n$. A generalized Pascal triangle $(a_{i,j})_{1\le j\le i\le n}$ is completely determined by its left side $L=(a_{i,1})_{1\le i\le n}$ and its right side $R=(a_{i,i})_{1\le i\le n}$. Therefore, we denote by $\Ptrig{L}{R}$ the generalized Pascal triangle generated from the sequences $L$ and $R$. The set of binary generalized Pascal triangles of size $n$ is denoted by $\PT{n}$. Since the set of generalized Pascal triangles is closed under addition modulo $2$, it follows that $\PT{n}$ is a vector space over $\Zn{2}$. An example of generalized Pascal triangle of size $5$ is depicted in Figure~\ref{fig12}.

\begin{figure}[htbp]
\centerline{
\begin{tikzpicture}[scale=0.25]
\figPT{{1,1,0,0,1,0,1,0,0}}{black}{col0}{col1}{white}
\end{tikzpicture}}
\caption{The generalized Pascal triangle $\Ptrig{(11100)}{(11001)}$}\label{fig12}
\end{figure}

Since a generalized Pascal triangle is uniquely determined by its left and right sides, which have the same first term, the dimension of $\PT{n}$ is $2n-1$. Moreover, there exists a natural isomorphism between $\PT{n}$ and $\ST{2n-1}$. Indeed, as depicted in Figure~\ref{fig23}, a generalized Pascal triangle of size $n$ can be seen as a subtriangle of a Steinhaus triangle of size $2n-1$.

\begin{figure}[htbp]
\centerline{
\begin{tikzpicture}[scale=0.25]
\pgfmathparse{sqrt(3)}\let\sq\pgfmathresult
\figST{{1,1,0,0,1,0,1,0,0}}{black}{col0}{col1}{white}
\draw[fill=col2!50!col1] (8,0) circle (1);
\draw (8,0) node {$1$};
\draw[fill=col2!50!col1] (7,-\sq) circle (1);
\draw (7,-\sq) node {$1$};
\draw[fill=col2!50!col1] (9,-\sq) circle (1);
\draw (9,-\sq) node {$1$};
\draw[fill=col2!50!col1] (6,-2*\sq) circle (1);
\draw (6,-2*\sq) node {$1$};
\draw[fill=col2] (8,-2*\sq) circle (1);
\draw (8,-2*\sq) node {$0$};
\draw[fill=col2] (10,-2*\sq) circle (1);
\draw (10,-2*\sq) node {$0$};
\draw[fill=col2] (5,-3*\sq) circle (1);
\draw (5,-3*\sq) node {$0$};
\draw[fill=col2!50!col1] (7,-3*\sq) circle (1);
\draw (7,-3*\sq) node {$1$};
\draw[fill=col2] (9,-3*\sq) circle (1);
\draw (9,-3*\sq) node {$0$};
\draw[fill=col2] (11,-3*\sq) circle (1);
\draw (11,-3*\sq) node {$0$};
\draw[fill=col2] (4,-4*\sq) circle (1);
\draw (4,-4*\sq) node {$0$};
\draw[fill=col2!50!col1] (6,-4*\sq) circle (1);
\draw (6,-4*\sq) node {$1$};
\draw[fill=col2!50!col1] (8,-4*\sq) circle (1);
\draw (8,-4*\sq) node {$1$};
\draw[fill=col2] (10,-4*\sq) circle (1);
\draw (10,-4*\sq) node {$0$};
\draw[fill=col2!50!col1] (12,-4*\sq) circle (1);
\draw (12,-4*\sq) node {$1$};
\end{tikzpicture}}
\caption{$\gamma\left(\Strig{(110010100)}\right)=\Ptrig{(11100)}{(11001)}$}\label{fig23}
\end{figure}

Let $\gamma$ be the linear map defined by
$$
\gamma : \ST{2n-1} \longrightarrow \PT{n} \quad\text{by}\quad \gamma\left(\left(a_{i,j}\right)_{1\le i\le j\le 2n-1}\right) = \left(a_{i,n-1+j}\right)_{1\le j\le i\le n}.
$$
The linear map $\gamma$ is well defined since the generalized Pascal triangles and the Steinhaus triangles share the same local rule.

\begin{prop}\label{prop13}
The linear map $\gamma : \ST{2n-1}\longrightarrow\PT{n}$ is an isomorphism.
\end{prop}

\begin{proof}
Let $\Strig{}=\left(a_{i,j}\right)_{1\le i\le j\le 2n-1}\in\ST{2n-1}$ and $\Delta=\gamma\left(\Strig{}\right)=\left(a_{i,n-1+j}\right)_{1\le j\le i\le n}$. The linear map $\gamma$ is an isomorphism since the set $G_{LR}$ of indices of the left and right sides of $\Delta=\left(a_{i,n-1+j}\right)_{1\le j\le i\le n}$; i.e.,
$$
G_{LR} = \left\{ (i,n)\ \middle|\ i\in\{1,\ldots,n\}\right\} \cup \left\{ (i,n-1+i)\ \middle|\ i\in\{2,\ldots,n\}\right\}
$$
is a generating index set of $\ST{2n-1}$. First, all the terms of the first row $\left(a_{1,j}\right)_{1\le j\le 2n-1}$ of $\nabla$ can be expressed as a function of the elements of the left side $\left(a_{i,n}\right)_{1\le i\le n}$ and of the right side $\left(a_{i,n-1+i}\right)_{1\le i\le n}$ of $\Delta$. Indeed, for every $j\in\{1,\ldots,n\}$, we know from Lemma~\ref{lem1} that
$$
a_{1,j} \equiv \sum_{k=0}^{n-j}\binom{n-j}{k}a_{k+1,n} \pmod{2}
$$
and
$$
a_{1,2n-j} \equiv \sum_{k=0}^{n-j}\binom{n-j}{k}a_{k+1,n+k} \pmod{2}.
$$
Since $G_1=\left\{ (1,j)\ \middle|\ j\in\left\{1,\ldots,2n-1\right\}\right\}$ is a generating index set of $\ST{2n-1}$, we conclude that $G_{LR}$ is also a generating index set of $\ST{2n-1}$. Therefore, the linear map $\gamma$ is an isomorphism.
\end{proof}

As for Steinhaus triangles, the action of the dihedral group $D_3=\left\langle r' , h' \right\rangle$ on $\PT{n}$ can be considered, where the automorphisms $r'$ and $h'$ of $\PT{n}$ are defined by
$$
\begin{array}{llll}
r' : \PT{n} \longrightarrow \PT{n} & \text{by} & r'\left((a_{i,j})_{1\le j\le i\le n}\right) = (a_{n+j-i,n+1-i})_{1\le j\le i\le n} & \text{and} \\[1.5ex]
h' : \PT{n} \longrightarrow \PT{n} & \text{by} & h'\left((a_{i,j})_{1\le j\le i\le n}\right) = (a_{i,1-j+i})_{1\le j\le i\le n}.
\end{array}
$$
For instance, for $L=(11100)$ and $R=(11001)$ and for all $g\in D_3$, the generalized Pascal triangles $g\left(\Ptrig{L}{R}\right)$ are depicted in Figure~\ref{fig16}.

\begin{figure}[htbp]
\centerline{\includegraphics[width=\textwidth]{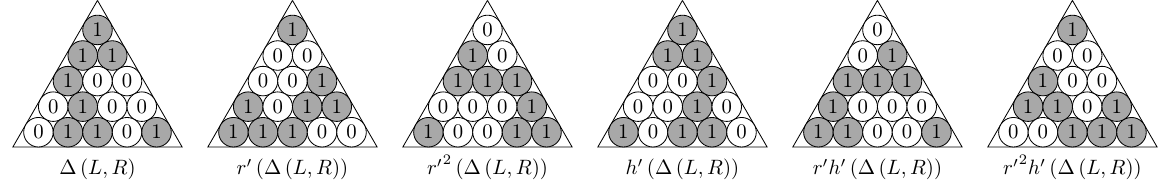}}
\caption{Action of $D_3$ on $\Ptrig{(11100)}{(11001)}$}\label{fig16}
\end{figure}

\begin{prop}\label{prop14}
For any positive integer $n$, we have
$$
\gamma r = r'\gamma \quad\text{and}\quad \gamma h = h'\gamma.
$$
\end{prop}

\begin{proof}
First, we have
$$
\gamma\left(r\left(\left(a_{i,j}\right)_{1\le i\le j\le 2n-1}\right)\right) = \gamma\left(\left(a_{j-i+1,2n-i}\right)_{1\le i\le j\le 2n-1}\right) = \left(a_{n+j-i,2n-i}\right)_{1\le j\le i\le n}
$$
and
$$
r'\left(\gamma\left(\left(a_{i,j}\right)_{1\le i\le j\le 2n-1}\right)\right) = r'\left(\left(a_{i,n-1+j}\right)_{1\le j\le i\le n}\right) = \left(a_{n+j-i,2n-i}\right)_{1\le j\le i\le n},
$$
for all $\left(a_{i,j}\right)_{1\le i\le j\le 2n-1}\in\ST{2n-1}$. Moreover,
$$
\gamma\left(h\left(\left(a_{i,j}\right)_{1\le i\le j\le 2n-1}\right)\right) = \gamma\left(\left(a_{i,2n-1-j+i}\right)_{1\le i\le j\le 2n-1}\right) = \left(a_{i,n+i-j}\right)_{1\le j\le i\le n}
$$
and
$$
h'\left(\gamma \left(\left(a_{i,j}\right)_{1\le i\le j\le 2n-1}\right)\right) = h'\left(\left(a_{i,n-1+j}\right)_{1\le j\le i\le n}\right) = \left(a_{i,n+i-j}\right)_{1\le j\le i\le n},
$$
for all $\left(a_{i,j}\right)_{1\le i\le j\le 2n-1}\in\ST{2n-1}$. This completes the proof.
\end{proof}

A generalized Pascal triangle $\Delta$ of size $n$ is said to be
\begin{itemize}
\item
{\em rotationally symmetric} if $r'(\Delta)=\Delta$,
\item
{\em horizontally symmetric} if $h'(\Delta)=\Delta$,
\item
{\em dihedrally symmetric} if $r'(\Delta)=h'(\Delta)=\Delta$.
\end{itemize}
The sets of horizontally symmetric, rotationally symmetric and dihedrally symmetric generalized Pascal triangles of size $n$ are denoted by $\HPT{n}$, $\RPT{n}$ and $\DPT{n}$, respectively. In other words, the sets $\HPT{n}$, $\RPT{n}$ and $\DPT{n}$ are simply the linear subspaces $\ker\left(h'-id_{\PT{n}}\right)$, $\ker\left(r'-id_{\PT{n}}\right)$ and $\ker\left(h'-id_{\PT{n}}\right)\cap\ker\left(r'-id_{\PT{n}}\right)$, respectively, where $id_{\PT{n}}$ is the identity map on $\PT{n}$. Examples of such triangles appear in Figure~\ref{fig13}.

\begin{figure}[htbp]
\centerline{
\begin{tabular}{c@{\quad\quad\quad}c@{\quad\quad\quad}c}
\begin{tikzpicture}[scale=0.25]
\figPT{{1,1,0,1,1,1,0,1,1}}{black}{col0}{col1}{white}
\end{tikzpicture}
&
\begin{tikzpicture}[scale=0.25]
\figPT{{0,0,1,1,1,0,1,0,0}}{black}{col0}{col1}{white}
\end{tikzpicture}
&
\begin{tikzpicture}[scale=0.25]
\figPT{{0,1,1,0,1,0,1,1,0}}{black}{col0}{col1}{white}
\end{tikzpicture}
\end{tabular}}
\caption{Triangles of $\HPT{5}$, $\RPT{5}$ and $\DPT{5}$.}\label{fig13}
\end{figure}

It is now easy to see that a symmetric generalized Pascal triangles of size $n$ corresponds to a symmetric Steinhaus triangle of size $2n-1$.

\begin{prop}\label{prop15}
For any positive integer $n$, a Steinhaus triangle $\nabla$, of size $2n-1$, is horizontally, rotationally, or dihedrally symmetric if and only if the generalized Pascal triangle $\gamma\left(\nabla\right)$, of size $n$, is horizontally, rotationally, or dihedrally symmetric, respectively.
\end{prop}

\begin{proof}
This follows directly from Propositions~\ref{prop13} and \ref{prop14}.
\end{proof}

\begin{cor}\label{cor12}
For all positive integers $n$, the linear map $\gamma$ induces isomorphisms of $\HST{2n-1}$ upon $\HPT{n}$, of $\RST{2n-1}$ upon $\RPT{n}$ and of $\DST{2n-1}$ upon $\DPT{n}$, respectively.
\end{cor}

\begin{proof}
This follows directly from Proposition~\ref{prop15}.
\end{proof}

Using the isomorphism $\gamma$ and the results of the previous sections, we obtain the dimension and a basis for each linear subspace of symmetric generalized Pascal triangles of size $n$.

\begin{prop}
For any positive integer $n$, we have
\begin{itemize}
\item
$\dim\HPT{n} = n$,
\item
$\dim\RPT{n} = 2\left\lfloor\frac{n-1}{3}\right\rfloor+1$,
\item
$\dim\DPT{n} = \left\lceil\frac{n}{3}\right\rceil$.
\end{itemize}
\end{prop}

\begin{proof}
Let $n$ be a positive integer. From Corollary~\ref{cor12} and Corollary~\ref{cor13}, we have that
$$
\dim\HPT{n} = \dim\HST{2n-1} = \left\lceil\frac{2n-1}{2}\right\rceil = n.
$$
Moreover, from Corollary~\ref{cor12} and Corollary~\ref{cor14}, we have
$$
\dim\RPT{n} = \dim\RST{2n-1}\begin{array}[t]{l}
 = \displaystyle\left\lfloor\frac{2n-1}{3}\right\rfloor + \deltamod{1}{2n-1}{3} \\ \ \\
 = \displaystyle\left\lfloor\frac{2n-1}{3}\right\rfloor+\deltamod{1}{n}{3} = 2\left\lfloor\frac{n-1}{3}\right\rfloor+1.
\end{array}
$$
Finally, from Corollary~\ref{cor12} and Corollary~\ref{cor15}, we obtain
$$
\dim\DPT{n} = \dim\DST{2n-1} \begin{array}[t]{l}
= \displaystyle\left\lfloor\frac{2n+2}{6}\right\rfloor + \deltamod{1}{2n-1}{6} \\ \ \\
= \displaystyle\left\lfloor\frac{n+1}{3}\right\rfloor+\deltamod{1}{n}{3} = \left\lceil\frac{n}{3}\right\rceil.
\end{array}
$$
This completes the proof.
\end{proof}

\begin{thm}\label{thm10}
Let $n$ and $m$ be positive integers such that $m=2\left\lfloor\frac{n-1}{3}\right\rfloor+1$. For any integers $\ell_0,\ldots,\ell_{m-1}$, the set
$$
\left\{ \gamma\rho\left(\Strig{\BS{2n-1}{k}{\ell_k}}\right)\ \middle|\ k\in\left\{0,\ldots,m-1\right\}\right\}
$$
is a basis of $\RPT{n}$, where
\begin{equation*}
\resizebox{\textwidth}{!}{
$
\displaystyle\gamma\rho\left(\Strig{\BS{2n-1}{k}{\ell_k}}\right) = \left( \left(\binom{\ell_k+j-i+n-1}{k+1-i} + \binom{\ell_k+n-j}{k+i-j-n+1} + \binom{\ell_k+i-1}{k+j-n}\right) \tpmod{2} \right)_{1\le j\le i\le n},
$}
\end{equation*}
for all $k\in\{0,\ldots,m-1\}$.
\end{thm}

\begin{proof}
This follows directly from Theorem~\ref{thm4} and Corollary~\ref{cor12}.
\end{proof}

\begin{thm}\label{thm11}
Let $n$ be a positive integer. The set
$$
\left\{ \gamma\left(\BS{2n-1}{2(n-k)-1}{-k}\right)\ \middle|\ k\in\left\{1,\ldots,n-1\right\} \right\} \cup \left\{\gamma\left(\Strig{(1)_{2n-1}}\right)\right\}
$$
is a basis of $\HPT{n}$, where
$$
\gamma\left(\Strig{\BS{2n-1}{2(n-k)-1}{-k}}\right) = \left(\binomd{-k+j-i+n-1}{2(n-k)-i}\right)_{1\le j\le i\le n},
$$
for all $k\in\left\{1,\ldots,\left\lfloor\frac{n}{2}\right\rfloor\right\}$ and $\gamma\left(\Strig{(1)_{2n-1}}\right)=\Ptrig{(1)\cdot(0)_{n-1}}{(1)\cdot(0)_{n-1}}$. 
\end{thm}

\begin{proof}
This follows directly from Theorem~\ref{thm5} and Corollary~\ref{cor12}.
\end{proof}

\begin{thm}\label{thm12}
Let $n$ and $m$ be positive integers such that $m=\left\lceil\frac{n}{3}\right\rceil$. Then, the set
$$
\left\{\gamma\left(U_{2n-1}\right)\right\}\cup\left\{ \gamma\rho\left(\BS{2n-1}{2k+1}{k-n+1}\right)\ \middle|\ k\in\left\{0,\ldots,m-2\right\}\right\}
$$
is a basis of $\DST{n}$, where $\gamma\left(U_{2n-1}\right)=\Ptrig{(1)\cdot(0)_{n-2}\cdot(1)}{(1)\cdot(0)_{n-2}\cdot(1)}$ and
\begin{equation*}
\resizebox{\textwidth}{!}{$
\displaystyle\gamma\rho\left(\BS{2n-1}{2k+1}{k-n+1}\right) = \left(\left(\binom{k+j-i}{2k-i+2}+\binom{k-j+1}{i-j+2k-n+2}+\binom{k-n+i}{2k+j-n+1}\right)\tpmod{2}\right)_{1\le j\le i\le n},
$}
\end{equation*}
for all $k\in\left\{0,\ldots,m-2\right\}$.
\end{thm}

\begin{proof}
This follows directly from Theorem~\ref{thm8} and Corollary~\ref{cor12}.
\end{proof}

We end this section by giving bases obtained from Theorem~\ref{thm10} for $\RPT{7}$, from Theorem~\ref{thm11} for $\HPT{4}$ and from Theorem~\ref{thm12} for $\DPT{11}$.

For $n=7$ and $\ell_0=\ell_1=\ell_2=\ell_3=\ell_4=0$, we obtain the following basis
$$
\left\{\gamma\rho\left(\Strig{\BS{13}{k}{0}}\right)\ \middle|\ k\in\{0,1,2,3,4\}\right\}
$$
of $\RPT{7}$. These are given in Table~\ref{tab4}.
\begin{table}[htbp]
$$
\begin{array}{|c|c|c|c|}
\hline
\vspace{-2ex} & & & \\
k & \BS{13}{k}{0} & \rho\left(\Strig{\BS{13}{k}{0}}\right) & \gamma\rho\left(\Strig{\BS{13}{k}{0}}\right) \\[1.5ex]
\hline
0 & (1111111111111) & \Strig{(0111111111110)} & \Delta_0=\Ptrig{(1000001)}{(1000001)} \\
\hline
1 & (0101010101010) & \Strig{(0001010101000)} & \Delta_1=\Ptrig{(0100010)}{(0100010)} \\
\hline
2 & (0011001100110) & \Strig{(0101001100010)} & \Delta_2=\Ptrig{(1110101)}{(1010111)} \\
\hline
3 & (0001000100010) & \Strig{(0100000101010)} & \Delta_3=\Ptrig{(0000010)}{(0100000)} \\
\hline
4 & (0000111100001) & \Strig{(1111011110001)} & \Delta_4=\Ptrig{(1011001)}{(1001101)} \\
\hline
\end{array}
$$
\caption{A basis of $\RPT{7}$}\label{tab4}
\end{table}
All the rotationally symmetric generalized Pascal triangles of size $7$ are depicted in Figure~\ref{fig17} where the elements of the basis $\left\{\Delta_0,\Delta_1,\Delta_2,\Delta_3,\Delta_4\right\}$ are in red and, for every $\Delta\in\RPT{7}$, the coordinate vector $(x_0,x_1,x_2,x_3,x_4)$ of $\Delta=x_0\Delta_0+x_1\Delta_1+x_2\Delta_2+x_3\Delta_3+x_4\Delta_4$ is given.

\begin{figure}[htbp]
\centerline{\includegraphics[width=\textwidth]{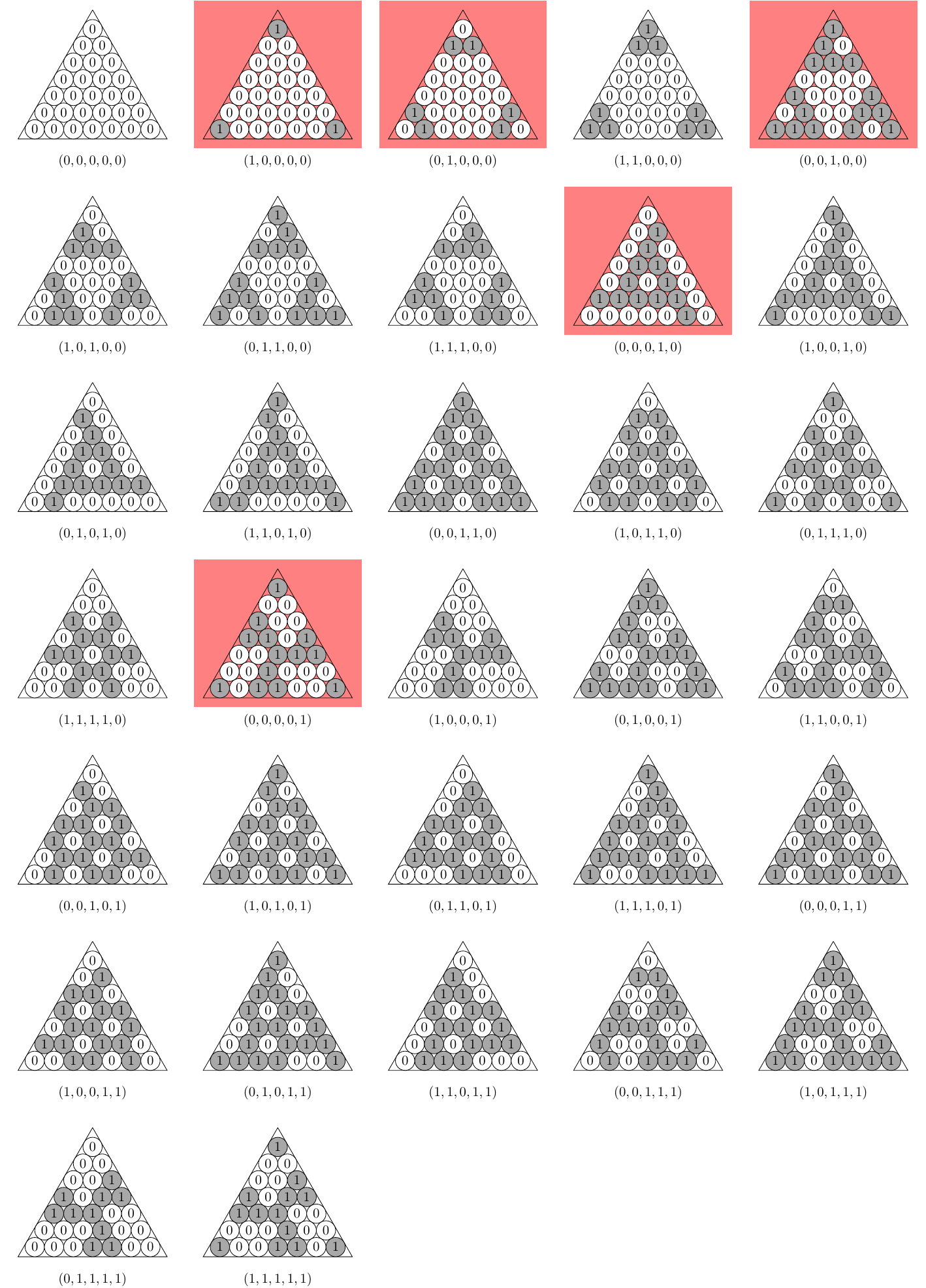}}
\caption{The $32$ triangles of $\RPT{7}$ where the $5$ red triangles form a basis}\label{fig17}
\end{figure}

For $n=4$, we obtain the following basis
$$
\left\{ \gamma\left(\BS{7}{7-2k}{-k}\right)\ \middle|\ k\in\left\{1,2,3\right\} \right\} \cup \left\{\gamma\left(\Strig{(1)_{7}}\right)\right\}
$$
of $\HPT{4}$. These are given in Table~\ref{tab5} with $\Delta_4=\gamma\left(\Strig{(1)_{7}}\right)=\Ptrig{(1000)}{(1000)}$.
\begin{table}[htbp]
$$
\begin{array}{|c|c|c|}
\hline
\vspace{-2ex} & & \\
k & \BS{7}{7-2k}{-k} & \gamma\rho\left(\Strig{\BS{7}{7-2k}{-k}}\right) \\[1.5ex]
\hline
1 & (1000001) & \Delta_1=\Ptrig{(0001)}{(0001)} \\
\hline
2 & (0100010) & \Delta_2=\Ptrig{(0011)}{(0011)} \\
\hline
3 & (1010101) & \Delta_3=\Ptrig{(0100)}{(0100)} \\
\hline
\hline
 & & \Delta_4=\Ptrig{(1000)}{(1000)} \\
\hline
\end{array}
$$
\caption{A basis of $\HPT{4}$}\label{tab5}
\end{table}
All the horizontally symmetric generalized Pascal triangles of size $4$ are depicted in Figure~\ref{fig18} where the elements of the basis $\left\{\Delta_1,\Delta_2,\Delta_3,\Delta_4\right\}$ are in red and for every $T\in\HPT{4}$, the coordinate vector $(x_1,x_2,x_3,x_4)$ of $\Delta=x_1\Delta_1+x_2\Delta_2+x_3\Delta_3+x_4\Delta_4$ is given.

\begin{figure}[htbp]
\centerline{\includegraphics[width=\textwidth]{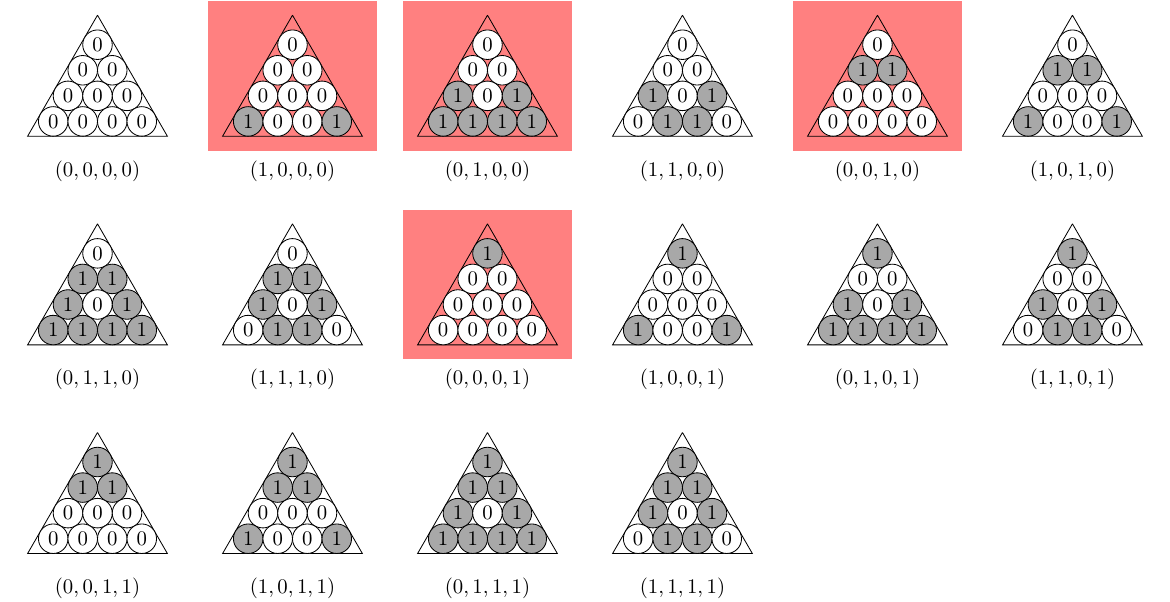}}
\caption{The $16$ triangles of $\HPT{4}$ where the $4$ red triangles form a basis}\label{fig18}
\end{figure}

For $n=11$, we obtain the following basis
$$
\left\{\gamma\left(U_{21}\right)\right\}\cup\left\{\gamma\rho\left(\Strig{\BS{21}{2k+1}{k-10}}\right)\ \middle|\ k\in\left\{0,1,2\right\}\right\}
$$
of $\DPT{11}$. These are given in Table~\ref{tab6} with $\Delta_0=\gamma\left(U_{21}\right)=\Ptrig{(10\cdots01)}{(10\cdots01)}$.
\begin{table}[htbp]
\resizebox{\textwidth}{!}{$
\begin{array}{|c|c|c|c|}
\hline
\vspace{-2ex} & & & \\
k & \BS{21}{2k+1}{k-10} & \rho\left(\Strig{\BS{21}{2k+1}{k-10}}\right) & \gamma\rho\left(\Strig{\BS{21}{2k+1}{k-10}}\right) \\[1.5ex]
\hline
0 & (010101010101010101010) & \Strig{(000101010101010101000)} & \Delta_1=\Ptrig{(01000000010)}{(01000000010)} \\
\hline
1 & (100010001000100010001) & \Strig{(011110001000100011110)} & \Delta_2=\Ptrig{(00110001100)}{(00110001100)} \\
\hline
2 & (000001010000010100000) & \Strig{(000000010000010000000)} & \Delta_3=\Ptrig{(00010001000)}{(00010001000)} \\
\hline
\hline
& & & \Delta_0=\Ptrig{(10000000001)}{(10000000001)} \\
\hline
\end{array}
$}
\caption{A basis of $\DPT{11}$}\label{tab6}
\end{table}
All the dihedrally symmetric generalized Pascal triangles of size $11$ are depicted in Figure~\ref{fig19} where the elements of the basis $\left\{\Delta_0,\Delta_1,\Delta_2,\Delta_3\right\}$ are in red and for every $\Delta\in\DPT{11}$, the coordinate vector $(x_0,x_1,x_2,x_3)$ of $\Delta=x_0\Delta_0+x_1\Delta_1+x_2\Delta_2+x_3\Delta_3$ is given.

\begin{figure}[htbp]
\centerline{\includegraphics[width=\textwidth]{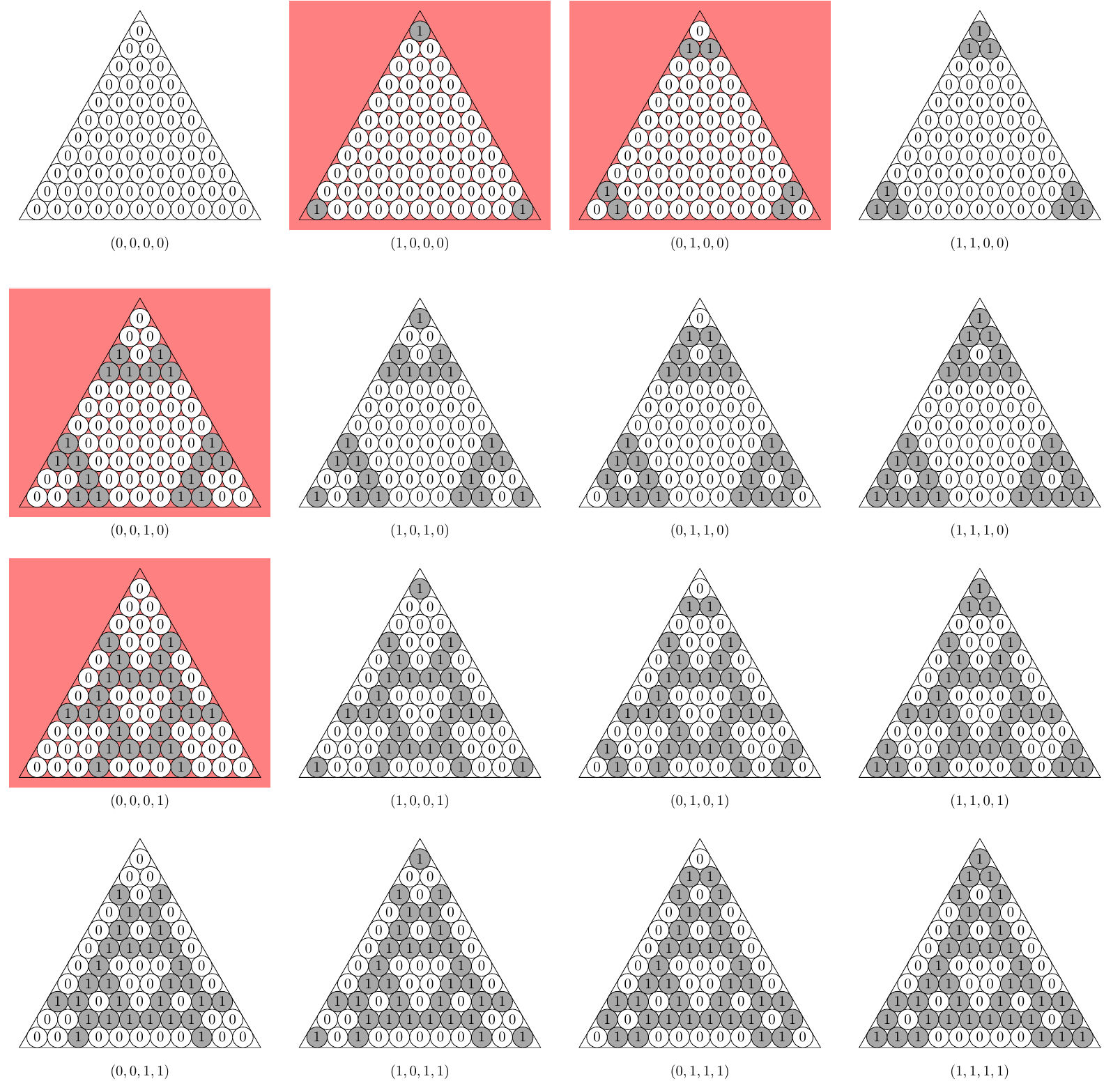}}
\caption{The $16$ triangles of $\DPT{11}$ where the $4$ red triangles form a basis}\label{fig19}
\end{figure}

%%%%%%%%%%%%%%%%%%%%%%%%%%%%%%%%
\section*{Acknowledgements}
The author would like to thank the two anonymous referees for the time spent reading this manuscript and for useful comments and remarks, which improved the presentation of the paper.

%%%%%%%%%%%%%%%%%%%%%%%%%%%%%%%%
%%%%%%%%%%%%%%%%%%%%%%%%%%%%%%%%
%\nocite{*}
\addcontentsline{toc}{section}{References}
\bibliographystyle{plain}
\bibliography{biblio}
%%%%%%%%%%%%%%%%%%%%%%%%%%%%%%%%
%%%%%%%%%%%%%%%%%%%%%%%%%%%%%%%%

%%%%%%%%%%%%%%%%%%%%%%%
%%%%%%%%%%%%%%%%%%%%%%%
\end{document}